\documentclass[twoside,11pt]{article}

\usepackage{booktabs}
\usepackage{xcolor}   
\usepackage{changepage}
\usepackage{amsfonts,amsmath,amsopn,amssymb}
\usepackage[scr=boondox,cal=esstix]{mathalpha}
\usepackage{graphicx}
\usepackage{epstopdf}
\usepackage{natbib}
\usepackage{algorithm, algpseudocode}
\usepackage{upgreek}
\usepackage{dsfont}
\usepackage{tikz}
\usepackage{tikz-cd}
\usepackage[T1]{fontenc}
\usepackage{tabularx,ragged2e}
\usetikzlibrary{arrows.meta, positioning, calc, shapes.geometric}
\usepackage{subcaption}
\usepackage{multirow}
\usepackage[preprint]{jmlr2e}
\renewenvironment{proof}[1][Proof]{\par\noindent{\bf #1\ }}{\hfill\BlackBox\\[2mm]}
\numberwithin{equation}{section}
\numberwithin{theorem}{section}
\numberwithin{example}{section}

\let\emptyset\varnothing
\newcommand{\Lip}{\operatorname{Lip}}

\newcommand{\pending}[1]{{\color{red}#1}}
\makeatletter
\DeclareRobustCommand{\res}[1]{\@ifundefined{res@#1}{\pending{[#1]}}{\@nameuse{res@#1}}}
\newcommand{\resrows}[1]{%
  \expandafter\ifx\csname resrows@#1\endcsname\relax
    \multicolumn{1}{l}{\pending{[#1]}} \\
  \else
    \csname resrows@#1\endcsname
  \fi
}
\let\orig@includegraphics\includegraphics
\renewcommand{\includegraphics}[2][]{%
  \IfFileExists{#2}{\orig@includegraphics[#1]{#2}}{\fbox{\pending{\detokenize{#2}}}}%
}
\makeatother
\IfFileExists{generated/numbers.tex}{
\expandafter\def\csname res@NumGraphs\endcsname{23}
\expandafter\def\csname res@NumGraphsAll\endcsname{29}
\expandafter\def\csname res@NumRandom\endcsname{18}
\expandafter\def\csname res@NumCanonical\endcsname{6}
\expandafter\def\csname res@NumReal\endcsname{5}
\expandafter\def\csname res@CanonicalList\endcsname{Cycle, Grid, Torus$(32,32)$, Torus$(40,40)$, Tree, and $K_{120}$}
\expandafter\def\csname res@LoTighterMin\endcsname{13.8}
\expandafter\def\csname res@LoTighterMax\endcsname{100.0}
\expandafter\def\csname res@LoTighterMinGraphs\endcsname{RReg$(2000,8)$}
\expandafter\def\csname res@LoTighterMaxGraphs\endcsname{HRG ($T=0.5$)}
\expandafter\def\csname res@GapRedMedian\endcsname{31.2}
\expandafter\def\csname res@GapRedCount\endcsname{23}
\expandafter\def\csname res@GapRedMax\endcsname{58.0}
\expandafter\def\csname res@GapRedMaxGraphs\endcsname{Power Grid}
\expandafter\def\csname res@MaxViolation\endcsname{$6.7\times10^{-16}$}
\expandafter\def\csname res@UpBetterCount\endcsname{21}
\expandafter\def\csname res@UpLooserList\endcsname{HRG ($T=0.5$) and Karate}
\expandafter\def\csname res@UpMaxExcess\endcsname{0.028}
\expandafter\def\csname res@UpHalvedCount\endcsname{15}
\expandafter\def\csname res@LazyPriorBetterCount\endcsname{17}
\expandafter\def\csname res@LazyEnvTighterMax\endcsname{73.7}
\expandafter\def\csname res@CertMoreCount\endcsname{18}
\expandafter\def\csname res@CertEqualCount\endcsname{5}
\expandafter\def\csname res@CertMajorityCount\endcsname{21}
\expandafter\def\csname res@CertTopGains\endcsname{HRG ($T=0.5$) (from 0.1\% to 35.4\%), Karate (from 51.3\% to 70.5\%), and Jazz (from 32.1\% to 50.8\%)}
\expandafter\def\csname res@SpearmanCount\endcsname{23}
\expandafter\def\csname res@SpearmanBestCount\endcsname{20}
\expandafter\def\csname res@SpearmanTopGains\endcsname{BA$(800,5)$ (from 0.64 to 0.79), BA$(1600,5)$ (from 0.66 to 0.79), and Arxiv (from 0.89 to 0.95)}
\expandafter\def\csname res@RecallOursBestCount\endcsname{17}
\expandafter\def\csname res@RecallCount\endcsname{23}
\expandafter\def\csname res@BdefMin\endcsname{50.6}
\expandafter\def\csname res@BdefMax\endcsname{100.0}
\expandafter\def\csname res@CorFirstTighterMin\endcsname{0.0}
\expandafter\def\csname res@CorFirstTighterMax\endcsname{99.6}
\expandafter\def\csname res@FigEdges:ER\endcsname{3{,}234}
\expandafter\def\csname res@FigBdef:ER\endcsname{100.0}
\expandafter\def\csname res@FigEdges:HRG\endcsname{59{,}231}
\expandafter\def\csname res@FigBdef:HRG\endcsname{63.9}
\expandafter\def\csname res@FigEdges:Yeast\endcsname{1{,}065}
\expandafter\def\csname res@FigBdef:Yeast\endcsname{92.9}
\expandafter\def\csname res@ScreenBenchSignSharpMin\endcsname{35.4}
\expandafter\def\csname res@ScreenBenchSignSharpMax\endcsname{100.0}
\expandafter\def\csname res@ScreenBenchSignBFMin\endcsname{0.2}
\expandafter\def\csname res@ScreenBenchSignBFMax\endcsname{84.1}
\expandafter\def\csname res@ScreenLargeCount\endcsname{15}
\expandafter\def\csname res@ScreenLargeMinEdges\endcsname{92{,}752}
\expandafter\def\csname res@ScreenLargeMaxEdges\endcsname{1{,}049{,}866}
\expandafter\def\csname res@ScreenSignAvoidMin\endcsname{27.4}
\expandafter\def\csname res@ScreenSignAvoidMax\endcsname{96.2}
\expandafter\def\csname res@ScreenSignAvoidMinGraph\endcsname{Squirrel}
\expandafter\def\csname res@ScreenSignAvoidMaxGraph\endcsname{GitHub}
\expandafter\def\csname res@ScreenSignSpeedupMax\endcsname{15.3}
\expandafter\def\csname res@ScreenSignSpeedupMaxGraph\endcsname{GitHub}
\expandafter\def\csname res@ScreenSignSpeedupMedian\endcsname{2.0}
\expandafter\def\csname res@ScreenSignWorkRatioMax\endcsname{26.2}
\expandafter\def\csname res@ScreenSignWorkRatioMedian\endcsname{5.9}
\expandafter\def\csname res@ScreenSignShareBFMax\endcsname{62.9}
\expandafter\def\csname res@ScreenSignShareTLWMin\endcsname{0.2}
\expandafter\def\csname res@ScreenSignShareTLWMax\endcsname{14.3}
\expandafter\def\csname res@ScreenStageOneShareMax\endcsname{154.6}
\expandafter\def\csname res@ScreenThreads\endcsname{32}
\expandafter\def\csname res@BenchBothLastN:BA\endcsname{200000}
\expandafter\def\csname res@BenchExactAtMax:BA\endcsname{66366}
\expandafter\def\csname res@BenchBFAtMax:BA\endcsname{735}
\expandafter\def\csname res@BenchRatioBFMax:BA\endcsname{90}
\expandafter\def\csname res@BenchMinN\endcsname{50}
\expandafter\def\csname res@BenchMaxN\endcsname{200000}
\expandafter\def\csname res@BenchNumGraphs\endcsname{100}
\expandafter\def\csname res@BenchMaxEdges\endcsname{27{,}599{,}625}
\expandafter\def\csname res@BenchThreadsText\endcsname{32 threads}
\expandafter\def\csname res@BenchBothLastN\endcsname{200000}
\expandafter\def\csname res@BenchExactAtMax\endcsname{9557}
\expandafter\def\csname res@BenchSharpAtMax\endcsname{586}
\expandafter\def\csname res@BenchRatioSharpRange\endcsname{6.4--22}
\expandafter\def\csname res@BenchRatioBFRange\endcsname{5.5--36}
\expandafter\def\csname res@BenchRatioFromN\endcsname{1020}
\expandafter\def\csname res@BenchChunkSentence\endcsname{Graphs with more than $2{,}000{,}000$ edges (from $n=31664$) are timed as the sum of the times of consecutive chunks of $250{,}000$ edges, each run once, so that a long run can be resumed after an interruption; smaller graphs are timed by the median of repeated runs on all edges.}
\expandafter\def\csname resrows@bound-comparison\endcsname{%
\multicolumn{10}{l}{\emph{Random graph models}} \\
BA(800, 2) & 1597 & 0.167 & 0.168 & \textbf{0.114} & 29.2 & 0.807 & 0.807 & \textbf{0.371} & 96.8 / 3.1 \\
BA(800, 5) & 3985 & 0.812 & 0.812 & \textbf{0.552} & 91.8 & 0.705 & 0.705 & \textbf{0.195} & 99.9 / 0.0 \\
BA(1600, 2) & 3197 & 0.112 & 0.113 & \textbf{0.080} & 22.0 & 0.872 & 0.872 & \textbf{0.412} & 97.8 / 1.9 \\
BA(1600, 5) & 7985 & 0.678 & 0.678 & \textbf{0.491} & 83.9 & 0.873 & 0.873 & \textbf{0.271} & 100.0 / 0.0 \\
ER(800, 0.0100125) & 3234 & 0.386 & 0.386 & \textbf{0.277} & 44.7 & 1.104 & 1.104 & \textbf{0.424} & 100.0 / 0.0 \\
ER(1600, 0.0050031) & 6420 & 0.235 & 0.235 & \textbf{0.179} & 26.5 & 1.259 & 1.259 & \textbf{0.546} & 99.9 / 0.1 \\
RGG(800, 0.056419) & 2985 & 0.196 & 0.211 & \textbf{0.149} & 20.6 & 0.267 & 0.267 & \textbf{0.137} & 76.7 / 19.6 \\
RGG(1600, 0.039894) & 6118 & 0.192 & 0.207 & \textbf{0.147} & 19.5 & 0.273 & 0.273 & \textbf{0.136} & 78.0 / 18.2 \\
HRG(800, 5.0, 1.0, 0.0) & 59231 & 0.475 & 0.476 & \textbf{0.437} & 62.5 & 0.172 & 0.172 & \textbf{0.087} & 55.4 / 44.6 \\
HRG(800, 5.0, 1.0, 0.5) & 77050 & 1.101 & 1.101 & \textbf{0.504} & 100.0 & \textbf{0.000} & \textbf{0.000} & 0.004 & 0.0 / 100.0 \\
WS(800, 10, 0.05) & 4000 & 0.235 & 0.236 & \textbf{0.162} & 32.8 & 0.354 & 0.354 & \textbf{0.132} & 89.4 / 4.6 \\
WS(800, 10, 0.2) & 4000 & 0.354 & 0.354 & \textbf{0.162} & 66.3 & 0.642 & 0.642 & \textbf{0.299} & 96.7 / 2.6 \\
WS(1600, 10, 0.05) & 8000 & 0.236 & 0.237 & \textbf{0.162} & 33.5 & 0.352 & 0.352 & \textbf{0.133} & 89.5 / 4.9 \\
WS(1600, 10, 0.2) & 8000 & 0.306 & 0.307 & \textbf{0.130} & 61.0 & 0.679 & 0.679 & \textbf{0.331} & 97.3 / 2.2 \\
SBM(2x500, 0.012, 0.004) & 4041 & 0.333 & 0.333 & \textbf{0.237} & 40.9 & 1.157 & 1.157 & \textbf{0.470} & 100.0 / 0.0 \\
SBM(2x500, 0.004, 0.012) & 4056 & 0.325 & 0.325 & \textbf{0.234} & 38.9 & 1.166 & 1.166 & \textbf{0.475} & 100.0 / 0.0 \\
RReg(1000, 8) & 4000 & 0.266 & 0.266 & \textbf{0.197} & 25.1 & 1.227 & 1.227 & \textbf{0.515} & 100.0 / 0.0 \\
RReg(2000, 8) & 8000 & 0.146 & 0.146 & \textbf{0.110} & 13.8 & 1.347 & 1.347 & \textbf{0.619} & 100.0 / 0.0 \\
\midrule
\multicolumn{10}{l}{\emph{Canonical graphs}} \\
Cycle(600) & 600 & \textbf{0.000} & \textbf{0.000} & \textbf{0.000} & 0.0 & \textbf{0.000} & \textbf{0.000} & \textbf{0.000} & 0.0 / 0.0 \\
Grid(40,40) & 3120 & 0.974 & 0.974 & \textbf{0.000} & 100.0 & \textbf{0.000} & \textbf{0.000} & 0.008 & 0.0 / 5.1 \\
Torus(32,32) & 2048 & 1.000 & 1.000 & \textbf{0.000} & 100.0 & \textbf{0.000} & \textbf{0.000} & \textbf{0.000} & 0.0 / 0.0 \\
Torus(40,40) & 3200 & 1.000 & 1.000 & \textbf{0.000} & 100.0 & \textbf{0.000} & \textbf{0.000} & \textbf{0.000} & 0.0 / 0.0 \\
Tree(4,6) & 5460 & \textbf{0.000} & \textbf{0.000} & \textbf{0.000} & 0.0 & \textbf{0.300} & \textbf{0.300} & \textbf{0.300} & 25.0 / 75.0 \\
Complete(120) & 7140 & \textbf{0.000} & \textbf{0.000} & \textbf{0.000} & 0.0 & \textbf{0.000} & \textbf{0.000} & 0.008 & 0.0 / 100.0 \\
\midrule
\multicolumn{10}{l}{\emph{Real networks}} \\
Arxiv & 420784 & 0.897 & 0.897 & \textbf{0.562} & 95.5 & 0.519 & 0.519 & \textbf{0.262} & 85.6 / 14.4 \\
Jazz & 2742 & 0.598 & 0.601 & \textbf{0.379} & 86.6 & 0.151 & 0.151 & \textbf{0.105} & 53.8 / 46.1 \\
Karate & 78 & 0.401 & 0.436 & \textbf{0.248} & 50.0 & \textbf{0.166} & \textbf{0.166} & 0.194 & 32.1 / 67.9 \\
Power Grid & 6594 & 0.099 & 0.110 & \textbf{0.042} & 15.3 & 0.393 & 0.393 & \textbf{0.267} & 61.9 / 29.0 \\
Yeast & 1065 & 0.357 & 0.360 & \textbf{0.241} & 54.6 & 0.382 & 0.382 & \textbf{0.267} & 43.0 / 56.1 \\
}
\expandafter\def\csname resrows@bound-comparison-lazy\endcsname{%
\multicolumn{9}{l}{\emph{Random graph models}} \\
BA(800, 2) & 0.158 & \textbf{0.108} & 33.5 & \textbf{0.171} & 0.343 & 14.7 & 93.1 & \textbf{95.9} \\
BA(800, 5) & 0.752 & \textbf{0.507} & 93.6 & 0.293 & \textbf{0.160} & 73.7 & 98.4 & \textbf{99.8} \\
BA(1600, 2) & 0.106 & \textbf{0.075} & 24.9 & \textbf{0.141} & 0.384 & 9.8 & 96.0 & \textbf{97.4} \\
BA(1600, 5) & 0.629 & \textbf{0.453} & 86.0 & 0.326 & \textbf{0.237} & 62.8 & 99.6 & \textbf{99.9} \\
ER(800, 0.0100125) & 0.350 & \textbf{0.250} & 48.3 & \textbf{0.126} & 0.441 & 3.1 & \textbf{100.0} & \textbf{100.0} \\
ER(1600, 0.0050031) & 0.214 & \textbf{0.162} & 29.0 & \textbf{0.078} & 0.554 & 0.4 & \textbf{100.0} & \textbf{100.0} \\
RGG(800, 0.056419) & 0.217 & \textbf{0.142} & 87.8 & \textbf{0.052} & 0.167 & 6.5 & 68.5 & \textbf{73.8} \\
RGG(1600, 0.039894) & 0.213 & \textbf{0.143} & 86.8 & \textbf{0.048} & 0.167 & 6.1 & 70.9 & \textbf{76.0} \\
HRG(800, 5.0, 1.0, 0.0) & 0.476 & \textbf{0.436} & 100.0 & \textbf{0.067} & 0.076 & 32.0 & 43.5 & \textbf{48.8} \\
HRG(800, 5.0, 1.0, 0.5) & 1.101 & \textbf{0.507} & 100.0 & \textbf{0.000} & \textbf{0.000} & 0.0 & 0.1 & \textbf{35.4} \\
WS(800, 10, 0.05) & 0.221 & \textbf{0.149} & 62.6 & \textbf{0.060} & 0.184 & 11.9 & 62.6 & \textbf{70.0} \\
WS(800, 10, 0.2) & 0.326 & \textbf{0.146} & 81.0 & \textbf{0.122} & 0.320 & 14.9 & 62.9 & \textbf{75.0} \\
WS(1600, 10, 0.05) & 0.222 & \textbf{0.149} & 64.4 & \textbf{0.059} & 0.183 & 12.6 & 62.1 & \textbf{69.2} \\
WS(1600, 10, 0.2) & 0.282 & \textbf{0.118} & 77.7 & \textbf{0.110} & 0.350 & 12.2 & 66.6 & \textbf{77.3} \\
SBM(2x500, 0.012, 0.004) & 0.302 & \textbf{0.214} & 44.9 & \textbf{0.113} & 0.483 & 1.5 & \textbf{100.0} & \textbf{100.0} \\
SBM(2x500, 0.004, 0.012) & 0.295 & \textbf{0.211} & 43.0 & \textbf{0.107} & 0.489 & 1.0 & 99.9 & \textbf{100.0} \\
RReg(1000, 8) & 0.237 & \textbf{0.175} & 25.1 & \textbf{0.003} & 0.569 & 0.0 & \textbf{100.0} & \textbf{100.0} \\
RReg(2000, 8) & 0.130 & \textbf{0.098} & 13.8 & \textbf{0.001} & 0.661 & 0.0 & \textbf{100.0} & \textbf{100.0} \\
\midrule
\multicolumn{9}{l}{\emph{Canonical graphs}} \\
Cycle(600) & \textbf{0.000} & \textbf{0.000} & 0.0 & \textbf{0.000} & 0.333 & 0.0 & \textbf{0.0} & \textbf{0.0} \\
Grid(40,40) & 0.785 & \textbf{0.000} & 100.0 & \textbf{0.010} & 0.197 & 5.1 & 0.0 & \textbf{5.1} \\
Torus(32,32) & 0.800 & \textbf{0.000} & 100.0 & \textbf{0.000} & 0.200 & 0.0 & \textbf{0.0} & \textbf{0.0} \\
Torus(40,40) & 0.800 & \textbf{0.000} & 100.0 & \textbf{0.000} & 0.200 & 0.0 & \textbf{0.0} & \textbf{0.0} \\
Tree(4,6) & 0.150 & \textbf{0.000} & 75.0 & \textbf{0.000} & 0.166 & 0.0 & \textbf{25.0} & \textbf{25.0} \\
Complete(120) & 0.017 & \textbf{0.000} & 100.0 & \textbf{0.000} & \textbf{0.000} & 0.0 & \textbf{100.0} & \textbf{100.0} \\
\midrule
\multicolumn{9}{l}{\emph{Real networks}} \\
Arxiv & 0.874 & \textbf{0.542} & 99.0 & 0.269 & \textbf{0.189} & 64.5 & 62.6 & \textbf{76.1} \\
Jazz & 0.599 & \textbf{0.377} & 99.8 & 0.070 & \textbf{0.066} & 32.2 & 32.1 & \textbf{50.8} \\
Karate & 0.438 & \textbf{0.232} & 100.0 & 0.110 & \textbf{0.106} & 24.4 & 51.3 & \textbf{70.5} \\
Power Grid & 0.163 & \textbf{0.042} & 50.3 & \textbf{0.041} & 0.261 & 5.9 & 65.5 & \textbf{82.4} \\
Yeast & 0.329 & \textbf{0.192} & 91.1 & \textbf{0.144} & 0.177 & 20.7 & 85.3 & \textbf{91.5} \\
}
\expandafter\def\csname resrows@ranking\endcsname{%
\multicolumn{9}{l}{\emph{Random graph models}} \\
BA(800, 2) & 0.95 & 0.92 & 0.97 & \textbf{0.98} & \textbf{0.81} & 0.77 & 0.72 & 0.77 \\
BA(800, 5) & 0.51 & 0.46 & 0.64 & \textbf{0.79} & 0.20 & 0.19 & 0.39 & \textbf{0.49} \\
BA(1600, 2) & 0.96 & 0.94 & 0.97 & \textbf{0.98} & \textbf{0.84} & 0.81 & 0.80 & 0.81 \\
BA(1600, 5) & 0.57 & 0.49 & 0.66 & \textbf{0.79} & 0.18 & 0.17 & 0.38 & \textbf{0.49} \\
ER(800, 0.0100125) & 0.70 & 0.41 & 0.84 & \textbf{0.88} & 0.36 & 0.16 & 0.67 & \textbf{0.69} \\
ER(1600, 0.0050031) & 0.74 & 0.54 & 0.90 & \textbf{0.92} & 0.36 & 0.23 & 0.72 & \textbf{0.75} \\
RGG(800, 0.056419) & 0.97 & 0.95 & \textbf{0.99} & \textbf{0.99} & 0.79 & 0.76 & \textbf{0.89} & 0.87 \\
RGG(1600, 0.039894) & 0.97 & 0.95 & \textbf{0.99} & \textbf{0.99} & 0.81 & 0.76 & \textbf{0.88} & \textbf{0.88} \\
HRG(800, 5.0, 1.0, 0.0) & 0.98 & 0.98 & \textbf{0.99} & \textbf{0.99} & 0.77 & \textbf{0.83} & 0.71 & 0.73 \\
HRG(800, 5.0, 1.0, 0.5) & 0.88 & 0.93 & \textbf{0.94} & 0.91 & 0.61 & 0.65 & \textbf{0.66} & 0.65 \\
WS(800, 10, 0.05) & 0.98 & 0.97 & \textbf{0.99} & \textbf{0.99} & 0.85 & 0.84 & \textbf{0.88} & \textbf{0.88} \\
WS(800, 10, 0.2) & 0.96 & 0.94 & \textbf{0.98} & \textbf{0.98} & 0.71 & 0.57 & \textbf{0.87} & 0.86 \\
WS(1600, 10, 0.05) & 0.98 & 0.97 & \textbf{0.99} & \textbf{0.99} & 0.87 & 0.86 & 0.87 & \textbf{0.89} \\
WS(1600, 10, 0.2) & 0.97 & 0.95 & 0.98 & \textbf{0.99} & 0.67 & 0.61 & \textbf{0.90} & \textbf{0.90} \\
SBM(2x500, 0.012, 0.004) & 0.73 & 0.46 & 0.86 & \textbf{0.90} & 0.35 & 0.18 & 0.65 & \textbf{0.69} \\
SBM(2x500, 0.004, 0.012) & 0.71 & 0.44 & 0.87 & \textbf{0.90} & 0.35 & 0.16 & 0.68 & \textbf{0.71} \\
RReg(1000, 8) & 0.64 & 0.25 & \textbf{0.99} & 0.96 & 0.13 & 0.10 & \textbf{0.88} & \textbf{0.88} \\
RReg(2000, 8) & 0.58 & 0.26 & \textbf{1.00} & 0.99 & 0.11 & 0.10 & \textbf{0.33} & \textbf{0.33} \\
\midrule
\multicolumn{9}{l}{\emph{Canonical graphs}} \\
Cycle(600) & -- & -- & -- & -- & \textbf{0.02} & \textbf{0.02} & \textbf{0.02} & \textbf{0.02} \\
Grid(40,40) & 0.74 & 0.67 & 0.74 & \textbf{1.00} & \textbf{0.11} & 0.10 & 0.10 & \textbf{0.11} \\
Torus(32,32) & -- & -- & -- & -- & \textbf{0.11} & \textbf{0.11} & \textbf{0.11} & \textbf{0.11} \\
Torus(40,40) & -- & -- & -- & -- & \textbf{0.10} & \textbf{0.10} & \textbf{0.10} & \textbf{0.10} \\
Tree(4,6) & \textbf{1.00} & \textbf{1.00} & \textbf{1.00} & \textbf{1.00} & \textbf{0.40} & \textbf{0.40} & \textbf{0.40} & \textbf{0.40} \\
Complete(120) & -- & -- & -- & -- & \textbf{0.10} & \textbf{0.10} & \textbf{0.10} & \textbf{0.10} \\
\midrule
\multicolumn{9}{l}{\emph{Real networks}} \\
Arxiv & 0.78 & 0.80 & 0.89 & \textbf{0.95} & 0.45 & 0.46 & 0.67 & \textbf{0.77} \\
Jazz & 0.93 & 0.94 & 0.96 & \textbf{0.98} & 0.63 & 0.63 & 0.75 & \textbf{0.84} \\
Karate & 0.85 & 0.77 & 0.89 & \textbf{0.93} & \textbf{0.62} & \textbf{0.62} & 0.50 & \textbf{0.62} \\
Power Grid & 0.94 & 0.88 & 0.98 & \textbf{0.99} & 0.85 & 0.69 & \textbf{0.91} & \textbf{0.91} \\
Yeast & 0.66 & 0.62 & 0.82 & \textbf{0.87} & 0.54 & 0.50 & 0.75 & \textbf{0.82} \\
}
\expandafter\def\csname resrows@bf-interval\endcsname{%
\multicolumn{8}{l}{\emph{Random graph models}} \\
BA(800, 2) & 0.113 & \textbf{0.108} & 9.1 & 77.8 & \textbf{95.9} & 0.89 & \textbf{0.98} \\
BA(800, 5) & 0.624 & \textbf{0.507} & 75.3 & 36.1 & \textbf{99.8} & 0.53 & \textbf{0.79} \\
BA(1600, 2) & 0.077 & \textbf{0.075} & 6.9 & 81.5 & \textbf{97.4} & 0.91 & \textbf{0.98} \\
BA(1600, 5) & 0.524 & \textbf{0.453} & 61.7 & 46.0 & \textbf{99.9} & 0.53 & \textbf{0.79} \\
ER(800, 0.0100125) & 0.254 & \textbf{0.250} & 12.2 & 52.0 & \textbf{100.0} & 0.66 & \textbf{0.88} \\
ER(1600, 0.0050031) & 0.164 & \textbf{0.162} & 4.7 & 71.1 & \textbf{100.0} & 0.67 & \textbf{0.92} \\
RGG(800, 0.056419) & 0.148 & \textbf{0.142} & 8.7 & 52.8 & \textbf{73.8} & 0.96 & \textbf{0.99} \\
RGG(1600, 0.039894) & 0.149 & \textbf{0.143} & 7.2 & 53.4 & \textbf{76.0} & 0.97 & \textbf{0.99} \\
HRG(800, 5.0, 1.0, 0.0) & 0.468 & \textbf{0.436} & 61.2 & 35.9 & \textbf{48.8} & 0.98 & \textbf{0.99} \\
HRG(800, 5.0, 1.0, 0.5) & 1.092 & \textbf{0.507} & 100.0 & 0.2 & \textbf{35.4} & \textbf{0.94} & 0.91 \\
WS(800, 10, 0.05) & 0.168 & \textbf{0.149} & 21.7 & 53.8 & \textbf{70.0} & 0.98 & \textbf{0.99} \\
WS(800, 10, 0.2) & 0.210 & \textbf{0.146} & 46.9 & 30.7 & \textbf{75.0} & 0.96 & \textbf{0.98} \\
WS(1600, 10, 0.05) & 0.168 & \textbf{0.149} & 21.8 & 53.7 & \textbf{69.2} & 0.98 & \textbf{0.99} \\
WS(1600, 10, 0.2) & 0.178 & \textbf{0.118} & 44.4 & 33.1 & \textbf{77.3} & 0.97 & \textbf{0.99} \\
SBM(2x500, 0.012, 0.004) & 0.218 & \textbf{0.214} & 11.1 & 55.2 & \textbf{100.0} & 0.66 & \textbf{0.90} \\
SBM(2x500, 0.004, 0.012) & 0.215 & \textbf{0.211} & 10.7 & 57.0 & \textbf{100.0} & 0.65 & \textbf{0.90} \\
RReg(1000, 8) & \textbf{0.175} & \textbf{0.175} & 0.6 & 72.5 & \textbf{100.0} & 0.64 & \textbf{0.96} \\
RReg(2000, 8) & \textbf{0.098} & \textbf{0.098} & 0.1 & 84.1 & \textbf{100.0} & 0.58 & \textbf{0.99} \\
\midrule
\multicolumn{8}{l}{\emph{Canonical graphs}} \\
Cycle(600) & \textbf{0.000} & \textbf{0.000} & 0.0 & \textbf{0.0} & \textbf{0.0} & -- & -- \\
Grid(40,40) & 0.001 & \textbf{0.000} & 4.9 & \textbf{5.1} & \textbf{5.1} & 0.74 & \textbf{1.00} \\
Torus(32,32) & \textbf{0.000} & \textbf{0.000} & 0.0 & \textbf{0.0} & \textbf{0.0} & -- & -- \\
Torus(40,40) & \textbf{0.000} & \textbf{0.000} & 0.0 & \textbf{0.0} & \textbf{0.0} & -- & -- \\
Tree(4,6) & \textbf{0.000} & \textbf{0.000} & 0.0 & \textbf{25.0} & \textbf{25.0} & \textbf{1.00} & \textbf{1.00} \\
Complete(120) & \textbf{0.000} & \textbf{0.000} & 0.0 & \textbf{100.0} & \textbf{100.0} & -- & -- \\
\midrule
\multicolumn{8}{l}{\emph{Real networks}} \\
Arxiv & 0.823 & \textbf{0.542} & 93.4 & 14.1 & \textbf{76.1} & 0.80 & \textbf{0.95} \\
Jazz & 0.546 & \textbf{0.377} & 83.1 & 29.0 & \textbf{50.8} & 0.94 & \textbf{0.98} \\
Karate & 0.280 & \textbf{0.232} & 39.7 & 52.6 & \textbf{70.5} & 0.88 & \textbf{0.93} \\
Power Grid & 0.044 & \textbf{0.042} & 2.7 & 53.2 & \textbf{82.4} & 0.95 & \textbf{0.99} \\
Yeast & 0.236 & \textbf{0.192} & 47.2 & 74.8 & \textbf{91.5} & 0.73 & \textbf{0.87} \\
}
\expandafter\def\csname resrows@screening\endcsname{%
Deezer & 92{,}752 & 172 & 0.056 & 43.3 & 94.4 & 0.9 \\
Facebook & 170{,}823 & 709 & 1.3 & 25.6 & 61.2 & 1.4 \\
Enron & 180{,}811 & 1{,}383 & 4.4 & 40.3 & 83.8 & 6.2 \\
AstroPh & 196{,}972 & 504 & 1.5 & 27.6 & 83.0 & 4.2 \\
Brightkite & 212{,}945 & 1{,}134 & 1.2 & 44.7 & 87.3 & 2.6 \\
Amazon-Photo & 119{,}043 & 1{,}434 & 2.0 & 14.5 & 47.5 & 1.7 \\
Twitch-DE & 153{,}138 & 4{,}259 & 7.8 & 33.6 & 82.1 & 2.4 \\
Crocodile & 170{,}773 & 3{,}546 & 2.0 & 62.9 & 76.7 & 2.0 \\
Squirrel & 198{,}353 & 1{,}903 & 22.7 & 17.1 & 27.4 & 1.0 \\
Amazon-Comp. & 245{,}778 & 2{,}992 & 8.0 & 13.8 & 52.8 & 1.4 \\
Coauthor-Phys. & 247{,}962 & 382 & 0.49 & 14.1 & 88.5 & 2.1 \\
GitHub & 289{,}003 & 9{,}458 & 13.9 & 51.8 & 96.2 & 15.3 \\
Epinions & 405{,}739 & 3{,}044 & 28.3 & 40.3 & 89.1 & 3.9 \\
Amazon & 925{,}872 & 549 & 0.43 & 47.5 & 77.2 & 0.5 \\
DBLP & 1{,}049{,}866 & 343 & 0.85 & 52.9 & 92.5 & 1.0 \\
}
\expandafter\def\csname resrows@runtime\endcsname{%
ER & 200000 & 27{,}599{,}625 & 6.4--21.8 & 5.5--36.3 & 263 & 586 & 9557 \\
BA & 200000 & 27{,}590{,}409 & 7.3--21.7 & 9.3--90.3 & 735 & 4274 & 66366 \\
RGG & 200000 & 27{,}103{,}580 & 5.4--29.2 & 3.9--42.8 & 71.0 & 104 & 3041 \\
WS & 200000 & 27{,}600{,}000 & 7.2--16.0 & 5.4--26.1 & 154 & 258 & 4006 \\
}
\expandafter\def\csname resrows@solver\endcsname{%
Yeast & 1{,}065 & 1{,}065 & 4.4 & 3.8 & 29 & -- \\
Jazz & 2{,}742 & 2{,}742 & 46 & 1.9 & 5.8 & -- \\
Power Grid & 6{,}594 & 3{,}000 & 2.9 & 3.3 & 40 & -- \\
Arxiv & 420{,}784 & 3{,}000 & 173 & 2.1 & 3.1 & -- \\
ER$(2000, \bar d=10)$ & 10{,}059 & 3{,}000 & 6.5 & 3.3 & 21 & -- \\
BA$(5000, 5)$ & 24{,}975 & 3{,}000 & 20 & 3.3 & 9.3 & -- \\
RGG$(3000, 0.05)$ & 33{,}595 & 3{,}000 & 23 & 1.6 & 8.1 & -- \\
ER$(600, \bar d=90)$ & 27{,}067 & 3{,}000 & 116 & 4.7 & 6.6 & -- \\
}
\expandafter\def\csname res@MapEdges\endcsname{2,985}
\expandafter\def\csname res@MapBF\endcsname{52.8}
\expandafter\def\csname res@MapBFPos\endcsname{98.6}
\expandafter\def\csname res@MapTwoHop\endcsname{63.4}
\expandafter\def\csname res@MapTwoHopNeg\endcsname{94.6}
\expandafter\def\csname res@MapTLWGain\endcsname{10.5}
\expandafter\def\csname res@MapSharp\endcsname{73.8}
\expandafter\def\csname res@MapSolved\endcsname{781}
\expandafter\def\csname res@MapSolvedPct\endcsname{26.2}
\expandafter\def\csname res@MapMedGapSolved\endcsname{0.10}
\expandafter\def\csname res@MapMedGapDecided\endcsname{0.43}
\expandafter\def\csname res@MapWorkRatio\endcsname{3.8}}{}

\definecolor{Accent}{HTML}{1F77B4}
\tikzset{
  v/.style      ={circle, fill=black, inner sep=1.0pt},             
  ring/.style   ={circle, draw=Accent, line width=0.7pt, fill=white, inner sep=1.0pt}, 
  E/.style      ={line width=1.0pt, draw=Accent},                    
  thin/.style   ={line width=0.45pt, draw=black!70},                 
  cross/.style  ={line width=0.6pt, dashed, dash pattern=on 2pt off 1.2pt}, 
  close/.style  ={line width=0.6pt, densely dotted},                 
  tag/.style    ={font=\scriptsize\itshape, fill=white, inner sep=1pt}
}

\newcommand{\BaseUV}{%
  \node[v,label=below:$i$] (u) at (0,0) {};
  \node[v,label=below:$j$] (v) at (2.9,0) {};
  \draw[E] (u)--(v);
}

\usepackage{lastpage}
\usepackage{adjustbox}
\jmlrheading{XX}{2026}{1-\pageref{LastPage}}{XX/XX; Revised XX/XX}{XX/XX}{XX-XXXX}{Giorgio Micaletto and Tebe Nigrelli%
}

\ShortHeadings{From Proxy to Certificate: BF-to-OR Bounds}{Micaletto and Nigrelli
}
\firstpageno{1}

\begin{document}

\title{From Proxy to Certificate: Two-Sided Ollivier--Ricci Bounds from Balanced Forman Curvature}

\begingroup
\renewcommand\thefootnote{\fnsymbol{footnote}}
\setcounter{footnote}{0}

\author{\name Giorgio Micaletto
        \email gmicaletto@uchicago.edu\\
        \addr Department of Statistics\\
        University of Chicago\\
        Chicago, IL 60637-1514, USA
        \AND 
        \name Tebe Nigrelli
        \email tebe.nigrelli@studbocconi.it\\
        \addr Department of Computing Sciences\\ 
        Bocconi University\\
        Via Roberto Sarfatti 25, 20136 Milan, Italy
}
\begingroup
\makeatletter
\renewcommand\@makefntext[1]{%
\noindent\makebox[1.8em][r]{\!\@makefnmark:}\,#1%
}
\makeatother
\endgroup
\endgroup

\editor{
}

\maketitle

\begin{abstract}%
Balanced Forman (BF) curvature is used in graph learning as a cheap stand-in for Ollivier--Ricci (OR) curvature, yet it does not lower-bound OR curvature. We show that the BF value of an edge certifies a two-sided interval for its lazy OR curvature without solving any optimal transport problem, that BF coincides with a quadrangle-augmented Jost--Liu (JL) bound with its positive parts removed, and that when the 2-hop neighborhood is available, the interval can be sharpened at a cost of $\mathcal O(d_{\max}^{2.5})$ per edge.
On \res{NumGraphs} benchmark graphs, our quadrangle-augmented bound is tighter than the bounds of JL and Tian--Lubberts--Weber (TLW) on \res{LoTighterMin}--\res{LoTighterMax}\% of the edges and reduces their mean gap to the exact curvature by a median of \res{GapRedMedian}\%; intersected with the bounds of TLW, the intervals settle the curvature sign of more than half of the edges on \res{CertMajorityCount} of the \res{NumGraphs} graphs, and their midpoints rank the edges at least as faithfully as every other proxy, BF and the surrogate of TLW included, on \res{SpearmanBestCount} of them. On \res{ScreenLargeCount} real networks with up to \res{ScreenLargeMaxEdges} edges, a certify-then-solve procedure answers sign queries exactly while solving a median of \res{ScreenSignWorkRatioMedian} times fewer transport problems, for median speedups of \res{ScreenSignSpeedupMedian}$\times$ over exact evaluation with a min-cut solver that is itself about two orders of magnitude faster than the widely used \texttt{GraphRicciCurvature}. The main results are machine-checked in Lean~4.
\end{abstract}

\section{Introduction}
\label{sec:introduction}
Discrete analogues of Ricci curvature have become standard tools in network science and graph learning, and two families account for most of the literature: transport-based coarse Ricci curvatures \citep{ollivier_ricci_2009, lin_curvature_2011}, which compare the neighborhoods of the two endpoints of an edge through an optimal transport problem, and combinatorial curvatures \citep{forman_method_2003,sreejith_forman_2016,weber_forman_2017,samal_comparative_2018,topping_squashing_2022}, which are functions of local counts like degrees, triangles, and 4-cycles. Curvature characterizes the dynamics of information flow, providing a framework to diagnose message-passing bottlenecks and guide graph operations like pooling and rewiring \citep{wu_ricci_2021,topping_squashing_2022,digiovanni_oversquashing_2023, nguyen_over_2023, liu_curvdrop_2023, sanders_pooling_2025, attali_rewiring_2026,feng_graph_2024, fesser_mitigating_2024, fu_discrete_2024,sun_ricci_2024, chen_graph_2025}. Beyond message passing, curvature and its associated discrete Ricci flows also serve as effective tools for community detection \citep{weber_forman_2017,saucan_discrete_2018,ni_community_2019,sia_ollivier_2019,eidi_edge_2020,tian_curvature_2026}.

The two families trade fidelity for cost. \citeauthor{ollivier_ricci_2009}--Ricci (OR) curvature requires, for every edge $(i,j)$, a Wasserstein-1 problem between measures supported on the closed neighborhoods of $i$ and $j$: a linear program with up to $(d_{\max}+1)^2$ variables, where $d_{\max}$ is the maximum degree of $G$, for which general-purpose solvers also need the local distance matrix between the two neighborhoods. Combinatorial curvatures cost only a few set intersections per edge, and Balanced Forman (BF) curvature in particular was introduced as a cheap counterpart of OR curvature \citep{topping_squashing_2022}. Methods and analyses that use BF in place of OR curvature implicitly rely on the two being close. \citet[Theorem~2]{topping_squashing_2022} show that BF lower-bounds the curvature of \citet[hereafter referred to as LLY curvature]{lin_curvature_2011}, the rescaled limit of OR curvature as idleness tends to $1$. 
BF is \emph{not} a lower bound on OR curvature at fixed idleness, lazy or not (Remark~\ref{rem:bf-not-lower}), albeit described as one in parts of the graph-learning literature \citep{nguyen_over_2023,park_lower_2025,fesser_mitigating_2024}.

\subsection{The Question}
We ask what the BF curvature of an edge certifies about its OR curvature. For $\kappa_{\rm BF}$ and $\kappa_{\rm OR}$ we derive nondecreasing \emph{transfer moduli} $\varphi_{a,b}$ and $\psi_{a,b}$ such that
\begin{equation}
\label{eq:moduli}
\kappa_{\rm BF}(i,j)\ge\zeta\implies\kappa_{\rm OR}(i,j)\ge\varphi_{a,b}(\zeta),
\qquad
\kappa_{\rm BF}(i,j)\le\zeta\implies\kappa_{\rm OR}(i,j)\le\psi_{a,b}(\zeta),
\end{equation}
and both depend only on $\zeta$, $a:=\min\{d_i,d_j\}$, and $b:=\max\{d_i,d_j\}$, where $d_u=\deg(u)$. At $\zeta=\kappa_{\rm BF}(i,j)$ they certify an interval that contains $\kappa_{\rm OR}(i,j)$, and at a threshold $\zeta$, each gives a one-sided guarantee for every edge whose BF value lies on its side of $\zeta$.\footnote{At a value of $\zeta$ that no edge with these degrees attains, the two moduli need not be ordered: for $a=b=2$ the attainable BF values are $0$, $1$, and $3/2$, and at $\zeta=49/48$, $\varphi_{a,b}(\zeta)=49/72>\psi_{a,b}(\zeta)=2/3$.}

Both moduli rest on the fact that once the mass shared by the two endpoints is kept in place, all remaining demand sits on the unique neighbors of the higher-degree endpoint (Lemma~\ref{lem:one-sided}). The lower modulus reads a BF value as a lower bound on triangles and matched 4-cycles, and the upper modulus reads it as an upper bound on the same statistics.

Along the way we establish that BF curvature is a quadrangle-augmented form of the lower bound of \citet[henceforth JL]{jost_clustering_2014} with the positive parts removed and the matching number of the cross-edge graph replaced by a 4-cycle proxy (Proposition~\ref{prop:bf-unclipped}).

\subsection{What This Affords}
The moduli turn BF from a heuristic proxy into a certificate, giving every edge a guaranteed range of OR values at no extra cost. By \eqref{eq:moduli}, threshold queries on OR curvature, such as its sign, are settled by BF alone on every edge whose interval does not straddle the threshold, so a \emph{certify-then-solve} procedure answers them exactly while solving transport problems only on the remaining edges. Together with the bounds of Appendix~\ref{app:reverse}, the moduli also translate hypotheses between the two curvatures, unifying message-passing analyses stated in either: \citet[Theorem~4]{topping_squashing_2022} bound over-squashing under conditions that include $\kappa_{\rm BF}(i,j)\le-2+\delta$ for small $\delta>0$, whereas \citet{nguyen_over_2023} link over-smoothing and over-squashing to positive and negative OR curvature, respectively. Lemma~\ref{lem:idleness} and Proposition~\ref{prop:mixing-transfer} transfer results between the lazy, non-lazy, and LLY curvatures.
\subsection{Contributions}
\begin{itemize}
\item[(C1)] For every edge, $\kappa_{\rm OR}(i,j)\ge\varphi_{a,b}(\kappa_{\rm BF}(i,j))$, where $\varphi_{a,b}$ is affine in the BF value. The bound is exact on every edge that lies on no cycle of length at most five and on every equal-degree edge whose cross-edge graph is a perfect matching,\footnote{This includes all edges of trees, complete graphs, and hypercubes.} (Theorem~\ref{thm:bf-lower}), and it dominates the lift of \citet{topping_squashing_2022} between BF and LLY curvature through the idleness theorem of \citet{bourne_idleness_2018} (Proposition~\ref{prop:bf-range} and Remark~\ref{rem:lazy-dominates}).
\item[(C2)] We derive the first upper bound on OR curvature in terms of BF (Theorem~\ref{thm:BF-to-OR-lazy}). The resulting interval is degenerate on complete graphs and on edges where $a=1$ (Corollary~\ref{cor:bf-interval}), and closed-form lower and upper bounds on BF in terms of OR curvature follow (Appendix~\ref{app:reverse}).
\item[(C3)] The identity of Proposition~\ref{prop:bf-unclipped} bounds the amount by which BF can exceed the non-lazy OR curvature and gives a sufficient condition under which BF lower-bounds it.
\item[(C4)] We prove a quadrangle-augmented JL lower bound on the non-lazy curvature through an explicit transport plan (Theorem~\ref{thm:JL-plus-squares}), which dominates the JL bound and the matching bound of \citet{bhattacharya_exact_2020} (Remark~\ref{rem:bm-comparison}). Its lazy counterpart (Theorem~\ref{thm:lazy-quadrangle}) and a lazy transport envelope (Proposition~\ref{prop:lazy-envelope-tight}) give a 2-hop interval contained in the BF interval (Proposition~\ref{prop:sharpened-interval}).
\item[(C5)] Given $\kappa_{\rm BF}$, the BF interval costs $\mathcal O(1)$ per edge, and the 2-hop interval costs $\mathcal O(d_{\max}^{2.5})$ per edge. On \res{NumGraphs} synthetic and real graphs we compare the bounds with prior combinatorial bounds, measuring how many curvature signs they certify and how faithfully they rank edges; on large real networks we evaluate certify-then-solve (Algorithm~\ref{alg:certify-then-solve}); and we compare runtimes across four graph families (Section~\ref{sec:analysis}).
\item[(C6)] The main bounds are formalized in Lean~4 with Mathlib, with no \texttt{sorry}. The results are proved from the definitions alone, and only the results of \cite{bourne_idleness_2018} (Lemma \ref{lem:idleness} below) enter as a hypothesis (see Appendix~\ref{app:lean}).
\end{itemize}

\noindent
Section~\ref{sec:related} reviews related work and Section~\ref{sec:prelims} fixes notation. Section~\ref{sec:edge-bound-or} proves local bounds on OR curvature, Section~\ref{sec:bf-to-or} derives the BF interval, the identity, and the 2-hop interval, Section~\ref{sec:analysis} reports the experiments, and Section~\ref{sec:conclusion} discusses limitations and open problems. Proofs not given in the main text are in Appendix~\ref{appendix:proofs}, and Appendices~\ref{app:bm-comparison}, \ref{app:envelope}, \ref{app:reverse}, and~\ref{app:lean} contain the comparison with \citet{bhattacharya_exact_2020}, an envelope for arbitrary idleness, the OR to BF direction, and details on the Lean formalization.

\section{Related Work}
\label{sec:related}

\emph{Combinatorial Bounds on OR Curvature.} 
JL give a triangle-based lower bound and the upper bound $t_{ij}/b$ for the non-lazy curvature, where $t_{ij}$ is the number of triangles that contain $(i,j)$. \citet{bhattacharya_exact_2020} bound the non-lazy curvature from below, derive expressions on bipartite graphs and on graphs of girth at least five, and use them to characterize both the Ricci-flat graphs of girth at least five and the limiting curvature of Erd\H{o}s--R\'enyi graphs. \citet{kelly_exact_2020} argues that a general combinatorial formula is unlikely and reduces the edgewise linear program instead. \citet[henceforth TLW]{tian_curvature_2026} extend both JL bounds to weighted graphs and use the mean of the JL bounds as a linear-time surrogate for OR curvature in community detection. \citet{kang_accelerated_2026} extend the JL lower bound to non-uniform one-step measures on spaces with integer-valued distances, including hypergraphs, and use the bound as a linear-time surrogate. More recently, \citet{gu_residual_2026} proposed a residual-shell lower bound for multi-step random walks, which they report to be tighter than that of \citet{kang_accelerated_2026}, and \citet{zhou_cycle_2026} proposed a transport heuristic based on cycle overlaps.

\emph{Idleness and the LLY Limit.}
\citet{bourne_idleness_2018} show that OR curvature, as a function of the idleness, is concave and linear on $[1/(b+1),1]$, and that the non-lazy and LLY curvatures differ by at most $2/b$; \citet{smith_ricci_2014} gives a matching condition under which the idleness function is linear, and \citet{munch_curvature_2017} give a limit-free formula for LLY curvature. Proposition~\ref{prop:mixing-transfer} uses the linearity result of \citet{bourne_idleness_2018} to transfer lower bounds on LLY curvature (like BF) to lazy OR curvature.

\emph{Relations Between OR and Forman-Type Curvatures.}
\citet{jost_forman_2021} characterize OR curvature as the maximum of a modified Forman curvature over all 2-complexes having the graph as 1-skeleton, albeit with no closed-form relation to BF. Empirical studies report strong but regime-dependent agreement between the two curvature families \citep{sreejith_systematic_2017, samal_comparative_2018, saucan_discrete_2018}, and enhanced and augmented Forman curvatures were designed to narrow this gap \citep{tee_enhanced_2021, fesser_mitigating_2024}. To the best of our knowledge, the comparison of \citet{topping_squashing_2022} between BF and LLY curvature has not previously been transferred to OR curvature.

\emph{Computing OR Curvature.}
The transport problem is solved by assignment or network-simplex methods \citep{kuhn_hungarian_1955, ahuja_network_1993}, and its complexity is tied to that of minimum-weight perfect matching on complete bipartite graphs through fine-grained reductions \citep{dasgupta_computing_2023}. Entropic solvers \citep{cuturi_sinkhorn_2013, knight_sinkhorn_2008} reduce the per-iteration cost but still require the local distance matrix and introduce a regularization bias, so they return neither the exact value nor a certified bound. 
Widely used implementations, such as \texttt{GraphRicciCurvature} \citep{ni_community_2019}, solve one such program per edge with \texttt{POT} \citep{flamary_pot_2021}. \citet{micaletto_exact_2026} exploits the bounded diameter of the support to reduce the problem to a single minimum cut; we use that solver in this work.

\section{Preliminaries}
\label{sec:prelims}
We work with a simple (so undirected, loopless, no multi-edges), finite, connected graph \(G=(V,E)\), where $n:=|V|$, $\mathcal N(u):=\{v\in V:(u,v)\in E\}$, $d_u:=|\mathcal N(u)|\equiv\deg(u)$, and $d_{\max}:=\max_{u\in V}d_u$.
Let $\mathrm{dist}_G$ be the shortest-path distance on $G$, then for probability measures $\mu,\nu$ on $V$, a \emph{coupling} $\pi\in\Pi(\mu,\nu)$ is a probability measure on $V\times V$ with marginals $\mu$ and $\nu$, where $\pi(x,y)$ is the mass moved from $x$ to $y$ at cost $\mathrm{dist}_G(x,y)$. The Wasserstein-1 distance is the minimal cost,
\[
W_1(\mu,\nu):=\min_{\pi\in\Pi(\mu,\nu)}\sum_{x,y\in V}\mathrm{dist}_G(x,y)\,\pi(x,y)
=\max_{\Lip(f)\le1}\sum_{x\in V}f(x)\left(\mu(x)-\nu(x)\right),
\]
where the second equality is the Kantorovich--Rubinstein duality \citep[see, for example,][]{villani_topics_2003}. Hence every explicit coupling gives an upper bound on $W_1$, and every 1-Lipschitz function $f$ gives the lower bound $\sum_{x}f(x)\left(\mu(x)-\nu(x)\right)$.

\begin{definition}[Lazy Random-Walk Measure]
\label{def:lazy-rw}
Fix \(u \in V\), then
\(
m_u :=w_u\delta_u +(1-w_u)\nu_u,
\)
with \(
w_u := 1/(d_u + 1),
\)
$\delta_u(k):=\mathds 1\{k=u\}$, and $\nu_u(k):=\mathds 1\{k\in\mathcal N(u)\}/d_u$.
$m_u$ is uniform on $\{u\}\cup\mathcal N(u)$, so it puts mass $w_u$ on every vertex of this set, and $w_u$ is its \emph{idleness}.
\end{definition}

Table~\ref{tab:definitions} summarizes the neighborhood sets and combinatorial quantities for an edge $(i,j) \in E$, and Figure~\ref{fig:local-summaries} visualizes the geometric motifs from which the quantities arise.

\begin{table}
\caption{Definitions evaluated for a fixed edge $(i,j) \in E$ where $d_i\le d_j$.}
\label{tab:definitions}
\centering
\small
\renewcommand{\arraystretch}{1.3}
\begin{tabular}{@{}p{0.29\textwidth}p{0.68\textwidth}@{}}
\toprule
\textbf{Concept / Symbol} & \textbf{Definition / Formula} \\
\midrule
\textbf{Neighborhood Sets} & Common: $\mathcal C := \mathcal N(i)\cap\mathcal N(j)$; \newline Unique: $\mathcal U_i := \mathcal N(i)\setminus \left(\mathcal N(j)\cup \{j\}\right)$ (respectively $\mathcal U_j$). \\
\textbf{Degrees and Triangles} & $a:=\min\{d_i,d_j\}$, $b:=\max\{d_i,d_j\}$; $t_{ij} := |\mathcal C|$. \\
\textbf{Cross-Edge Graph} ($H_{ij}$) & The bipartite graph with sides $\mathcal U_i$, $\mathcal U_j$ and the edges of $G$ between them (the \emph{cross edges}); $M_{ij}$ is its matching number (the size of a maximum matching), and $\gamma_{\max}$ is its maximum degree. \\
\textbf{Cross-Edge Vertices} ($\Xi_{ij}$) & $\Xi_{ij}:=|\xi_i| + |\xi_j|$, the number of non-isolated vertices of $H_{ij}$, where $\xi_i := \{k \in \mathcal U_i : \mathcal N(k)\cap\mathcal U_j\neq\emptyset\}$ (respectively $\xi_j$). \\
\textbf{4-Cycle Term} & $\square_{ij} := \Xi_{ij} / (\gamma_{\max}\,b)$, and $\square_{ij}:=0$ if $\gamma_{\max}=0$. \\
\textbf{Matching Statistic} & $\sigma_{ij} := M_{ij}/b$. \\
\textbf{BF Coefficients} & $S := 2/d_i+2/d_j-2=2/a+2/b-2$; \quad $T := 2/b + 1/a$. \\
\midrule
\textbf{Idleness Gap} & $\Delta := |w_i-w_j|=1/(a+1)-1/(b+1)$. \\
\bottomrule
\end{tabular}
\end{table}

\begin{definition}[Edge Curvatures]\label{def:lly}
For $(i,j)\in E$, the BF curvature \citep[Definition~1]{topping_squashing_2022} is
\begin{equation}
\label{eq:BF-lower-decomp}
\kappa_{\rm BF}(i,j) := \left(S+T\,t_{ij}+\square_{ij}\right)\mathds{1}{\{a>1\}},
\end{equation}
and the lazy OR curvature is $\kappa_{\rm OR}(i,j):=1-W_1(m_i,m_j)$. For uniform idleness $p\in[0,1]$, with $m_u^{p}:=p\,\delta_u+(1-p)\,\nu_u$, let $\kappa_p(i,j):=1-W_1(m^p_i,m^p_j)$; $\kappa_0$ is the non-lazy curvature, and the LLY curvature is $\kappa_{\rm LLY}(i,j):=\lim_{p\to1}\kappa_p(i,j)/(1-p)$.
\end{definition}

\emph{Choice of Measure.}
The literature uses the non-lazy measure $\nu_u$ \citep{jost_clustering_2014,bhattacharya_exact_2020,nguyen_over_2023}, a fixed idleness such as $p=1/2$, the default of \texttt{GraphRicciCurvature} \citep{ni_community_2019}, and the limit $p\to1$ of \citet{lin_curvature_2011}. We use $m_u$ for four reasons. (i)~It is the uniform measure on the ball of radius one around $u$, the graph instance of the $\varepsilon$-step random walk with which \citet{ollivier_ricci_2009} defines coarse Ricci curvature; the same construction, on balls that grow with the graph, underlies the convergence of the Ollivier curvature of random geometric graphs to Ricci curvature \citep{vanderhoorn_ollivier_2021}. (ii)~Laziness removes a parity artifact: on a bipartite graph, the non-lazy measures of adjacent vertices have disjoint supports, so $\kappa_0\le0$ on every edge, even on the hypercube, which the lazy walk shows to be positively curved \citep[Example~8]{ollivier_ricci_2009}. (iii)~On an edge with $d_i=d_j=d$, Lemma~\ref{lem:idleness} gives $\kappa_{\rm OR}(i,j)=\frac{d}{d+1}\,\kappa_{\rm LLY}(i,j)$, so $m_u$ retains, at a fixed idleness, the LLY curvature that BF lower-bounds \citep[Theorem~2]{topping_squashing_2022}. (iv)~$m_u$ is the distribution over which a mean-aggregation layer with self-loops averages, the row at $u$ of $\tilde D^{-1}(A+I)$, with $\tilde D$ the degree matrix of $A+I$ \citep{kipf_semi_2017}: it is the aggregation in the mean case of the message passing analysed by \citet[Equation~1 and Theorem~4.2]{nguyen_over_2023}, and the self-loop augmentation through which \citet{topping_squashing_2022} bound over-squashing.

We use two known results, the first on the dependence of OR curvature on the idleness and the second relating BF to LLY curvature.

\begin{lemma}[{\citealp[Theorems~1.1 and~5.5]{bourne_idleness_2018}}]
\label{lem:idleness}
$\forall (i,j)\in E$ and any idleness $p \in [1/(b+1),1]$,
\(
\kappa_p(i,j)=(1-p)\,\kappa_{\rm LLY}(i,j).
\)
\end{lemma}

\begin{lemma}[{\citealp[Theorem~2]{topping_squashing_2022}}]
\label{lem:topping}
$\forall (i,j)\in E$, $\kappa_{\rm BF}(i,j)\le\kappa_{\rm LLY}(i,j)$.
\end{lemma}

\noindent
The non-lazy curvature is also below the LLY curvature. Mixing an optimal coupling of $\nu_i$ and $\nu_j$ with the coupling $\delta_i\otimes\delta_j$ couples $m^p_i$ and $m^p_j$ at cost $(1-p)\,W_1(\nu_i,\nu_j)+p$, so $\kappa_p(i,j)\ge(1-p)\,\kappa_0(i,j)$ for every $p\in[0,1]$, and hence $\kappa_0(i,j)\le\kappa_{\rm LLY}(i,j)$.

\begin{remark}
\label{rem:bf-not-lower}
Lemma~\ref{lem:topping} does not imply that BF lower-bounds OR curvature. On $Q_n$ ($n\ge2$), $\kappa_{\rm BF}=\kappa_{\rm LLY}=2/n$, whereas for any $p\in[0,1]$, $m^p_i$ and $m^p_j$ share only $2\min\{p,(1-p)/n\}\le2/(n+1)$ and every other unit moves at cost $\ge1$, so $\kappa_p\le2/(n+1)<\kappa_{\rm BF}$.\footnote{On $Q_4$, for example, $\kappa_{\rm OR}=2/5$ and $\kappa_0=0$, against $\kappa_{\rm BF}=1/2$. The proper relationship incorporates idleness: because $Q_n$ is regular, Lemmas~\ref{lem:idleness} and~\ref{lem:topping} yield $\kappa_{\rm OR}=\tfrac{n}{n+1}\,\kappa_{\rm LLY}\ge\tfrac{n}{n+1}\,\kappa_{\rm BF}$.} Theorem~\ref{thm:bf-lower} recovers $2/(n+1)$ on $Q_n$ and also gives a lower bound on irregular edges, and Proposition~\ref{prop:bf-unclipped} bounds $\kappa_{\rm BF}-\kappa_0$ and gives a sufficient condition for it to be nonpositive.
\end{remark}

\begin{figure}
\centering
\begin{tikzpicture}[baseline]
  \BaseUV
  \coordinate (m) at ($(u)!0.5!(v)$);
  \node[v] (c1) at ($(m)+(0, 1.05)$) {};
  \node[v] (c2) at ($(m)+(0,-1.05)$) {};
  \draw[thin] (u)--(c1)--(v) (u)--(c2)--(v);
  \node[tag] at ($(u)+(-0.45,1.55)$) {{\rm (a)}};
\end{tikzpicture}
\hspace{1.2em}
\begin{tikzpicture}[baseline]
  \BaseUV
  \node[v]    (k1) at ($(u)+(-0.65,  1.00)$) {};
  \node[v]    (k2) at ($(u)+(-0.65, -1.00)$) {};
  \node[v,label=right:$v$]    (w1) at ($(v)+( 0.65,  1.00)$) {};
  \draw[thin] (u)--(k1) (u)--(k2) (v)--(w1);

  \draw[cross] (k1) .. controls ($(k1)+(0.90, 0.00)$) and ($(w1)+(-0.90, 0.00)$) .. (w1);
  \draw[cross] (k2) .. controls ($(k2)+(1.10,-0.80)$) and ($(w1)+(-1.10,-0.80)$) .. (w1);

  \node[ring] at (k1) {};
  \node[ring] at (k2) {};
  \node[ring] at (w1) {};
  \node[tag]  at ($(u)+(-0.45,1.55)$) {{\rm (b)}};
\end{tikzpicture}
\hspace{1.2em}
\begin{tikzpicture}[baseline]
  \BaseUV
  \coordinate (m) at ($(u)!0.5!(v)$);

  \node[v,label=left:$k$] (k)  at ($(u)+(-0.65, 1.20)$) {};

  \node[v]                (w1) at ($(v)+( 0.55,  1.40)$) {};
  \node[v]                (w2) at ($(v)+( 0.65, -1.00)$) {};

  \node[v]                (c1) at ($(m)+(0,0.95)$) {};

  \draw[thin]
    (u)--(k)
    (v)--(w1)
    (v)--(w2)
    (u)--(c1)
    (v)--(c1)

    (c1)--(w1);

  \draw[close,shorten <= 0pt,shorten >= 0pt]
       (k) to[out=12,in=180,looseness=0.95] (w1);
  \draw[close,shorten <= 0pt,shorten >= 0pt]
       (k) to[out=-12,in=180,looseness=0.95] (w2);

  \node[tag]  at ($(u)+(-0.45,1.55)$) {{\rm (c)}};
\end{tikzpicture}

\caption{
(a) $t_{ij}=|\mathcal C|=2$; (b) dashed lines denote cross edges of $H_{ij}$, whose endpoints (ring-marked) yield $\Xi_{ij}$ ($=3$ here), and whose maximum degree is $\gamma_{\max}=2$, attained by $v$; (c) dotted lines are again cross edges, here $k$ attains $\gamma_{\max}=2$.
}
\label{fig:local-summaries}
\end{figure}

\noindent
Every vertex of $H_{ij}$ has at most $\max\{|\mathcal U_i|,|\mathcal U_j|\}\le b-1$ neighbors in it, so
\begin{equation}
\label{eq:box-count}
\gamma_{\max}\ \le\ b-1.
\end{equation}
The 4-cycle term of BF is controlled by the matching number of the cross-edge graph.

\begin{lemma}[4-Cycle Term and Matching Number]
\label{lem:c4-matching}
$\forall (i,j)\in E$, $\square_{ij}\le2\sigma_{ij}$.
\end{lemma}

\begin{proof}
If $\gamma_{\max}=0$, then $\square_{ij}=0$ by convention. Otherwise, by K\H{o}nig's theorem \citep[see, for example,][]{lovasz_matching_1986}, $M_{ij}$ equals the minimum size of a vertex cover of $H_{ij}$. Each vertex covers at most $\gamma_{\max}$ edges, and $H_{ij}$ has at least $\max\{|\xi_i|,|\xi_j|\}\ge\Xi_{ij}/2$ edges, so $M_{ij}\ge\Xi_{ij}/(2\gamma_{\max})$, that is, $\square_{ij}=\Xi_{ij}/(\gamma_{\max}\,b)\le2\sigma_{ij}$.
\end{proof}

\section{Local Bounds on Ollivier--Ricci Curvature}
\label{sec:edge-bound-or}
This section proves the bounds on OR curvature from which the intervals are built. Proposition~\ref{prop:mixing-transfer} and Corollary~\ref{cor:mixing-triangle-free} transfer lower bounds on LLY curvature to the lazy curvature; Theorem~\ref{thm:JL-plus-squares} sharpens the JL lower bound with the matching number of the cross-edge graph, Theorem~\ref{thm:lazy-quadrangle} is its lazy counterpart, and Proposition~\ref{prop:lazy-envelope-tight} is an upper envelope expressed through the statistics of BF. Appendix~\ref{app:envelope} extends the upper bound to arbitrary idleness. Throughout, recall that $m_u$ (Definition~\ref{def:lazy-rw}) puts mass $w_u$ on all vertices of $\{u\}\cup\mathcal N(u)$.

\begin{proposition}[Idleness Transfer by Mixing]
\label{prop:mixing-transfer}
Let $(i,j)\in E$ and $q_{ij}:=(2+t_{ij})/b$. Then $\kappa_{\rm LLY}(i,j)\le q_{ij}$ and
\begin{equation}
\label{eq:mixing-transfer}
\kappa_{\rm OR}(i,j)\ \ge\ \frac{a}{a+1}\,\kappa_{\rm LLY}(i,j)-\Delta\,(1-q_{ij}).
\end{equation}
Consequently, for any $L\le\kappa_{\rm LLY}(i,j)$,
\(
\kappa_{\rm OR}(i,j)\ge\Lambda^{\rm mix}_{ij}(L):=L{a}/({a+1})-\Delta(1-q_{ij}),
\)
and $\kappa_{\rm OR}(i,j)\le\vartheta\implies\kappa_{\rm LLY}(i,j)\le\big(\vartheta+\Delta(1-q_{ij})\big){(a+1)}/{a}$.
\end{proposition}

\begin{proof}[Proof Sketch]
Label $i,j$ such that $d_i\le d_j$; $m_i$ mixes $m_i^{w_j}$ and $\delta_i$, while $m_j=m_j^{w_j}$, so we couple $m_i$ and $m_j$ by combining the product coupling of $\delta_i$ and $m_j$ with an optimal coupling of $m_i^{w_j}$ and $m_j$, whose cost is $1-(1-w_j)\kappa_{\rm LLY}(i,j)$ by Lemma~\ref{lem:idleness}. Then $\kappa_{\rm LLY}(i,j)\le q_{ij}$ follows from Lemma~\ref{lem:idleness} and the mass shared by $m^{w_j}_i$ and $m_j$ (see Appendix~\ref{subsec:proof-mixing-transfer}).
\end{proof}

\noindent
\eqref{eq:mixing-transfer} needs $t_{ij}$ only through $q_{ij}$. Dropping it gives a lift that depends on the $d_i,d_j$ alone.

\begin{corollary}[Triangle-Free Version of the Mixing Transfer]
\label{cor:mixing-triangle-free}
For every $(i,j)\in E$ and every $L\le\kappa_{\rm LLY}(i,j)$,
\(
\kappa_{\rm OR}(i,j) \ge \Lambda^{\rm tf}_{ij}(L):=L{a}/{(a+1)}-\Delta(1-\max\{L,2/b\}).
\)
The map $\Lambda^{\rm tf}_{ij}$ depends only on $L$, $a$, and $b$, and coincides with $\Lambda^{\rm mix}_{ij}$ on triangle-free edges.
\end{corollary}

\begin{proof}
By Proposition~\ref{prop:mixing-transfer}, $q_{ij}\ge\kappa_{\rm LLY}(i,j)\ge L$, and $q_{ij}\ge2/b$, so $\Lambda^{\rm mix}_{ij}(L)\ge\Lambda^{\rm tf}_{ij}(L)$ because $\Delta\ge0$; on triangle-free edges $q_{ij}=2/b\ge L$, and the two coincide. 
\end{proof}

\begin{theorem}[Quadrangle-Augmented \citeauthor{jost_clustering_2014} Lower Bound]
\label{thm:JL-plus-squares}
$\forall(i,j)\in E$,
\begin{equation}
\label{eq:JL-aug}
\kappa_0(i,j)\ \ge\ L^\square_{ij}:=\frac{t_{ij}}{b}-\big[A^\square_{ij}\big]_+-\big[B^\square_{ij}\big]_+,
\end{equation}
where $A^\square_{ij}:=1-1/a-1/b-t_{ij}/b-\sigma_{ij}$ and $B^\square_{ij}:=1-1/a-1/b-t_{ij}/a-\sigma_{ij}$. The bound $L^\square_{ij}$ is nondecreasing in $\sigma_{ij}$, is the JL bound if $M_{ij}=0$, and equality holds on any edge that lies on no cycle of length at most $5$ and on any where $a=b$ and $H_{ij}$ has a perfect matching.
\end{theorem}

\begin{proof}[Proof Sketch]
Label the endpoints so that $d_i\le d_j$. Keeping $1/b$ at each common neighbor and moving $1/b$ along a maximum matching of $H_{ij}$ accounts for $t_{ij}/b$ and $\sigma_{ij}$; the remaining mass moves at cost $1$, except what reaches $\mathcal U_j$ from a common neighbor (cost $2$) or from $\mathcal U_i$ (cost $3$), whose extra cost is $[A^\square_{ij}]_+$ and $[B^\square_{ij}]_+$. Exactness follows from a 1-Lipschitz potential for cost $2$ edges and from the JL upper bound for cost $3$ (see Appendix~\ref{subsec:proof-JL-squares}).
\end{proof}

\noindent
On trees, \eqref{eq:JL-aug} reduces to the JL bound. On $Q_n$, flipping the coordinate in which $i$ and $j$ differ is a perfect matching of $H_{ij}$, so \eqref{eq:JL-aug} is exact, giving $0$ where the JL bound gives $-2[1-2/n]_+$; in our experiments it is also exact on every edge of the square grid (Table~\ref{tab:bound-comparison}), which neither case of Theorem~\ref{thm:JL-plus-squares} covers: every grid edge lies on a $4$-cycle, and none has a cross-edge graph with a perfect matching. Since $L^\square_{ij}\le\kappa_0(i,j)\le\kappa_{\rm LLY}(i,j)$ (Section~\ref{sec:prelims}), Proposition~\ref{prop:mixing-transfer} lifts it to $\kappa_{\rm OR}$, and Theorem~\ref{thm:lazy-quadrangle} below dominates this lift (Remark~\ref{rem:lazy-dominates}).

\begin{remark}[Comparison to \citeauthor{bhattacharya_exact_2020}]
\label{rem:bm-comparison}
A bound on $\kappa_0$ via a matching is also in \citet[Theorem~5.1]{bhattacharya_exact_2020}, but between $\mathcal Q_i:=\mathcal N(i)\setminus\mathcal C$ and $\mathcal Q_j:=\mathcal N(j)\setminus\mathcal C$, which contain the endpoints ($j\in\mathcal Q_i$, $i\in\mathcal Q_j$). With $M^{\rm BM}_{ij}$ the size of a maximum matching between $\mathcal Q_i$ and $\mathcal Q_j$, their bound reads
\[
\kappa_0(i,j)\ \ge\ L^{\rm BM}_{ij}:=\frac{t_{ij}}{b}-2\left(1-\frac{M^{\rm BM}_{ij}+t_{ij}}{b}\right).
\]
Since $j$ is adjacent to every vertex of $\mathcal Q_j$ and $i$ to every vertex of $\mathcal Q_i$, $M^{\rm BM}_{ij}=\min\{|\mathcal U_i|+1,\,|\mathcal U_j|+1,\,M_{ij}+2\}$, so both bounds use the same matching information. However, $L^{\rm BM}_{ij}$ is affine in $M^{\rm BM}_{ij}$, whereas \eqref{eq:JL-aug} subtracts $\sigma_{ij}$ inside the two positive parts and keeps the separate contributions of $1/a$ and $1/b$. As a result, $L^{\square}_{ij}\ge L^{\rm BM}_{ij}$ on every edge, with agreement if and only if $a=b$ (Appendix~\ref{app:bm-comparison}); on a tree edge with $a\ge 2$ $L^{\rm BM}_{ij}=-2+4/b$ whereas $L^\square_{ij}=-2+2/a+2/b$. $L^{\square}_{ij}$ also lets us express \eqref{eq:JL-aug} via the components of $\kappa_{\rm BF}$ in Section~\ref{sec:bf-to-or}.
\end{remark}

\begin{remark}[When the Matching Term Helps]
\label{rem:coarse-C4}
Write $L^{\rm JL}_{ij}$ for the JL bound, that is, $L^{\square}_{ij}$ with $\sigma_{ij}$ replaced by $0$. The identity $[x]_+-[x-s]_+=\min\{s,[x]_+\}$, valid for $s\ge 0$, gives
\[
L^{\square}_{ij}-L^{\rm JL}_{ij}
=\min\!\Big\{\sigma_{ij},\Big[1-\frac1a-\frac1b-\frac{t_{ij}}{b}\Big]_+\Big\}+\min\!\Big\{\sigma_{ij},\Big[1-\frac1a-\frac1b-\frac{t_{ij}}{a}\Big]_+\Big\},
\]
which lies in $\left[0,2\sigma_{ij}\right]$.
\end{remark}

For the lazy measure, the transport problem has a one-sided structure, which yields a direct lazy counterpart of Theorem~\ref{thm:JL-plus-squares} and an upper envelope.

\begin{lemma}[One-Sided Residuals]
\label{lem:one-sided}
Fix $(i,j)\in E$ with $d_i\le d_j$. Then
\[
(m_i-m_j)_+=\Delta\,\mathds 1_{\{i,j\}\cup\mathcal C}+w_i\,\mathds 1_{\mathcal U_i},\qquad
(m_i-m_j)_-=w_j\,\mathds 1_{\mathcal U_j},
\]
and both have mass $\tau_{ij}:=|\mathcal U_j|\,w_j=1-(2+t_{ij})\,w_j$. In particular, all residual demand sits on $\mathcal U_j$, $i$ is at distance $2$ from all of it, and $m_i$ and $m_j$ share the mass $\sum_k\min\{m_i(k),m_j(k)\}=1-\tau_{ij}=(2+t_{ij})\,w_j$.
\end{lemma}

\begin{proof}
For $u\in\{i,j\}$, $\{u\}\cup\mathcal N(u)=\{i,j\}\sqcup\mathcal C\sqcup \mathcal U_u$, and $m_u=w_u$ on it. Labeling $i,j$ so $d_i\le d_j$ implies $w_i\ge w_j$, which gives both residuals. A vertex of $\mathcal U_j$ is by definition at distance $2$ from $i$, and finally, $\tau_{ij}=|\mathcal U_j|\,w_j=(b-1-t_{ij})/(b+1)=1-(2+t_{ij})\,w_j$, and $\min\{m_i,m_j\}=w_j$ on $\{i,j\}\cup\mathcal C$ and $0$ elsewhere.
\end{proof}

\noindent
When $a=b$, Lemma~\ref{lem:one-sided} gives $\Delta=0$, so both residuals are uniform, on the sets $\mathcal U_i$ and $\mathcal U_j$ of equal size, and the transport problem is a minimum-weight perfect matching between $\mathcal U_i$ and $\mathcal U_j$ with weights in $\{1,2,3\}$. This is the reformulation of \citet{azarhooshang_review_2020} on which \citet[Fact~1]{dasgupta_computing_2023} base their fine-grained reduction; Lemma~\ref{lem:one-sided} extends the one-sided structure to unequal degrees, where the residual supply is no longer uniform.

\begin{theorem}[Lazy Quadrangle Bound]
\label{thm:lazy-quadrangle}
For every $(i,j)\in E$,
\begin{equation}
\label{eq:lazy-quadrangle}
\kappa_{\rm OR}(i,j)\ \ge\ L^{\rm lazy}_{ij}:=\frac{3+t_{ij}}{a+1}+\frac{3+2t_{ij}+2M_{ij}}{b+1}-2 .
\end{equation}
With the endpoints labeled so that $d_i\le d_j$, write $\Theta_{ij}:=(3+t_{ij})\,w_j-w_i$ for the coverage bound $\kappa_{\rm OR}(i,j)\le\Theta_{ij}$ of Proposition~\ref{prop:lazy-envelope-tight}, and $R_{ij}:=|\mathcal U_i|\,w_i-M_{ij}\,w_j$ for the supply on $\mathcal U_i$ that a maximum matching of $H_{ij}$ leaves unmatched. Then $L^{\rm lazy}_{ij}=\Theta_{ij}-t_{ij}\,\Delta-2R_{ij}$: compared with $\Theta_{ij}$, the residual $\Delta$ on each common neighbor travels one step further, and each unmatched unit on $\mathcal U_i$ two steps further. Equality holds on every edge that lies on no cycle of length at most $5$ (in particular on trees), and on every edge with $a=b$ whose cross-edge graph has a perfect matching (in particular on complete graphs and hypercubes).
\end{theorem}

\begin{proof}[Proof Sketch]
Keeping $\min\{m_i(k),m_j(k)\}$ in place at every $k$, it suffices to couple the residuals of Lemma~\ref{lem:one-sided}. Move $w_j$ along every edge of a maximum matching of $H_{ij}$ (cost $1$), the residual at $j$ into $\mathcal U_j$ (cost $1$), the residuals at $i$ and on $\mathcal C$ into $\mathcal U_j$ through $j$ (cost $2$), and the rest of $\mathcal U_i$ into $\mathcal U_j$ through $i$ and $j$ (cost $3$), splitting the unmatched demand in proportion to these amounts. The cost of this coupling is $1-L^{\rm lazy}_{ij}$. On the first class, a 1-Lipschitz potential shows that no coupling is cheaper; on the second, $L^{\rm lazy}_{ij}=\Theta_{ij}$, which is an upper bound (see 
Appendix~\ref{subsec:proof-lazy-quadrangle}).
\end{proof}

\noindent
Theorem~\ref{thm:lazy-quadrangle} needs neither the idleness results of \citet{bourne_idleness_2018} nor a lift, and it dominates the lift of Theorem~\ref{thm:JL-plus-squares} (Remark~\ref{rem:lazy-dominates}).

\begin{proposition}[Lazy Transport Envelope]
\label{prop:lazy-envelope-tight}
Fix $(i,j)\in E$ such that $d_i\le d_j$. Let $\Theta_{ij}$ be as in Theorem~\ref{thm:lazy-quadrangle}, and let $m^{\rm UU}_{ij}:=\min\{M_{ij}\,w_i,\ |\xi_j|\,w_j\}\le\min\{|\mathcal U_i|\,w_i,\ \Xi_{ij}/(a+b+2)\}$. Then
\begin{equation}
\label{eq:lazy-envelope-tight}
\kappa_{\rm OR}(i,j)\ \le\ U^{\rm lazy}_{ij}:=-1+2(2+t_{ij})\,w_j+\big(1+t_{ij}\big)\,\Delta+m^{\rm UU}_{ij}\ \le\ \Theta_{ij}=(2+t_{ij})\,w_j-\Delta .
\end{equation}
Equivalently, $U^{\rm lazy}_{ij}=\Theta_{ij}-\big(|\mathcal U_i|\,w_i-m^{\rm UU}_{ij}\big)$: every unit on $\mathcal U_i$ that does not leave along a cross edge travels at least one step further than the coverage bound $\kappa_{\rm OR}(i,j)\le\Theta_{ij}$ charges.
\end{proposition}

\begin{proof}
By Kantorovich--Rubinstein duality (Section~\ref{sec:prelims}), $W_1(m_i,m_j)\ge\sum_x f(x)\big(m_i(x)-m_j(x)\big)$ for every 1-Lipschitz $f$. Let $K$ be a vertex cover of $H_{ij}$, and let $f_K:=2$ on $\{i\}\cup(\mathcal U_i\setminus K)$, $f_K:=0$ on $\mathcal U_j\setminus K$, and $f_K:=1$ elsewhere. Its values differ by $2$ only between $\{i\}\cup(\mathcal U_i\setminus K)$ and $\mathcal U_j\setminus K$, and no such pair is adjacent, because $i$ has no neighbor in $\mathcal U_j$ and an edge between $\mathcal U_i\setminus K$ and $\mathcal U_j\setminus K$ would be a cross edge that $K$ misses; hence $f_K$ is 1-Lipschitz. By Lemma~\ref{lem:one-sided}, $m_i-m_j=\Delta\,\mathds 1_{\{i,j\}\cup\mathcal C}+w_i\,\mathds 1_{\mathcal U_i}-w_j\,\mathds 1_{\mathcal U_j}$, so
\[
W_1(m_i,m_j)\ \ge\ (3+t_{ij})\,\Delta+2|\mathcal U_i|\,w_i-\big(|K\cap\mathcal U_i|\,w_i+|K\cap\mathcal U_j|\,w_j\big).
\]
The set $\xi_j$ is a vertex cover of $H_{ij}$, and by K\H{o}nig's theorem so is some set of $M_{ij}$ vertices, for which the subtracted term is at most $M_{ij}\,w_i$ because $w_j\le w_i$. Hence $W_1(m_i,m_j)\ge(3+t_{ij})\,\Delta+2|\mathcal U_i|\,w_i-m^{\rm UU}_{ij}$. Since $(2+t_{ij})\,\Delta+|\mathcal U_i|\,w_i=\tau_{ij}=1-(2+t_{ij})\,w_j$ by Lemma~\ref{lem:one-sided}, this reads $\kappa_{\rm OR}(i,j)\le U^{\rm lazy}_{ij}=\Theta_{ij}-\big(|\mathcal U_i|\,w_i-m^{\rm UU}_{ij}\big)$, and $\Theta_{ij}=(2+t_{ij})\,w_j-\Delta$. The bounds on $m^{\rm UU}_{ij}$ follow from $M_{ij}\le|\xi_i|\le|\mathcal U_i|$ and from $\min\{x,y\}\le(x/w_i+y/w_j)/(1/w_i+1/w_j)$ with $x=|\xi_i|\,w_i$, $y=|\xi_j|\,w_j$, and $1/w_i+1/w_j=a+b+2$; the first gives $U^{\rm lazy}_{ij}\le\Theta_{ij}$.
\end{proof}

\noindent
The two ends of $[L^{\rm lazy}_{ij},U^{\rm lazy}_{ij}]$ are certified explicitly: the lower by the coupling of Theorem~\ref{thm:lazy-quadrangle}, the upper by the potential $f_K$. On most edges the dual bound of Theorem~9 of TLW is tighter than $U^{\rm lazy}_{ij}$ (Table~\ref{tab:bound-comparison-lazy}), but the two do different jobs: relaxed through $m^{\rm UU}_{ij}\le\Xi_{ij}/(a+b+2)$, $U^{\rm lazy}_{ij}$ requires only $a$, $b$, $t_{ij}$, and $\Xi_{ij}$, statistics on which $\kappa_{\rm BF}$ also depends in \eqref{eq:BF-lower-decomp}; this is what lets Theorem~\ref{thm:BF-to-OR-lazy} turn a BF value into an upper bound on OR curvature. The TLW bound depends on the distances between the vertices of positive and of negative net mass, which BF does not determine, so it cannot serve this purpose; we use it instead to sharpen $U^{\rm lazy}_{ij}$ where that pays, in the second stage of Algorithm~\ref{alg:certify-then-solve}. Appendix~\ref{app:envelope} gives an envelope for arbitrary idleness, which we evaluate at zero idleness in Section~\ref{sec:comparison}; for the lazy measure it is dominated by $U^{\rm lazy}_{ij}$.

\section{Bounding Ollivier--Ricci Curvature by Balanced Forman Curvature}
\label{sec:bf-to-or}
We now derive what a BF value certifies about OR curvature, that is, nondecreasing transfer moduli $\varphi_{a,b}$ and $\psi_{a,b}$ satisfying \eqref{eq:moduli}. We first relate BF to the bound of Theorem~\ref{thm:JL-plus-squares}.

\begin{proposition}[BF as an Unclipped JL Bound]
\label{prop:bf-unclipped}
Let $(i,j)\in E$ with $a\ge2$, and define the \emph{BF deficits} $A^{\rm BF}_{ij}:=1-1/a-1/b-t_{ij}/b-\tfrac12\,\square_{ij}$ and $B^{\rm BF}_{ij}:=1-1/a-1/b-t_{ij}/a-\tfrac12\,\square_{ij}$.

\begin{itemize}
\item[{\rm (i)}] $\kappa_{\rm BF}(i,j)=t_{ij}/b-A^{\rm BF}_{ij}-B^{\rm BF}_{ij}$. Equivalently, $\kappa_{\rm BF}(i,j)$ is the right-hand side of \eqref{eq:JL-aug} with $\sigma_{ij}$ replaced by $\tfrac12\square_{ij}$ and the two positive parts removed.
\item[{\rm (ii)}]
\(
L^{\rm BF}_{ij}:=t_{ij}/b-[A^{\rm BF}_{ij}]_+-[B^{\rm BF}_{ij}]_+
=\kappa_{\rm BF}(i,j)-[-A^{\rm BF}_{ij}]_+-[-B^{\rm BF}_{ij}]_+
\)
satisfies $L^{\rm BF}_{ij}\le L^\square_{ij}\le\kappa_0(i,j)$. Thus,
\begin{equation}
\label{eq:bf-excess}
\kappa_{\rm BF}(i,j)-\kappa_0(i,j)\ \le\ \left[-A^{\rm BF}_{ij}\right]_++\left[-B^{\rm BF}_{ij}\right]_+.
\end{equation}
\item[{\rm (iii)}] $A^{\rm BF}_{ij}\ge B^{\rm BF}_{ij}$. In particular, if $B^{\rm BF}_{ij}\ge0$, then $\kappa_{\rm BF}(i,j)\le\kappa_0(i,j)$.
\end{itemize}
\end{proposition}

\begin{proof}
(i) Since $S=-2\,(1-1/a-1/b)$,
\[
\frac{t_{ij}}{b}-A^{\rm BF}_{ij}-B^{\rm BF}_{ij}=S+\frac{2t_{ij}}{b}+\frac{t_{ij}}{a}+\square_{ij}=S+T\,t_{ij}+\square_{ij},
\]
which is $\kappa_{\rm BF}(i,j)$ by \eqref{eq:BF-lower-decomp}; the arguments of the positive parts in \eqref{eq:JL-aug} become $A^{\rm BF}_{ij}$ and $B^{\rm BF}_{ij}$ when $\sigma_{ij}$ is replaced by $\tfrac12\square_{ij}$. (ii) follows from $[x]_+=x+[-x]_+$. By Lemma~\ref{lem:c4-matching}, $\tfrac12\square_{ij}\le\sigma_{ij}$, and the right-hand side of \eqref{eq:JL-aug} is nondecreasing in $\sigma_{ij}$ (Theorem~\ref{thm:JL-plus-squares}), so $L^{\rm BF}_{ij}\le L^\square_{ij}$; Theorem~\ref{thm:JL-plus-squares} gives $L^\square_{ij}\le\kappa_0(i,j)$, and \eqref{eq:bf-excess} follows. (iii) $A^{\rm BF}_{ij}-B^{\rm BF}_{ij}=t_{ij}\,(1/a-1/b)\ge0$; if $B^{\rm BF}_{ij}\ge0$, both clipped terms in~(ii) vanish, and $\kappa_{\rm BF}(i,j)=L^{\rm BF}_{ij}\le\kappa_0(i,j)$.
\end{proof}

\begin{remark}[Where BF Overstates OR Curvature]
\label{rem:bf-interpretation}
In Theorem~\ref{thm:JL-plus-squares}, the arguments of the two positive parts are the masses left on the unique neighborhoods after the moves of cost $0$ and $1$, so triangles and matched 4-cycles help only until these masses are used up. BF credits them linearly instead: it replaces the matching number by its 4-cycle proxy and drops the positive parts. Where triangles and 4-cycles are scarce ($B^{\rm BF}_{ij}\ge0$) the positive parts are inactive and BF lower-bounds $\kappa_0$; where they are abundant, BF can exceed $\kappa_0$ by up to \eqref{eq:bf-excess}, with equality on $Q_n$, where $A^{\rm BF}_{ij}=B^{\rm BF}_{ij}=-1/n$ and $\kappa_{\rm BF}-\kappa_0=2/n$. The condition is not necessary: on a spoke of the wheel whose rim is a $7$-cycle, $B^{\rm BF}_{ij}=-1/7$, yet $\kappa_{\rm BF}=4/21<\kappa_0=2/7$. On the graphs of Section~\ref{sec:comparison}, $B^{\rm BF}_{ij}\ge0$ holds on \res{BdefMin}--\res{BdefMax}\% of the edges.
\end{remark}

\begin{remark}[The Lazy Bound Dominates the Lifts]
\label{rem:lazy-dominates}
Write $P_{ij}:=t_{ij}/b-A^\square_{ij}-B^\square_{ij}=S+T\,t_{ij}+2\sigma_{ij}$ for the third affine piece of $L^\square_{ij}=\min\{t_{ij}/b,\,t_{ij}/b-A^\square_{ij},\,P_{ij}\}$, where $A^\square_{ij}\ge B^\square_{ij}$ are the arguments of the positive parts in \eqref{eq:JL-aug}. A direct computation gives the identity of rational functions of $(t_{ij},M_{ij})$
\[
L^{\rm lazy}_{ij}-\Lambda^{\rm mix}_{ij}(P_{ij})=\frac{(b-a)\big(2M_{ij}+t_{ij}\big)}{b\,(a+1)(b+1)}\ \ge\ 0.
\]
Both $L^\square_{ij}$ and $\kappa_{\rm BF}(i,j)$ are at most $P_{ij}$, the latter by Proposition~\ref{prop:bf-unclipped}~(i) and Lemma~\ref{lem:c4-matching}. Since $\Lambda^{\rm mix}_{ij}$ is nondecreasing, $L^{\rm lazy}_{ij}$ dominates the lift of $\Lambda^{\rm mix}_{ij}$ on every edge.
\end{remark}

\noindent
The lower modulus must bound OR curvature through $\zeta$ and the endpoint degrees alone. Per unit of BF, the bound of Theorem~\ref{thm:lazy-quadrangle} rises by $\lambda_1$ through triangles and by at least $\lambda_2\ge\lambda_1$ through 4-cycles (both defined below), so the worst edge with a given BF value spends the whole budget $[\zeta-S]_+$ on triangles.

\begin{theorem}[A Lower Bound on OR Curvature]
\label{thm:bf-lower}
For $(i,j)\in E$, let
\[
L^{0}_{ij}:=\frac{3}{a+1}+\frac{3}{b+1}-2,\qquad \lambda_1:=\frac{1/(a+1)+2/(b+1)}{T},\qquad \lambda_2:=\frac{b}{b+1},
\]
and
\begin{equation}
\label{eq:phi-elementary}
\varphi_{a,b}(\zeta):=L^{0}_{ij}+\lambda_1\left[\zeta-S\right]_+ .
\end{equation}
Then $\kappa_{\rm BF}(i,j)\ge\zeta$ implies $\kappa_{\rm OR}(i,j)\ge\varphi_{a,b}(\zeta)$. The modulus is continuous, nondecreasing, convex, and piecewise linear, and $\varphi_{a,b}(\zeta)\equiv L^{\rm lazy}_{ij}$ at $t_{ij}=[\zeta-S]_+/T$ and $M_{ij}=0$: the lazy bound of a virtual edge with these degrees and no cross edges. At $\zeta=\kappa_{\rm BF}(i,j)$ it is exact on every edge that lies on no cycle of length at most $5$ and on every edge with $a=b$ whose cross-edge graph is itself a perfect matching (in particular on trees, $K_n$, and $Q_n$).
\end{theorem}
\begin{proof}
A direct computation gives $L^{\rm lazy}_{ij}=L^0_{ij}+\lambda_1T\,t_{ij}+2\lambda_2\,\sigma_{ij}$, which is the statement about the virtual edge, and
\[
\lambda_2T-\left(\frac1{a+1}+\frac2{b+1}\right)=\frac{b-a}{a(a+1)(b+1)}\ \ge\ 0,
\]
so $0\le\lambda_1\le\lambda_2$; in particular the modulus is nondecreasing and convex. If $a=1$, then $t_{ij}\le a-1=0$ and $\mathcal U_i=\emptyset$, so $t_{ij}=\sigma_{ij}=0$, and Theorem~\ref{thm:lazy-quadrangle} gives $\kappa_{\rm OR}(i,j)\ge L^0_{ij}=\varphi_{a,b}(\zeta)$ for every $\zeta\le0=\kappa_{\rm BF}(i,j)$, because $S=2/b>0$. Let $a\ge2$. By \eqref{eq:BF-lower-decomp} and Lemma~\ref{lem:c4-matching}, $0\le\kappa_{\rm BF}(i,j)-S\le T\,t_{ij}+2\sigma_{ij}$, so Theorem~\ref{thm:lazy-quadrangle} and $\lambda_1\le\lambda_2$ give
\[
\kappa_{\rm OR}(i,j)\ \ge\ L^0_{ij}+\lambda_1\big(T\,t_{ij}+2\sigma_{ij}\big)\ \ge\ L^0_{ij}+\lambda_1\big(\kappa_{\rm BF}(i,j)-S\big)\ \ge\ \varphi_{a,b}(\zeta)
\]
whenever $\kappa_{\rm BF}(i,j)\ge\zeta$. On the first class of edges, $t_{ij}=\sigma_{ij}=0$ and $\kappa_{\rm BF}(i,j)=S$ (or $0$ if $a=1$), so $\varphi_{a,b}(\kappa_{\rm BF}(i,j))=L^0_{ij}=L^{\rm lazy}_{ij}$. On the second, $\lambda_1=\lambda_2$ because $a=b$, and $\square_{ij}=2\sigma_{ij}$ because every vertex of $H_{ij}$ has degree one, so again $\varphi_{a,b}(\kappa_{\rm BF}(i,j))=L^{\rm lazy}_{ij}$. Theorem~\ref{thm:lazy-quadrangle} is exact on both classes.
\end{proof}

\begin{proposition}[Range of BF and Comparison with the Lifts]
\label{prop:bf-range}
Every edge satisfies $\kappa_{\rm BF}(i,j)\le S+T(a-1)$. On this range the lower modulus dominates the lifts of Corollary~\ref{cor:mixing-triangle-free}: $\varphi_{a,b}(\zeta)\ge\Lambda^{\rm tf}_{ij}(\zeta)$ for every $\zeta\le S+T(a-1)$ if $a\ge2$, and for every $\zeta\le0=\kappa_{\rm BF}(i,j)$ if $a=1$.
\end{proposition}

\noindent
The proof is in Appendix~\ref{subsec:prop-bf-range}. Neither the modulus nor the comparison uses Lemma~\ref{lem:topping} or Lemma~\ref{lem:idleness}: a BF value is read as a lower budget on triangles and matched 4-cycles, which the coupling of Theorem~\ref{thm:lazy-quadrangle} turns into transport at cost $1$ or $2$.

For the upper modulus, a BF value is read as an upper budget on the same statistics.

\begin{lemma}[Budget Region]
\label{lem:budget}
Fix $(i,j)\in E$ and $\zeta\in\mathbb R$ with $\kappa_{\rm BF}(i,j)\le\zeta$, and let $\beta(\zeta):=[\,\zeta-S\,]_+$ and $t_{\max}(\zeta):=\min\{a-1,\lfloor\beta(\zeta)/T\rfloor\}$. Then
\begin{equation}
\label{eq:bf-budget-region}
t_{ij}\in\{0,1,\dots,t_{\max}(\zeta)\}, \qquad \Xi_{ij} \le b(b-1)\big(\beta(\zeta)-T\,t_{ij}\big).
\end{equation}
\end{lemma}

\begin{proof}
If $a=1$, then $t_{ij}=\Xi_{ij}=0$. Otherwise \eqref{eq:BF-lower-decomp} and $\kappa_{\rm BF}(i,j)\le\zeta$ give $T\,t_{ij}+\square_{ij}\le\zeta-S$, and since the left-hand side is nonnegative it is at most $\beta(\zeta)$. This gives $t_{ij}\le\beta(\zeta)/T$, while $t_{ij}\le a-1$ in a simple graph, so the integer $t_{ij}$ is at most $t_{\max}(\zeta)$; and $\Xi_{ij}=\gamma_{\max}\,b\,\square_{ij}\le b(b-1)\,\square_{ij}$ by \eqref{eq:box-count}.
\end{proof}

\begin{theorem}[An Upper Bound on OR Curvature]
\label{thm:BF-to-OR-lazy}
Fix $(i,j)\in E$, label the endpoints so that $d_i\le d_j$, and fix $\zeta\in\mathbb R$ with $\kappa_{\rm BF}(i,j)\le \zeta$. For real $t$, let $\Theta_{ij}(t):=(3+t)\,w_j-w_i$ and define the envelope
\begin{equation}
\label{eq:widehat-Psi}
\begin{aligned}
\Psi(t) := \min\Big\{&(2+t)\,w_j-\Delta,\\
&-1+2(2+t)\,w_j+(1+t)\,\Delta+\frac{b(b-1)\left(\beta(\zeta)-T\,t\right)}{a+b+2}\Big\};
\end{aligned}
\end{equation}
equivalently, $\Psi(t)=\Theta_{ij}(t)-\big[(a-1-t)\,w_i-b(b-1)(\beta(\zeta)-T\,t)/(a+b+2)\big]_+$. Then
\[
\kappa_{\rm OR}(i,j) \le \psi_{a,b}(\zeta) := \max\left\{\Psi(t):\ t\in\{0,1,\dots,t_{\max}(\zeta)\}\right\},
\]
with $\beta(\zeta)$ and $t_{\max}(\zeta)$ as in Lemma~\ref{lem:budget}. The first piece of $\Psi$ is increasing and the second is affine, so the maximum is attained at $0$, at $t_{\max}(\zeta)$, or at an integer adjacent to the point where the two pieces cross. The modulus $\psi_{a,b}$ is nondecreasing.
\end{theorem}

\begin{proof}[Proof Sketch]
In Proposition~\ref{prop:lazy-envelope-tight}, $m^{\rm UU}_{ij}$ is at most $|\mathcal U_i|\,w_i$, which gives the first piece of $\Psi(t_{ij})$, and at most $\Xi_{ij}/(a+b+2)$, which Lemma~\ref{lem:budget} turns into the second. Since $t_{ij}$ is not determined by $\zeta$, we maximize over its admissible values. As $\zeta$ grows, the range grows and the second piece increases, so $\psi_{a,b}$ is nondecreasing.
\end{proof}

\begin{corollary}[The BF Interval]
\label{cor:bf-interval}
For every edge $(i,j)\in E$,
\[
\kappa_{\rm OR}(i,j)\ \in\ \mathcal I^{\rm BF}_{ij}:=\left[\varphi^{\mathbb Z}_{a,b}\left(\kappa_{\rm BF}(i,j)\right),\ \psi_{a,b}\left(\kappa_{\rm BF}(i,j)\right)\right],
\]
where, with $L^0_{ij}$, $\lambda_1$, and $\lambda_2$ as in Theorem~\ref{thm:bf-lower} and $\beta(\zeta)$ and $t_{\max}(\zeta)$ as in Lemma~\ref{lem:budget},
\begin{equation}
\label{eq:phi-int}
\varphi^{\mathbb Z}_{a,b}(\zeta):=L^0_{ij}+\lambda_1T\,t_{\max}(\zeta)+\lambda_2\big(\beta(\zeta)-T\,t_{\max}(\zeta)\big)\ \ge\ \varphi_{a,b}(\zeta).
\end{equation}
Both ends depend only on $\kappa_{\rm BF}(i,j)$, $a$, and $b$, and are evaluated in $\mathcal O(1)$ time. On every edge with an endpoint of degree one, both ends equal $3/(b+1)-1/2=\kappa_{\rm OR}(i,j)$, and on $K_n$ both equal $1$, so the BF interval is a single point there.
\end{corollary}

\begin{proof}
The upper end is Theorem~\ref{thm:BF-to-OR-lazy} with $\zeta=\kappa_{\rm BF}(i,j)$, and it requires at most four evaluations of $\Psi$; the lower end is in closed form. Since $\varphi_{a,b}(\zeta)=L^0_{ij}+\lambda_1\beta(\zeta)$, the inequality in \eqref{eq:phi-int} holds because $\varphi^{\mathbb Z}_{a,b}(\zeta)-\varphi_{a,b}(\zeta)=(\lambda_2-\lambda_1)(\beta(\zeta)-T\,t_{\max}(\zeta))$, where $\lambda_1\le\lambda_2$ (proof of Theorem~\ref{thm:bf-lower}) and $T\,t_{\max}(\zeta)\le\beta(\zeta)$. If $a=1$, then $\kappa_{\rm BF}(i,j)=0<S$, so $\beta(0)=t_{\max}(0)=0$ and $\varphi^{\mathbb Z}_{a,b}(0)=L^0_{ij}\le\kappa_{\rm OR}(i,j)$ by Theorem~\ref{thm:bf-lower}. Let $a\ge2$ and $\zeta:=\kappa_{\rm BF}(i,j)$, so that $\beta(\zeta)=T\,t_{ij}+\square_{ij}$ by \eqref{eq:BF-lower-decomp}. Theorem~\ref{thm:lazy-quadrangle} and Lemma~\ref{lem:c4-matching} give
\[
\kappa_{\rm OR}(i,j)\ \ge\ L^0_{ij}+\lambda_1T\,t_{ij}+2\lambda_2\sigma_{ij}\ \ge\ L^0_{ij}+\lambda_2\beta(\zeta)-(\lambda_2-\lambda_1)\,T\,t_{ij} .
\]
Since $t_{ij}$ is an integer with $T\,t_{ij}\le\beta(\zeta)$ and $t_{ij}\le a-1$, we have $t_{ij}\le t_{\max}(\zeta)$, and the right-hand side is at least $L^0_{ij}+\lambda_2\beta(\zeta)-(\lambda_2-\lambda_1)\,T\,t_{\max}(\zeta)=\varphi^{\mathbb Z}_{a,b}(\zeta)$. On an edge with $a=1$, $\beta(0)=t_{\max}(0)=0$, so $\psi_{a,b}(0)=\Psi(0)$, whose two pieces both equal $3/(b+1)-1/2=L^0_{ij}$; this is $\kappa_{\rm OR}(i,j)$ because the edge lies on no cycle (Theorem~\ref{thm:lazy-quadrangle}). On $K_n$, $a=b=n-1$ and $\kappa_{\rm BF}(i,j)=n/(n-1)$, so $\varphi^{\mathbb Z}_{a,b}=\varphi_{a,b}=1$ (Theorem~\ref{thm:bf-lower}), and $\psi_{a,b}\le1$ because the first piece of $\Psi(t)$ is at most $(2+n-2)/n=1$ for $t\le a-1=n-2$.
\end{proof}

\noindent
Unlike $\varphi_{a,b}$, the lower end $\varphi^{\mathbb Z}_{a,b}$ is not monotone in $\zeta$, because it uses the integrality of the triangle count of the edge itself, so it sharpens the lower end but is not a modulus in the sense of \eqref{eq:moduli}. The BF interval is what a BF value certifies by itself, which is the relevant object when BF is all that a method uses: $\psi_{a,b}(\kappa_{\rm BF}(i,j))<0$ certifies that $(i,j)$ is negatively curved in the OR sense, and $\varphi^{\mathbb Z}_{a,b}(\kappa_{\rm BF}(i,j))>0$ that it is positively curved. When the 2-hop neighborhood of the edge is available, both ends can be improved.

\begin{proposition}[The 2-Hop Interval]
\label{prop:sharpened-interval}
For every $(i,j)\in E$, with $L^{\rm lazy}_{ij}$ and $U^{\rm lazy}_{ij}$ as in Theorem~\ref{thm:lazy-quadrangle} and Proposition~\ref{prop:lazy-envelope-tight},
\begin{equation}\label{eq:sharpened-interval}
L^{\rm lazy}_{ij} \le \kappa_{\rm OR}(i,j) \le U^{\rm lazy}_{ij},
\end{equation}
and this interval is contained in $\mathcal I^{\rm BF}_{ij}$.
\end{proposition}

\begin{proof}
The inequalities are Theorem~\ref{thm:lazy-quadrangle} and Proposition~\ref{prop:lazy-envelope-tight}. If $a=1$, then $t_{ij}=M_{ij}=0$, so $L^{\rm lazy}_{ij}=L^0_{ij}=\varphi^{\mathbb Z}_{a,b}(\kappa_{\rm BF}(i,j))$. If $a\ge2$, then with $\zeta:=\kappa_{\rm BF}(i,j)$,
\[
L^{\rm lazy}_{ij}-\varphi^{\mathbb Z}_{a,b}(\zeta)=(\lambda_2-\lambda_1)\,T\,\big(t_{\max}(\zeta)-t_{ij}\big)+\lambda_2\big(2\sigma_{ij}-\square_{ij}\big)\ \ge\ 0
\]
by Lemmas~\ref{lem:budget} and~\ref{lem:c4-matching}. At the upper end, $(t_{ij},\Xi_{ij})$ satisfies \eqref{eq:bf-budget-region} with $\zeta=\kappa_{\rm BF}(i,j)$, and the two bounds on $m^{\rm UU}_{ij}$ in Proposition~\ref{prop:lazy-envelope-tight} give $U^{\rm lazy}_{ij}\le\Psi(t_{ij})\le\psi_{a,b}(\kappa_{\rm BF}(i,j))$.
\end{proof}

\noindent
Any further valid upper bound, like the dual bound of TLW, can be intersected with $U^{\rm lazy}_{ij}$, as we do in Section~\ref{sec:analysis}. That bound needs, for every vertex of positive net mass $m_i-m_j$, its distance to the set of vertices of negative net mass, and conversely. These distances are $1$, $2$, or $3$, and they are determined by the edges between $\mathcal U_i\cup\mathcal C$ and $\mathcal U_j$ and by whether a vertex of $\mathcal U_i$ and one of $\mathcal U_j$ have a common neighbor, which one more pass over the neighbor lists of $\mathcal U_i\cup\mathcal U_j$ decides; the bound therefore needs no distance matrix and costs $\mathcal O(d_{\max}^2)$ per edge. Section~\ref{sec:comparison} evaluates it in the pass that computes the 2-hop statistics, and certify-then-solve only on the edges that the 2-hop interval leaves undecided (Section~\ref{sec:screening}). With constant-time adjacency tests (hashed adjacency sets, or vertex marks as in our implementation), the triangle count, 4-cycle statistics, and edge set of $H_{ij}$ also cost $\mathcal O(d_{\max}^2)$ per edge, and so do $\kappa_{\rm BF}$ and $\mathcal I^{\rm BF}_{ij}$. Both ends of the 2-hop interval use the matching number of $H_{ij}$, whose computation dominates their cost, at $\mathcal O(|E(H_{ij})|\sqrt{|V(H_{ij})|})\le\mathcal O(d_{\max}^{2.5})$ \citep{hopcroft_n52_1973}. Table~\ref{tab:intervals} summarizes what each interval needs and costs.

\begin{table}
\centering
\caption{The intervals described; each contains the lazy OR curvature and is contained in the one above it. Costs are per edge with constant-time adjacency tests.}
\label{tab:intervals}
\footnotesize
\renewcommand{\arraystretch}{1.25}
\begin{tabularx}{\textwidth}{@{}>{\RaggedRight\arraybackslash}p{0.27\textwidth}>{\RaggedRight\arraybackslash}X>{\RaggedRight\arraybackslash}p{0.19\textwidth}>{\RaggedRight\arraybackslash}X@{}}
\toprule
Interval & Inputs & Cost & Used in \\
\midrule
BF interval $\mathcal I^{\rm BF}_{ij}$ (Corollary~\ref{cor:bf-interval}) & $\kappa_{\rm BF}(i,j)$, $a$, $b$ & $\mathcal O(1)$ given $\kappa_{\rm BF}(i,j)$, which costs $\mathcal O(d_{\max}^2)$ & Sections~\ref{sec:bands} and~\ref{sec:screening} \\
2-hop interval $[L^{\rm lazy}_{ij},U^{\rm lazy}_{ij}]$ (Proposition~\ref{prop:sharpened-interval}) & $d_i\le d_j$, $t_{ij}$, $|\xi_j|$, and $M_{ij}$ & $\mathcal O(d_{\max}^{2.5})$, dominated by one maximum matching of $H_{ij}$ & First stage of certify-then-solve (Section~\ref{sec:screening}) \\
Sharpened interval: the 2-hop interval intersected with $U^{\rm prior}_{ij}$ (Section~\ref{sec:comparison}) & The above and the distances between $\mathcal U_i\cup\mathcal C$ and $\mathcal U_j$ & $\mathcal O(d_{\max}^{2.5})$, as above & Section~\ref{sec:comparison}; second stage of certify-then-solve \\
\bottomrule
\end{tabularx}
\end{table}

\section{Empirical Evaluation}
\label{sec:analysis}

We test whether the new bounds are tighter than the prior combinatorial bounds and certify more curvature signs (Section~\ref{sec:comparison}), what a BF value alone certifies (Section~\ref{sec:bands}), how much transport certification saves in exact threshold queries on networks with up to $10^6$ edges (Section~\ref{sec:screening}), and how the cost of the intervals grows with graph size, up to tens of millions of edges (Section~\ref{sec:runtime}). Every exact curvature is computed with the solver of \citet{micaletto_exact_2026} (package \texttt{orcut}, version~1.0.0), which reduces each edgewise transport problem to one minimum cut in exact integer arithmetic and is about two orders of magnitude faster than the widely used \texttt{GraphRicciCurvature} \citep{ni_community_2019,micaletto_exact_2026}. This cuts both ways: it gives exact ground truth on every edge of networks with over $10^6$ edges, but it makes the baseline of Sections~\ref{sec:screening} and~\ref{sec:runtime} the fastest exact computation available to us, so the wall-clock gains reported there are modest; against \texttt{GraphRicciCurvature} they would be about a hundred times larger, a gain that belongs to the solver.

The testbeds are random graph models that control degree heterogeneity, clustering, geometry, and community structure, namely Erd\H{o}s--R\'enyi (ER; \citealp{erdos_random_1959}), Barab\'asi--Albert (BA; \citealp{barabasi_emergence_1999}), Watts--Strogatz (WS; \citealp{watts_collective_1998}), random geometric (RGG; \citealp{penrose_random_2003}), hyperbolic random (HRG; \citealp{krioukov_hyperbolic_2010}), stochastic block model (SBM; \citealp{holland_stochastic_1983}), and random regular graphs \citep{steger_generating_1999}, with parameters listed in the tables (for instance ER$(n,p)$ and HRG$(n,R,\alpha,T)$, with disk radius $R$, radial exponent $\alpha$, and temperature $T$); canonical families that provide clean extremes (cycles, grids, tori, $d$-ary trees, and complete graphs); and real networks (Karate, Jazz, Power Grid, Yeast, and the arXiv HepPh citation network, denoted Arxiv in the tables; \citealp{zachary_karate_1977,gleiser_community_2003,watts_collective_1998,milo_motifs_2002,gehrke_arxiv_2003}). Unless stated otherwise, all curvatures use the idleness of Definition~\ref{def:lazy-rw}. Every graph is restricted to its largest connected component and every quantity is computed on every edge, so the tables, the figures, and the screening fractions of Section~\ref{sec:screening} come from the same per-edge values. The tables list \res{NumGraphsAll} graphs in three groups: \res{NumRandom} random graphs, \res{NumCanonical} canonical graphs (\res{CanonicalList}), and \res{NumReal} real networks. Aggregate statements (ranges, medians, and counts over graphs) refer to the \res{NumGraphs} random and real graphs; the canonical graphs are listed for reference only, as several bounds are exact on them by construction, making ratios like the gap reduction undefined.

\subsection{Comparison with Prior Combinatorial Bounds}
\label{sec:comparison}

For the non-lazy curvature $\kappa_0$, the setting in which the published bounds are stated, we compare the lower bound $L^\square_{ij}$ of Theorem~\ref{thm:JL-plus-squares} with the lower bounds of JL and TLW (specialized to unit weights and zero idleness), and the envelope of Proposition~\ref{prop:lazy-envelope-sharpened} at zero idleness with the JL upper bound $t_{ij}/b$ and the TLW dual bound. BF is not among them, since it lower-bounds $\kappa_0$ only under the condition $B^{\rm BF}_{ij}\ge0$ of Proposition~\ref{prop:bf-unclipped}. For the lazy curvature $\kappa_{\rm OR}$, we compare the lazy quadrangle bound $L^{\rm lazy}_{ij}$ of Theorem~\ref{thm:lazy-quadrangle} with the JL bound lifted by Proposition~\ref{prop:mixing-transfer}, and the envelope $U^{\rm lazy}_{ij}$ of Proposition~\ref{prop:lazy-envelope-tight} with the best prior certificate $U^{\rm prior}_{ij}$, the minimum of the TLW dual bound and the zero-cost bound $\kappa_{\rm OR}(i,j)\le(2+t_{ij})\,w_j$ (the lazy form of the JL upper bound). We call $[L^{\rm lazy}_{ij},\min\{U^{\rm prior}_{ij},U^{\rm lazy}_{ij}\}]$ the \emph{sharpened interval}. By Remark~\ref{rem:lazy-dominates} and Proposition~\ref{prop:sharpened-interval}, lifting $L^\square_{ij}$ or BF, or adding the lower end of the BF interval, cannot improve its lower end. The bounds are evaluated in floating point and the exact curvatures in integer arithmetic, so every decision taken from an interval (a certified sign in Tables~\ref{tab:bound-comparison-lazy} and~\ref{tab:bf-interval}, and the screening of Section~\ref{sec:screening}) uses the interval widened by $10^{-9}$ on each side; across all graphs, no bound exceeds the exact value on the wrong side by more than \res{MaxViolation}.

\begin{table}
\centering
\caption{Tightness of combinatorial bounds on $\kappa_0$. Entries are mean gaps to the exact value for JL, TLW, and ours; the smallest gap is in bold. ``Tighter'' is the percentage of edges on which our bound is strictly tighter than the best prior bound; for upper bounds we also report the percentage on which it is looser.}
\label{tab:bound-comparison}
\begin{adjustbox}{max width=\textwidth}
\footnotesize
\setlength{\tabcolsep}{4pt}
\begin{tabular}{l r ccc r ccc r}
\toprule
& & \multicolumn{4}{c}{Lower bound} & \multicolumn{4}{c}{Upper bound} \\
\cmidrule(lr){3-6} \cmidrule(lr){7-10}
Graph & $|E|$ & JL & TLW & Ours & Tighter (\%) & JL & TLW & Ours & Tighter / looser (\%) \\
\midrule
\resrows{bound-comparison}
\bottomrule
\end{tabular}
\end{adjustbox}
\end{table}

\subsubsection{Lower Bounds}
The bound $L^\square_{ij}$ dominates JL by Remark~\ref{rem:coarse-C4} (in our testing $L^{\rm JL}_{ij}>L^\square_{ij}$ never happens), and it is never weaker than TLW in our experiments. It is tighter than the best prior bound on \res{LoTighterMin}\% (\res{LoTighterMinGraphs}) to \res{LoTighterMax}\% (\res{LoTighterMaxGraphs}) of the edges, and it reduces the mean gap to the exact curvature by up to \res{GapRedMax}\% (\res{GapRedMaxGraphs}), with a median of \res{GapRedMedian}\% across the \res{GapRedCount} graphs on which the reduction is defined (Table~\ref{tab:bound-comparison}). The canonical graphs isolate the source of the gain: on the grid and the tori every edge lies on 4-cycles, and the matching term makes $L^\square_{ij}$ exact where JL and TLW are off by about one unit. In the lazy setting, $L^{\rm lazy}_{ij}$ provably dominates the lift of $L^\square_{ij}$ (Remark~\ref{rem:lazy-dominates}), and its mean gap is below that of the lifted JL bound on every random and real graph (Table~\ref{tab:bound-comparison-lazy}).

\subsubsection{Upper Bounds}
At zero idleness, the envelope has the smaller mean gap on \res{UpBetterCount} of the \res{NumGraphs} graphs, and it at least halves the gap of the JL upper bound on \res{UpHalvedCount} of them; where it is looser (\res{UpLooserList}), the mean excess is at most \res{UpMaxExcess}. In the lazy setting the comparison reverses: the prior certificate has the smaller mean gap on \res{LazyPriorBetterCount} of the \res{NumGraphs} graphs, and the envelope is tighter on at most \res{LazyEnvTighterMax}\% of the edges (Table~\ref{tab:bound-comparison-lazy}), consistent with Remark~\ref{rem:envelope-role}.

\begin{table}
\centering
\caption{We compare lower bounds (JL lifted by Proposition~\ref{prop:mixing-transfer} versus Theorem~\ref{thm:lazy-quadrangle}) and upper bounds (the tightest prior bound versus the bound of Proposition~\ref{prop:lazy-envelope-tight}). For each bound, columns display the mean absolute gap to $\kappa_{\rm OR}$ and the percentage of edges where our bound is tighter. The last columns report the fraction of edges with a certified curvature sign (lower bound $>0$ or upper bound $<0$) using prior bounds versus the sharpened interval. Best performances are highlighted in bold.}
\label{tab:bound-comparison-lazy}
\begin{adjustbox}{max width=\textwidth}
\footnotesize
\setlength{\tabcolsep}{4pt}
\begin{tabular}{l cc r cc r cc}
\toprule
& \multicolumn{3}{c}{Lower bound} & \multicolumn{3}{c}{Upper bound} & \multicolumn{2}{c}{Sign certified (\%)} \\
\cmidrule(lr){2-4} \cmidrule(lr){5-7} \cmidrule(lr){8-9}
Graph & JL & Ours & Tighter (\%) & Prior & Ours & Tighter (\%) & Prior & Ours \\
\midrule
\resrows{bound-comparison-lazy}
\bottomrule
\end{tabular}
\end{adjustbox}
\end{table}

\subsubsection{Sign Certification and Ranking}
The sharpened interval certifies the curvature sign (lower bound positive or upper bound negative) on more edges than the prior certificates on \res{CertMoreCount} of the \res{NumGraphs} graphs and on as many on \res{CertEqualCount}, and it certifies the sign of more than half of the edges on \res{CertMajorityCount} of them, with the largest gains on \res{CertTopGains} (Table~\ref{tab:bound-comparison-lazy}). Rankings matter as much as signs, because curvature-based methods act on the most negatively curved edges: SDRF rewires around the edges of smallest BF \citep{topping_squashing_2022}, and TLW cluster with the mean of the two JL bounds. Table~\ref{tab:ranking} measures how faithfully each proxy ranks the edges by exact lazy curvature. The midpoint of the sharpened interval attains the highest (or tied-highest) Spearman correlation on \res{SpearmanBestCount} of the \res{SpearmanCount} graphs on which the correlation is defined, with the largest gains over the best other proxy on \res{SpearmanTopGains}, and its gains over BF are largest where BF ranks worst. Recall of the 10\% most negatively curved edges is less uniform: the sharpened interval is best (or tied-best) on \res{RecallOursBestCount} of the \res{RecallCount} graphs.

\begin{table}
\centering
\caption{Spearman correlation $\rho$ and recall of the 10\% most negatively curved edges (for $\kappa_{\rm OR}$), for BF, the mean of the two JL bounds (the surrogate of TLW), the midpoint of the best prior certified interval, and the midpoint of the sharpened interval. The best value is in bold; ties in the recall are broken at random with a fixed seed.}
\label{tab:ranking}
\begin{adjustbox}{max width=\textwidth}
\footnotesize
\setlength{\tabcolsep}{4pt}
\begin{tabular}{l cccc cccc}
\toprule
& \multicolumn{4}{c}{Spearman $\rho$} & \multicolumn{4}{c}{Top-10\% recall} \\
\cmidrule(lr){2-5} \cmidrule(lr){6-9}
Graph & BF & JL mean & Prior interval & Ours & BF & JL mean & Prior interval & Ours \\
\midrule
\resrows{ranking}
\bottomrule
\end{tabular}
\end{adjustbox}
\end{table}

\begin{table}
\centering
\caption{The BF interval $\mathcal I^{\rm BF}_{ij}$ (Corollary~\ref{cor:bf-interval}) and the sharpened interval (Section~\ref{sec:comparison}). Mean gap to the exact lazy curvature of the lower end $\varphi^{\mathbb Z}_{a,b}(\kappa_{\rm BF})$ of the BF interval and of the lazy quadrangle bound $L^{\rm lazy}$; percentage of edges on which $L^{\rm lazy}$ is strictly tighter; percentage of edges with a certified sign; and Spearman correlation of the interval midpoint with the exact lazy curvature.}
\label{tab:bf-interval}
\begin{adjustbox}{max width=\textwidth}
\footnotesize
\setlength{\tabcolsep}{4pt}
\begin{tabular}{l cc r cc cc}
\toprule
& \multicolumn{3}{c}{Lower end} & \multicolumn{2}{c}{Sign certified (\%)} & \multicolumn{2}{c}{Spearman $\rho$} \\
\cmidrule(lr){2-4} \cmidrule(lr){5-6} \cmidrule(lr){7-8}
Graph & BF interval & Lazy bound & Tighter (\%) & BF interval & Sharpened & BF interval & Sharpened \\
\midrule
\resrows{bf-interval}
\bottomrule
\end{tabular}
\end{adjustbox}
\end{table}

\subsection{What the BF Interval Certifies Across Graph Families}
\label{sec:bands}

Figures~\ref{fig:scatters} and \ref{fig:dists_sel} compare the lazy curvature with the interval $\mathcal I^{\rm BF}_{ij}$ of Corollary~\ref{cor:bf-interval}, that is, with what a BF value and the two endpoint degrees certify without any further information: the certificate that every method using BF obtains at no extra cost.

\subsubsection{Edgewise View}
Figure~\ref{fig:scatters} plots the density of the edges against their BF value, with the conditional medians of $\kappa_{\rm OR}$ and of the two ends of the BF interval. On ER, SBM, BA, WS, and RGG, most edges lie closer to the lower end than to the upper end, so typical edges nearly attain the lower modulus. The upper end is loosest on the sparse graphs, where it jumps as soon as an edge lies on a single 4-cycle: Lemma~\ref{lem:budget} can only bound $\Xi_{ij}$ by $b(b-1)\square_{ij}$, so one 4-cycle lets $\Psi$ reach the coverage bound $\Theta_{ij}$ (the step in the ER and SBM panels). On the triangle-rich HRG, the median of $\kappa_{\rm OR}$ instead follows the upper end, and the lower end lies far below. BA graphs fan out: on hub-incident edges the degree penalties make BF very negative, while OR curvature remains comparatively bounded because transport absorbs part of the mass at the endpoints.

\begin{figure}
\centering
\includegraphics[width=.43\textwidth]{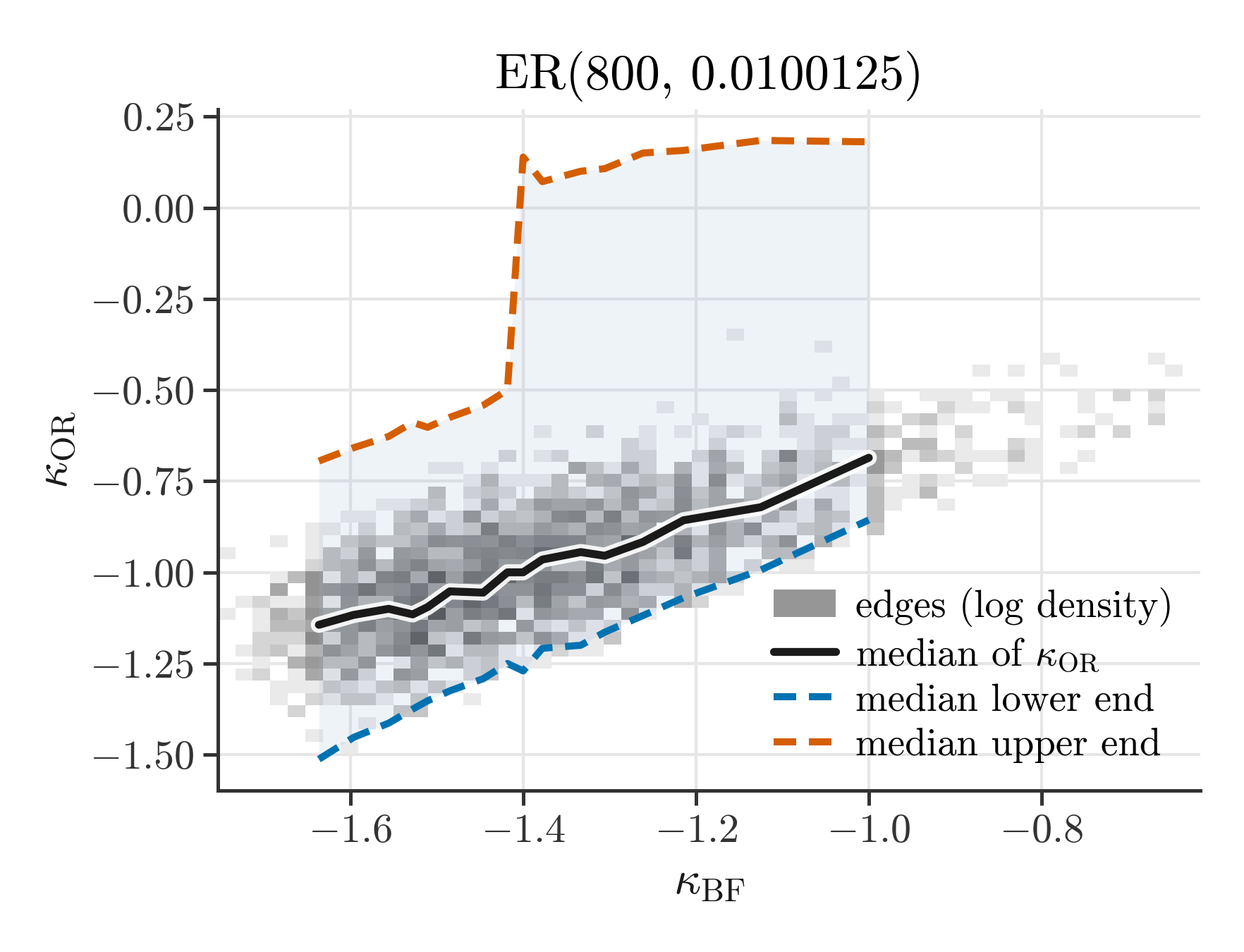}
\hfill
\includegraphics[width=.43\textwidth]{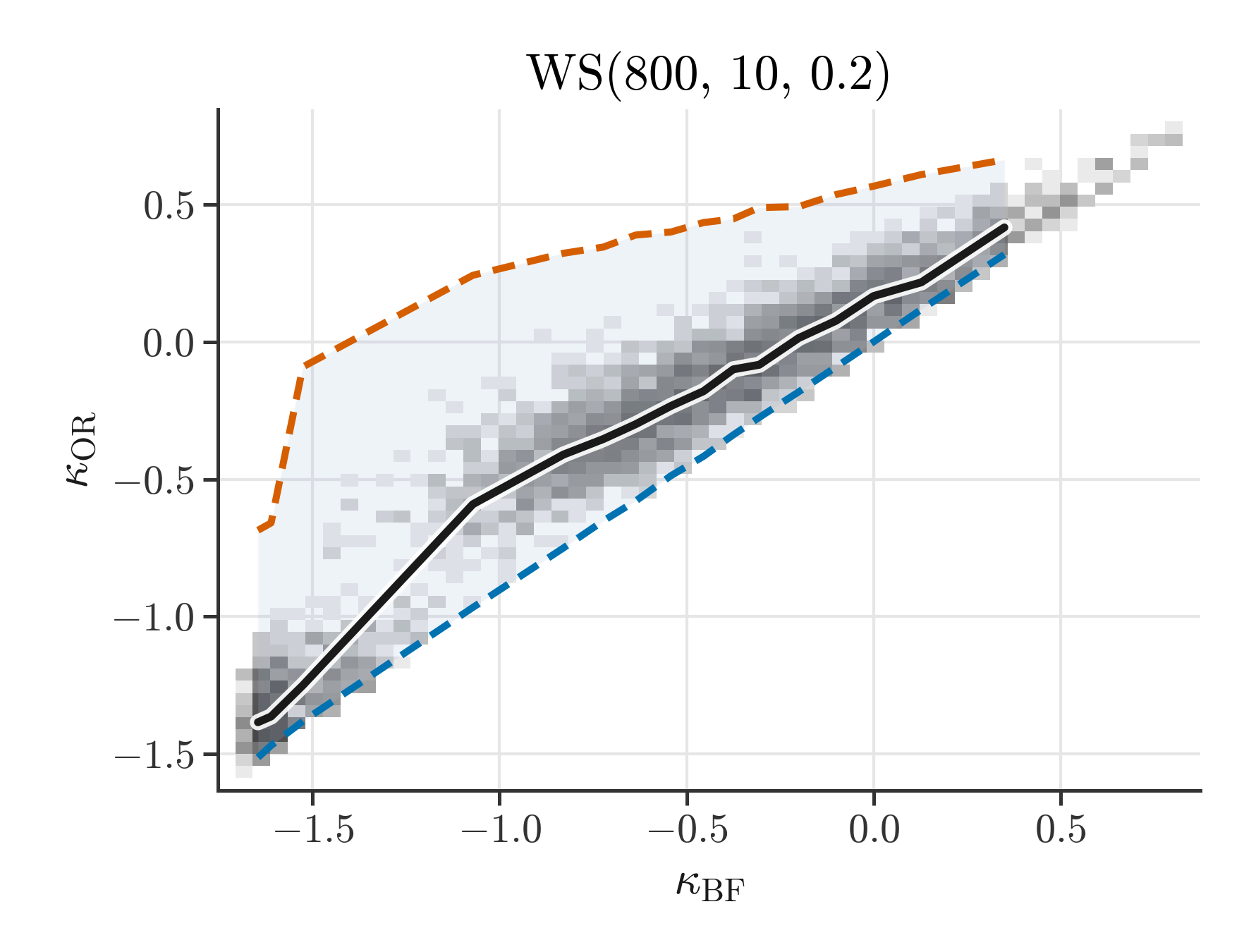}

\includegraphics[width=.43\textwidth]{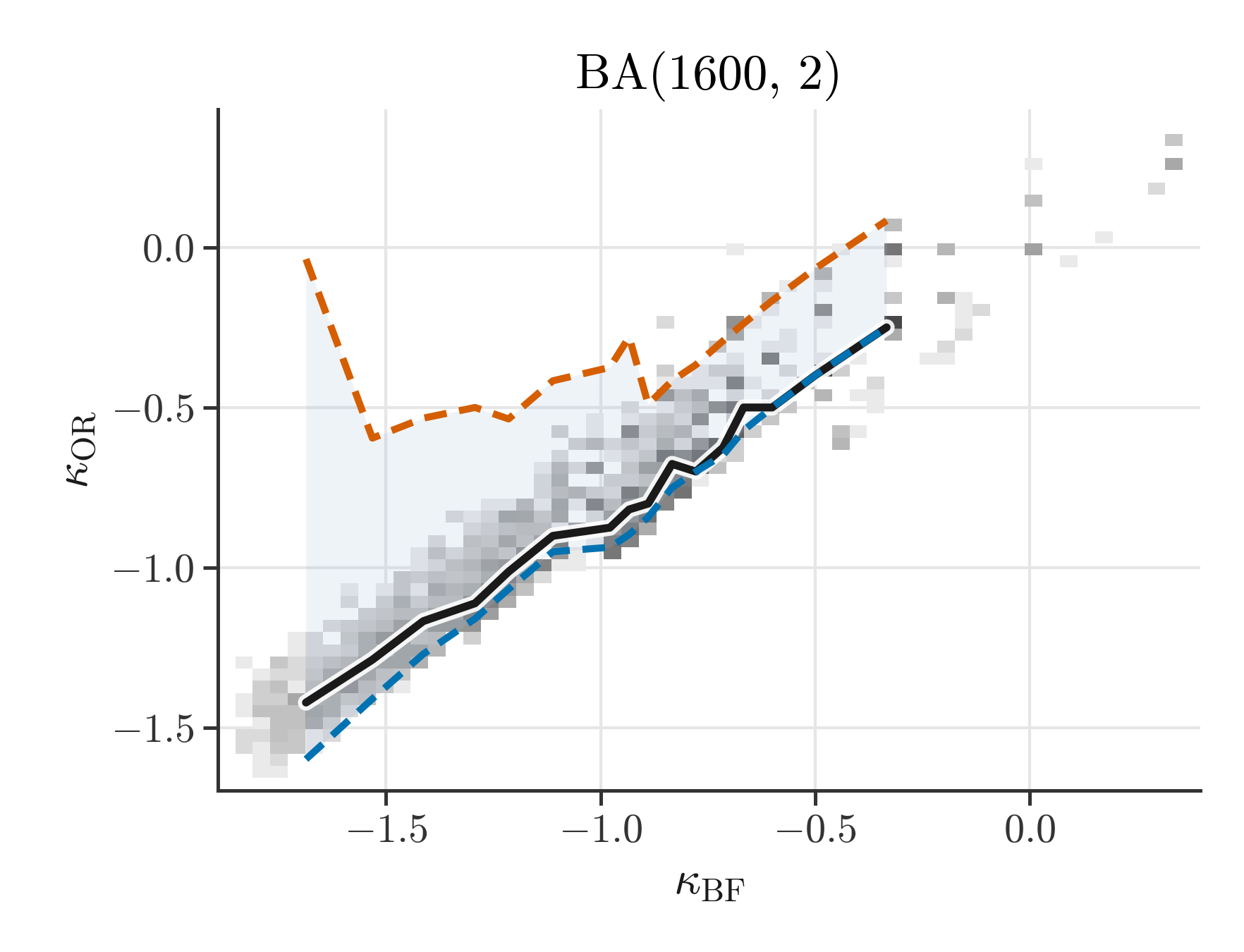}
\hfill
\includegraphics[width=.43\textwidth]{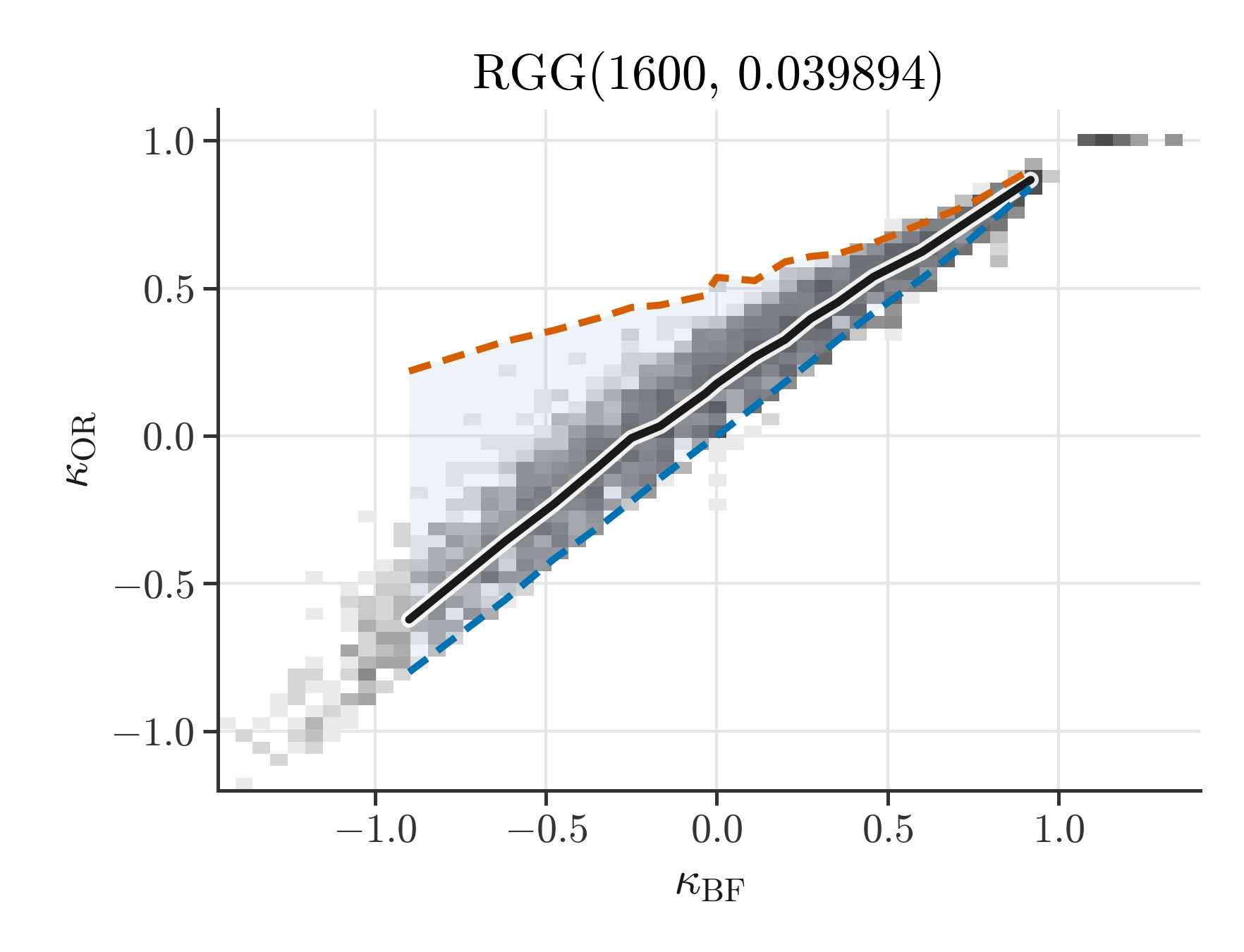}

\includegraphics[width=.43\textwidth]{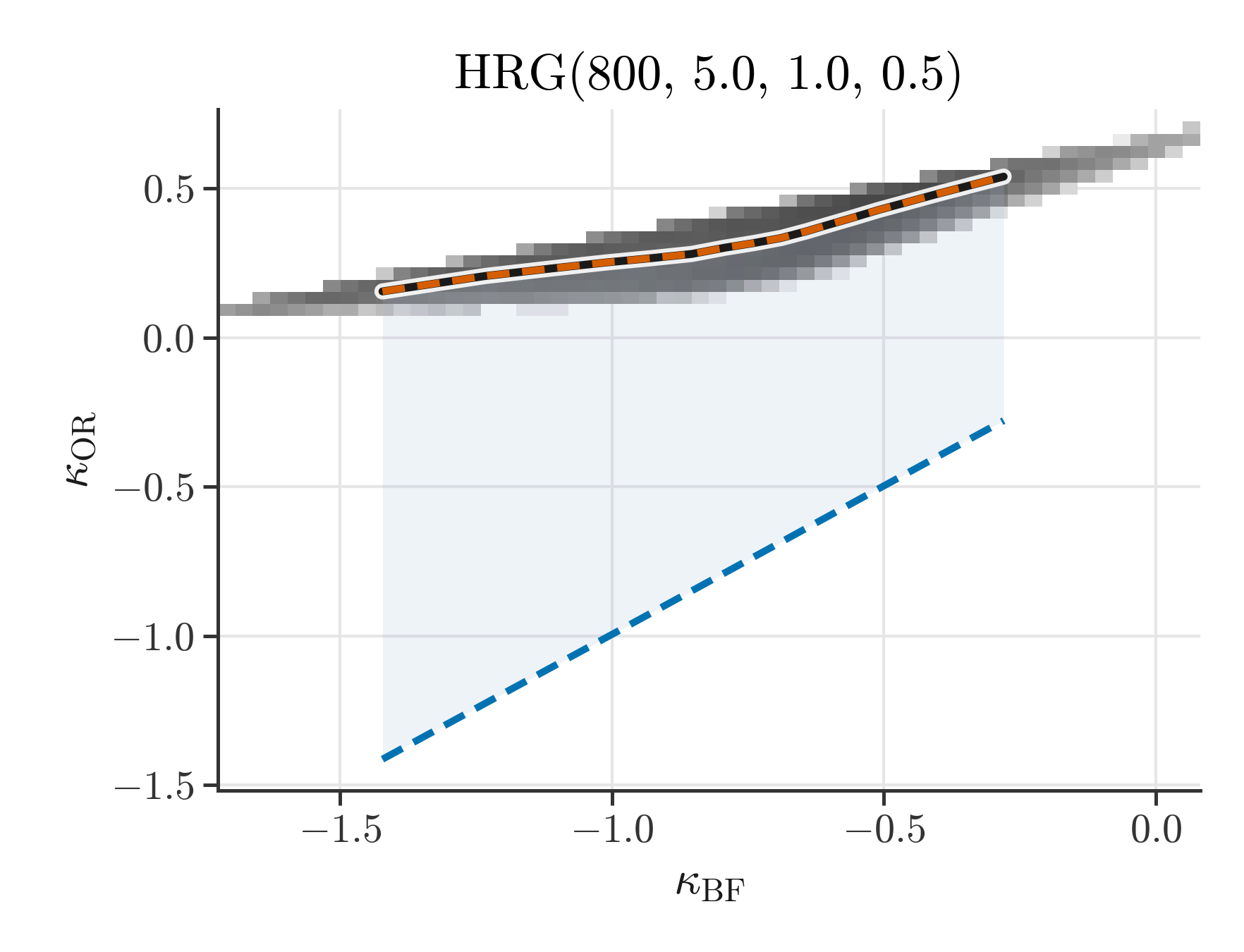}
\hfill
\includegraphics[width=.43\textwidth]{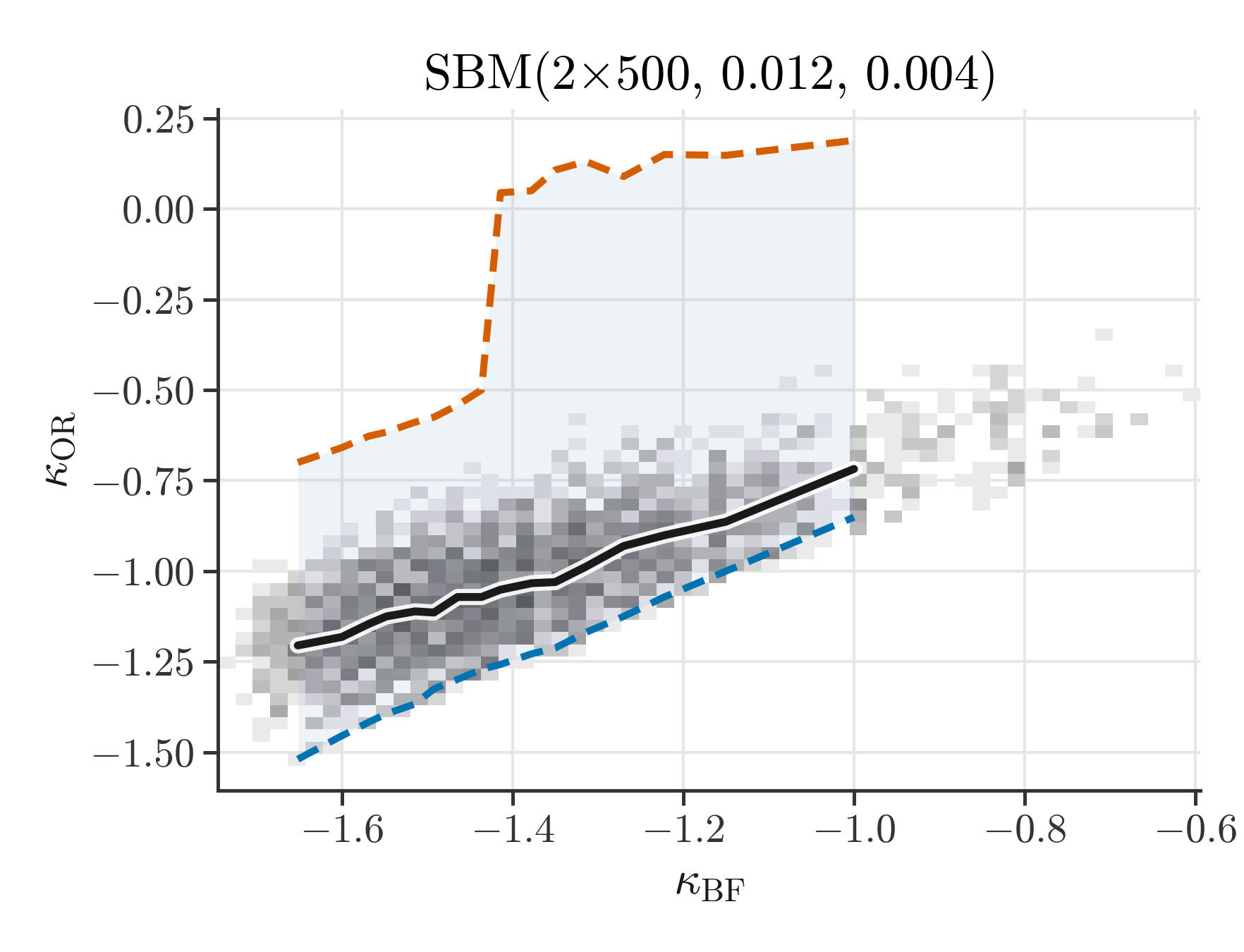}

\caption{
Density of edges in the $(\kappa_{\rm BF},\kappa_{\rm OR})$ plane (log scale; central 99.8\%); black: $\mathrm{median}[\kappa_{\rm OR}\mid \kappa_{\rm BF}]$; dashed blue (vermillion): conditional median of the lower (upper) end of the BF interval.}
\label{fig:scatters}
\end{figure}

\subsubsection{Distributional View}
Figure~\ref{fig:dists_sel} shows the histogram of $\kappa_{\rm OR}$ and its quantile profile for three graphs that span the regimes of Proposition~\ref{prop:bf-unclipped}. ER$(800, 0.0100125)$ (\res{FigEdges:ER} edges) is sparse, with scarce triangles and 4-cycles, and its BF deficits are nonnegative on \res{FigBdef:ER}\% of the edges, on which BF therefore lower-bounds the non-lazy curvature.
HRG$(800, 5, 1, 0)$ (\res{FigEdges:HRG} edges) is rich in triangles; its deficits are nonnegative on only \res{FigBdef:HRG}\% of the edges, the regime in which BF can overstate $\kappa_0$ (Remark~\ref{rem:bf-interpretation}).
Yeast (\res{FigEdges:Yeast} edges, deficits nonnegative on \res{FigBdef:Yeast}\%) has pronounced degree heterogeneity, and on its hub-incident edges, where the degree term $S$ dominates BF, the BF interval is widest.

\begin{figure}
\centering
\includegraphics[width=.43\textwidth]{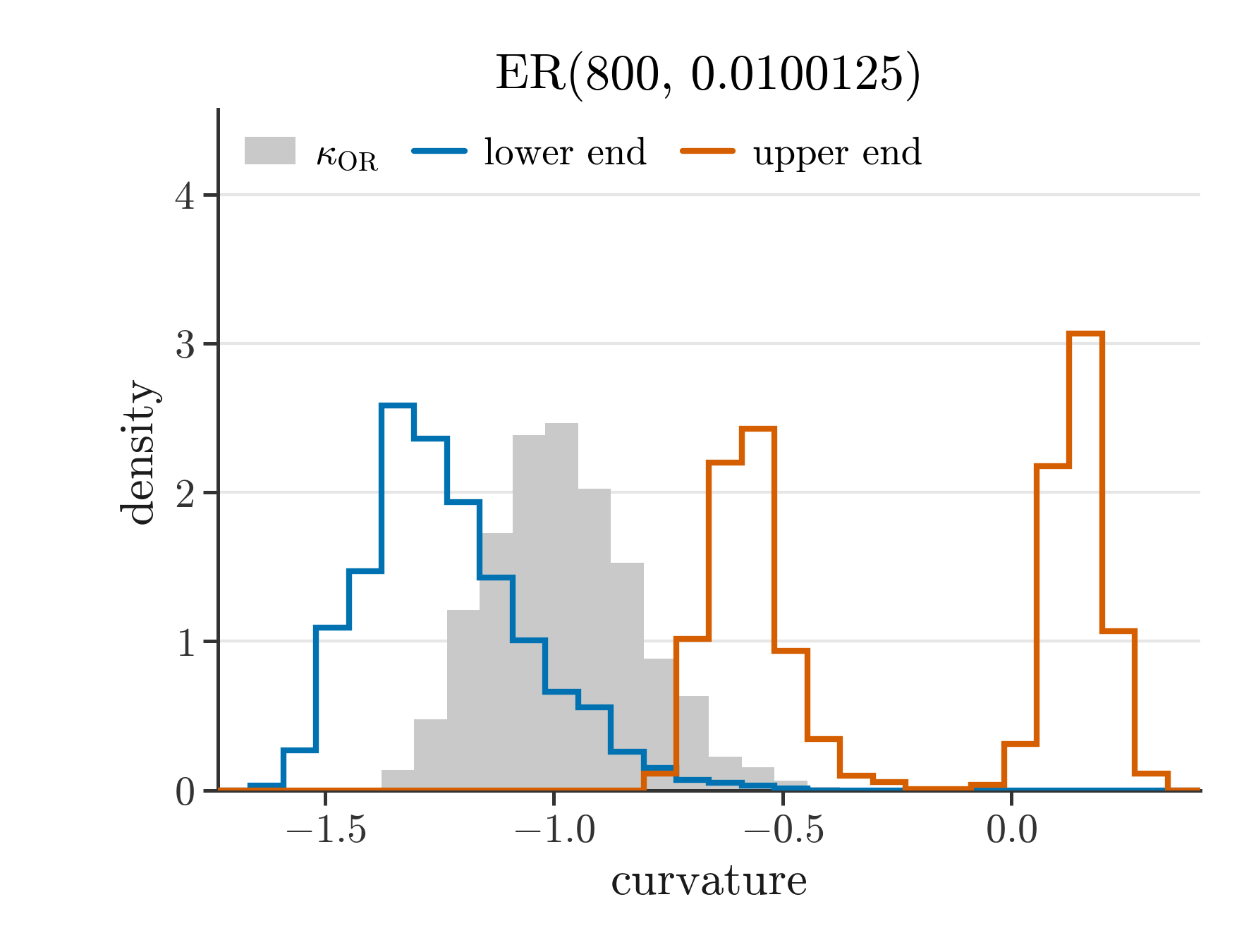}
\hfill
\includegraphics[width=.43\textwidth]{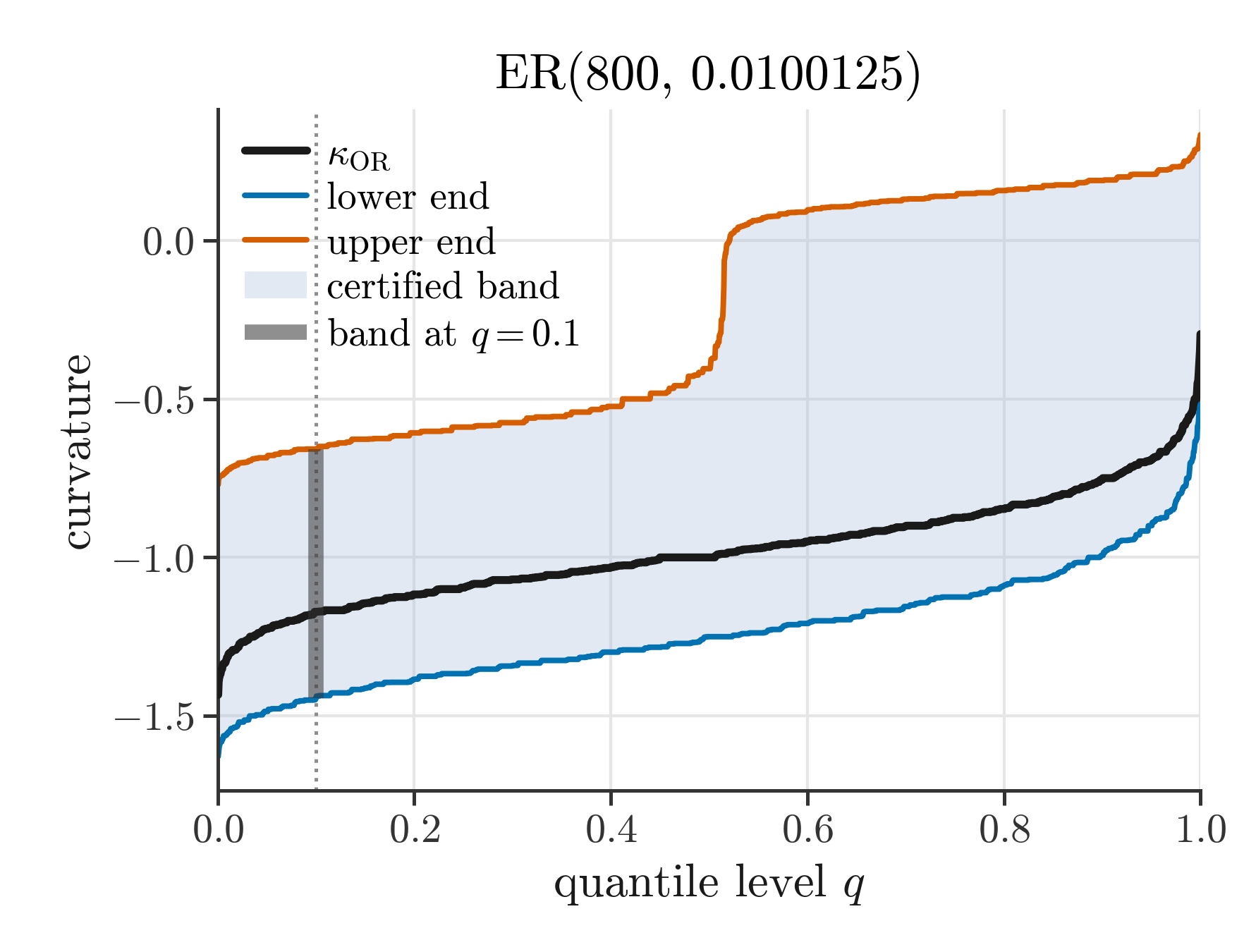}

\includegraphics[width=.43\textwidth]{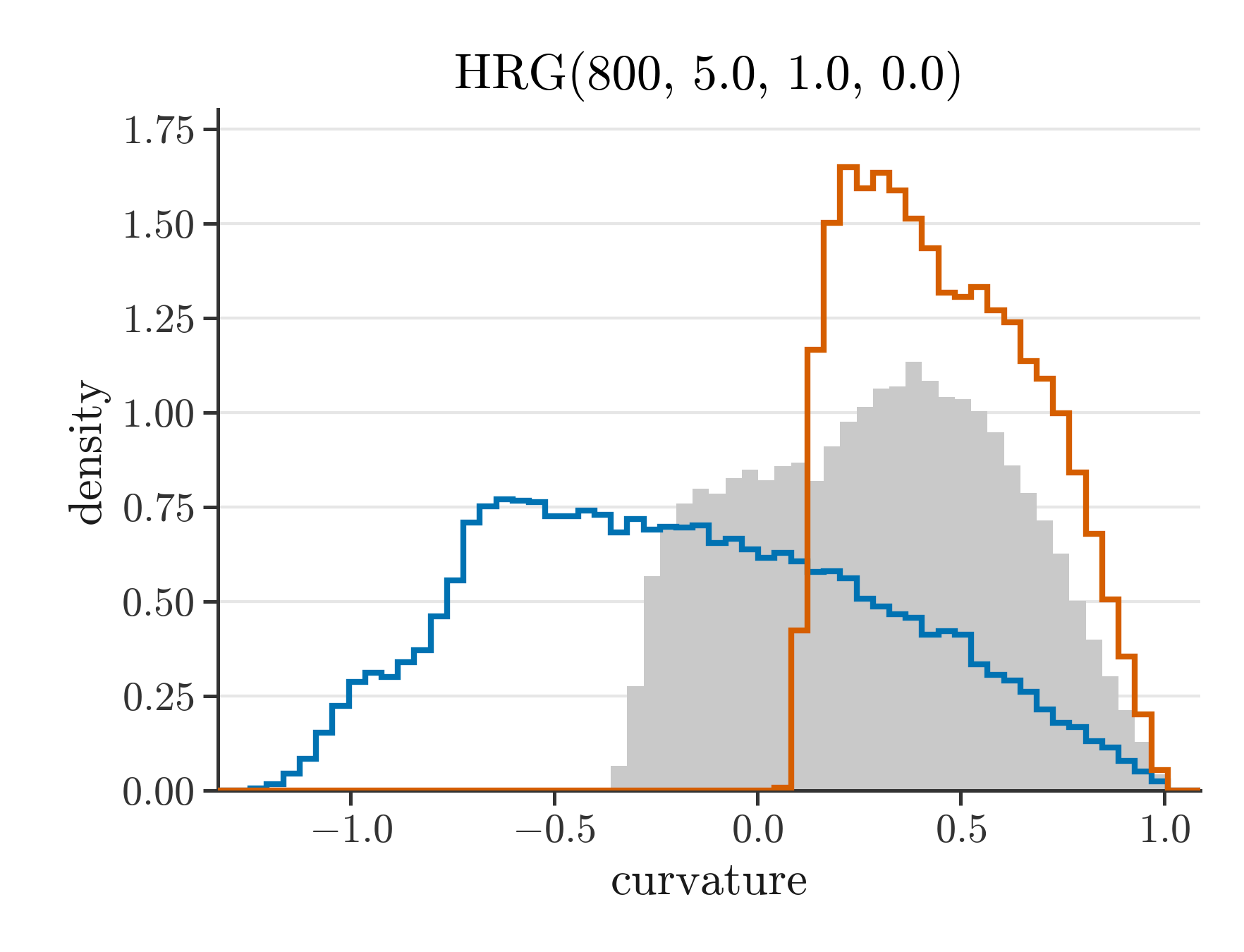}
\hfill
\includegraphics[width=.43\textwidth]{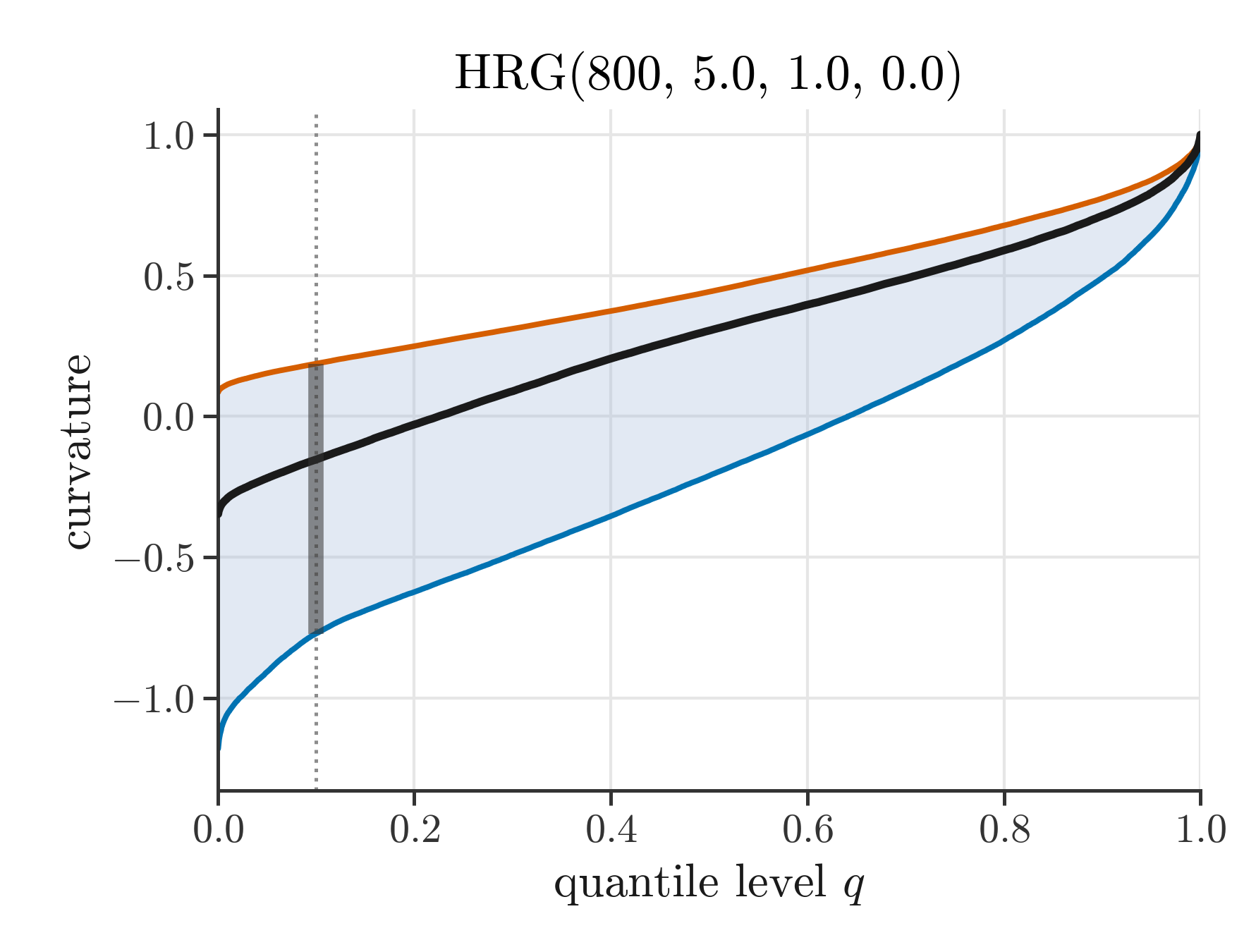}

\includegraphics[width=.43\textwidth]{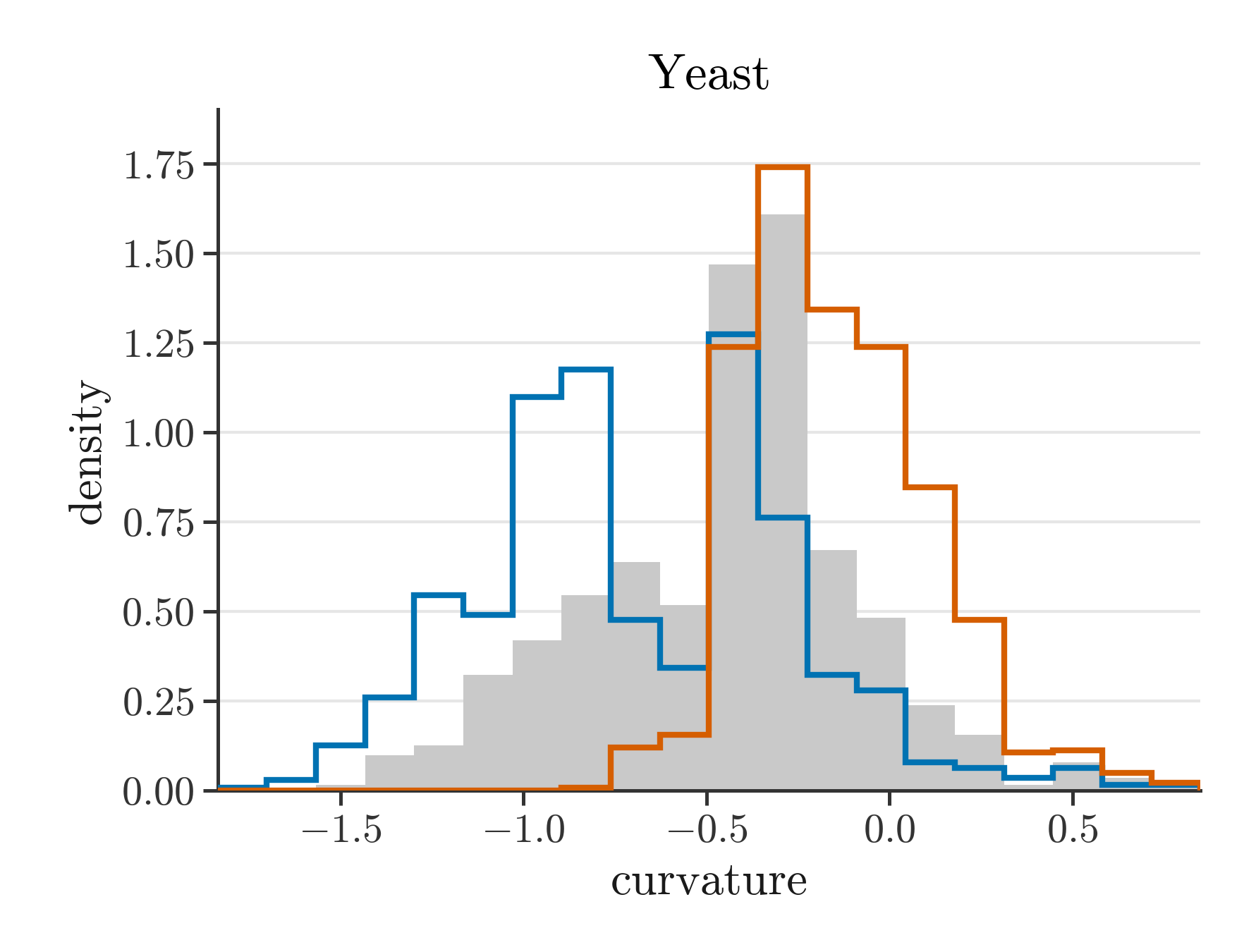}
\hfill
\includegraphics[width=.43\textwidth]{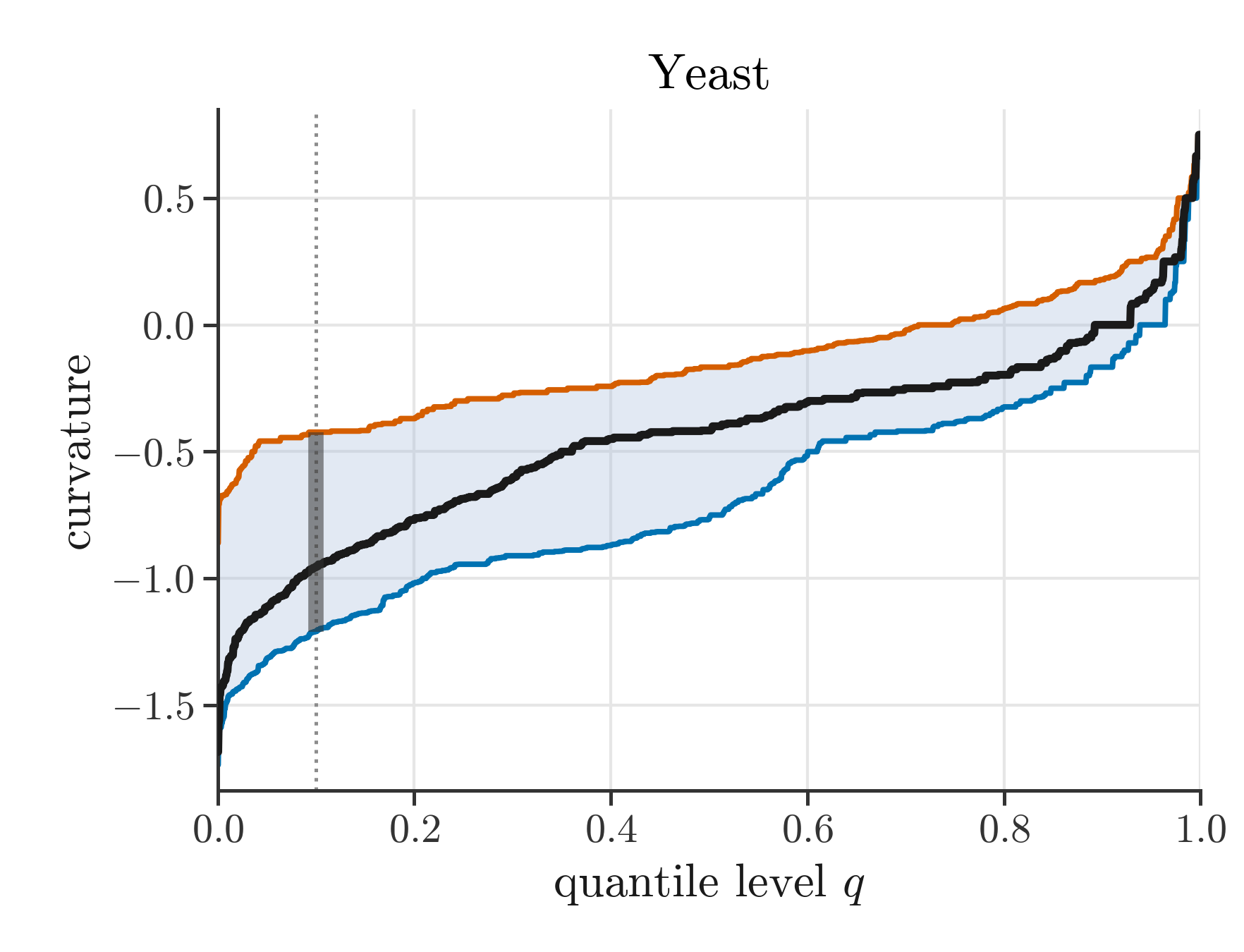}

\caption{
Left: histograms of $\kappa_{\rm OR}$ (filled) and of the BF interval (outlines).
Right: their quantiles; the band certifies each quantile of $\kappa_{\rm OR}$ (bar: $q=0.1$).}
\label{fig:dists_sel}
\end{figure}

The quantile profiles are certified: writing $\ell_{ij}\le\kappa_{\rm OR}(i,j)\le u_{ij}$ for the two ends and $F^{-1}_x(q)$ for the empirical $q$-quantile of a quantity $x$ over the edges, edgewise domination gives $F^{-1}_\ell(q)\le F^{-1}_{\kappa_{\rm OR}}(q)\le F^{-1}_u(q)$ for every $q\in(0,1]$, and the band between the outer profiles at level $q$ is a certified range for the $q$-quantile of OR curvature from BF alone. At $q=0.1$ this range bounds the curvature of the most negatively curved decile, the quantity that bottleneck detection targets; on ER$(800, 0.0100125)$ the upper profile is negative there, so BF alone certifies that at least a tenth of the edges are negatively curved.

Table~\ref{tab:bf-interval} shows what the 2-hop statistics add: on the sparse graphs the two lower ends nearly coincide and the gain lies in the upper end, whereas on HRG$(800,5,1,0.5)$, Arxiv, and Jazz, $L^{\rm lazy}_{ij}$ is strictly tighter than the lower end of the BF interval on most edges.

\subsection{Certified Screening}
\label{sec:screening}

A certified interval $[\ell_{ij},u_{ij}]$ answers the threshold query $\{\kappa_{\rm OR}<\vartheta\}$ on every edge with $u_{ij}<\vartheta$ or $\ell_{ij}\ge\vartheta$ without solving a transport problem, so only the remaining edges require the exact curvature. Since an exact value is an interval of zero width, certify-then-solve (Algorithm~\ref{alg:certify-then-solve}) evaluates certificates of increasing cost, each only on the edges that the previous ones leave undecided: the 2-hop interval, which is contained in the BF interval and needs no distance matrix; its intersection with the TLW certificate (Section~\ref{sec:bf-to-or}); and the exact solver. It returns the exact set $\{(i,j):\kappa_{\rm OR}(i,j)<\vartheta\}$, and its correctness is formally verified for any valid certificate in the second stage (Appendix~\ref{app:lean}).

\begin{algorithm}[t]
\caption{Certify-then-solve for the threshold set $\{(i,j)\in E:\kappa_{\rm OR}(i,j)<\vartheta\}$. Decisions taken from bounds use the intervals widened by $10^{-9}$ on each side (Section~\ref{sec:comparison}).}
\label{alg:certify-then-solve}
\begin{algorithmic}[1]
\Require graph $G=(V,E)$ and a threshold $\vartheta$
\For{every edge $(i,j)\in E$} \Comment{2-hop statistics and one matching, no distance matrix}
  \State $[\ell_{ij},u_{ij}]\gets[L^{\rm lazy}_{ij},U^{\rm lazy}_{ij}]$ \Comment{contained in $\mathcal I^{\rm BF}_{ij}$}
\EndFor
\For{every edge $(i,j)$ with $\ell_{ij}<\vartheta\le u_{ij}$}
  \State $u_{ij}\gets\min\{u_{ij},U^{\rm prior}_{ij}\}$ \Comment{one more pass over the neighbor lists, no distance matrix}
  \If{$\ell_{ij}<\vartheta\le u_{ij}$}
    \State solve the transport problem exactly \citep{micaletto_exact_2026} and set $\ell_{ij},u_{ij}\gets\kappa_{\rm OR}(i,j)$
  \EndIf
\EndFor
\State \Return $\{(i,j)\in E:u_{ij}<\vartheta\}$
\end{algorithmic}
\end{algorithm}

On the graphs of Section~\ref{sec:comparison}, the decisions follow from the stored per-edge intervals: for $\{\kappa_{\rm OR}<0\}$ the sharpened interval avoids \res{ScreenBenchSignSharpMin}--\res{ScreenBenchSignSharpMax}\% of the transport problems and the BF interval alone \res{ScreenBenchSignBFMin}--\res{ScreenBenchSignBFMax}\%, never fewer than the certified signs of Tables~\ref{tab:bound-comparison-lazy} and~\ref{tab:bf-interval}. Figure~\ref{fig:screening-map} follows the cascade on one of them, RGG$(800, 0.056419)$. The certificates divide the work by local structure: the BF interval alone decides \res{MapBF}\% of the \res{MapEdges} edges, \res{MapBFPos}\% of them through its lower end. The 2-hop interval raises the share to \res{MapTwoHop}\% (where \res{MapTwoHopNeg}\% of the edges added are negatively curved). The TLW certificate, which can only lower the upper end, certifies a negative sign on a further \res{MapTLWGain}\%, for \res{MapSharp}\% in total. The \res{MapSolved} transport problems that remain (\res{MapSolvedPct}\% of the edges, \res{MapWorkRatio}$\times$ fewer than exact evaluation) lie on loosely knit edges between the two regimes, where the curvature is close to the threshold: their median $|\kappa_{\rm OR}|$ is \res{MapMedGapSolved}, against \res{MapMedGapDecided} on the edges decided without one.

\begin{figure}
\centering
\includegraphics[width=\textwidth]{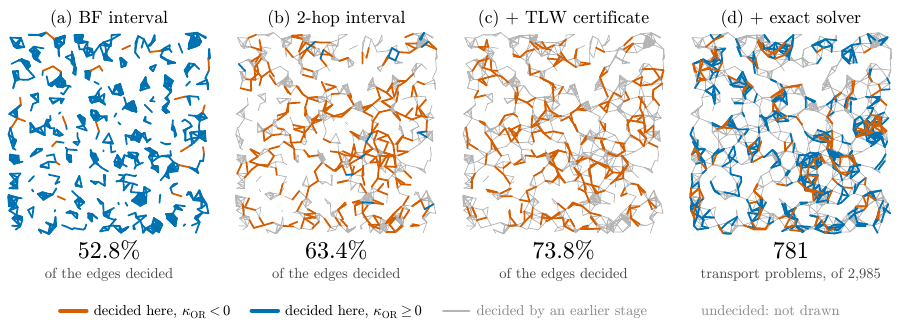}
\caption{Certify-then-solve for $\{\kappa_{\rm OR}<0\}$ on RGG$(800, 0.056419)$. Each panel adds one certificate of Algorithm~\ref{alg:certify-then-solve}: (a)~the BF interval alone (Corollary~\ref{cor:bf-interval}), which the algorithm does not use; (b)~the 2-hop interval of its first stage; (c)~its intersection with $U^{\rm prior}_{ij}$; and (d)~the exact solver. In each panel, the edges that its certificate decides are vermillion (blue) where $\kappa_{\rm OR}<0$ ($\kappa_{\rm OR}\ge0$), which the upper (lower) end of an interval or, in (d), the solver decides. Gray edges were decided by an earlier certificate, and undecided edges are not drawn. Below each panel: the share of edges decided, and in (d) the number of transport problems solved.}
\label{fig:screening-map}
\end{figure}
Because exact evaluation of every edge of these graphs takes at most a few minutes, and milliseconds on most of them, we measure wall-clock gains on \res{ScreenLargeCount} real networks with \res{ScreenLargeMinEdges} to \res{ScreenLargeMaxEdges} edges (Table~\ref{tab:screening}): the Deezer, Facebook, Twitch-DE, and GitHub social networks \citep{rozemberczki_characteristic_2020,rozemberczki_multiscale_2021}, the Crocodile and Squirrel Wikipedia networks \citep{rozemberczki_multiscale_2021}, the Enron email network \citep{leskovec_community_2009}, the AstroPh, Coauthor-Physics, and DBLP collaboration networks \citep{leskovec_graph_2007,shchur_pitfalls_2018,yang_defining_2012}, the Brightkite friendship network \citep{cho_friendship_2011}, the Epinions trust network \citep{richardson_trust_2003}, and the Amazon, Amazon-Photo, and Amazon-Computers co-purchasing networks \citep{yang_defining_2012,shchur_pitfalls_2018}, obtained from public collections \citep{leskovec_snap_2014} and from the repositories that accompany the cited works. They range from sparse networks to networks with large hubs, on which the transport problems are larger, and their maximum degrees keep exact evaluation of every edge feasible.

\begin{table}
\centering
\caption{Certify-then-solve for $\{\kappa_{\rm OR}<0\}$ on large real networks: number of edges, maximum degree, and wall-clock time of exact evaluation of every edge; percentage of edges decided by the BF interval alone; percentage of transport problems avoided; and speedup (time of total exact evaluation divided by that of certify-then-solve). Every answer coincides with the one computed from exact curvatures.}
\label{tab:screening}
\begin{adjustbox}{max width=\textwidth}
\footnotesize
\setlength{\tabcolsep}{4pt}
\begin{tabular}{l rrr cc c}
\toprule
Graph & $|E|$ & $d_{\max}$ & Exact (s) & BF (\%) & Avoided (\%) & Speedup \\
\midrule
\resrows{screening}
\bottomrule
\end{tabular}
\end{adjustbox}
\end{table}

On every network, the answer coincides with the one computed from exact curvatures. Certify-then-solve avoids \res{ScreenSignAvoidMin}\% (\res{ScreenSignAvoidMinGraph}) to \res{ScreenSignAvoidMax}\% (\res{ScreenSignAvoidMaxGraph}) of the transport problems, so it solves \res{ScreenSignWorkRatioMedian} times fewer of them than exact evaluation in the median, and up to \res{ScreenSignWorkRatioMax} times fewer; the BF interval alone decides up to \res{ScreenSignShareBFMax}\% of the edges, and the TLW certificate a further \res{ScreenSignShareTLWMin}--\res{ScreenSignShareTLWMax}\%. This saving depends only on the intervals. The wall-clock speedup also depends on the solver: it reaches \res{ScreenSignSpeedupMax}$\times$ on \res{ScreenSignSpeedupMaxGraph}, where transport problems are large and the intervals decide most edges, and its median is \res{ScreenSignSpeedupMedian}$\times$, because next to so fast a solver the bounds are not free: the first stage alone costs up to \res{ScreenStageOneShareMax}\% of the time of exact evaluation. It falls below one where transport problems are too small to pay for the bounds, and the cost of the bounds is all that the procedure can lose. Against a slower solver, it would approach the reduction in transport problems. All times are medians of up to three runs in one process on \res{ScreenThreads} threads after just-in-time compilation, with the edges in a fixed random order so that the expensive edges at hubs are spread over the threads.

\subsection{Runtime}
\label{sec:runtime}

\begin{figure}
\centering
\includegraphics[width=.49\textwidth]{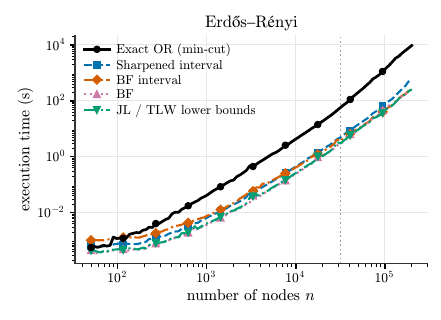}\hfill
\includegraphics[width=.49\textwidth]{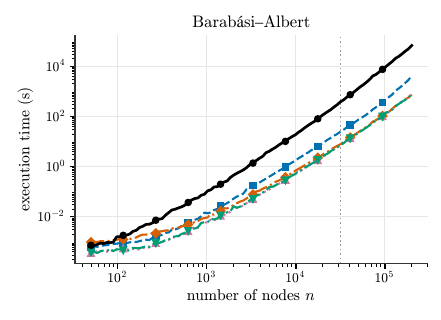}

\includegraphics[width=.49\textwidth]{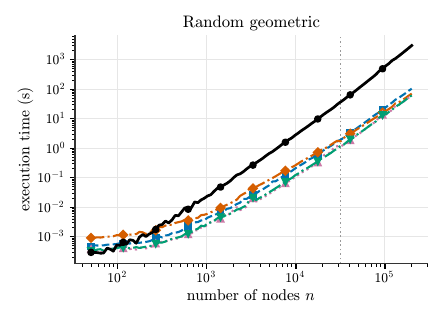}\hfill
\includegraphics[width=.49\textwidth]{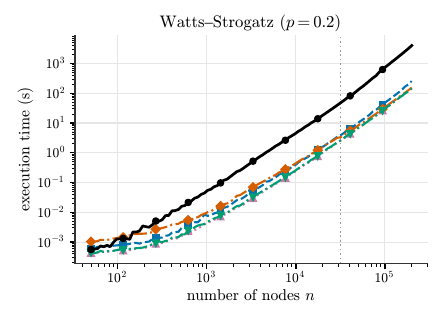}
\caption{Log-log comparison of the total execution time of exact $\kappa_{\rm OR}$ and of the bounds on four random graph families whose average degree grows as $2.1\,n^{0.4}$. Series: $\kappa_{\rm OR}$, the interval of Section~\ref{sec:comparison}, the BF interval, $\kappa_{\rm BF}$, and the prior bounds.
The dashed vertical line marks the size from which times are measured in chunks.}
\label{fig:runtime_scaling}
\end{figure}

\begin{table}
\centering
\caption{Runtime summary per family (Figure~\ref{fig:runtime_scaling}). Largest size $n$ and its number of edges; range over $n\ge1000$ of the ratio of the exact time to the time of the sharpened and of the BF interval; times in seconds at the largest size.}
\label{tab:runtime}
\begin{adjustbox}{max width=\textwidth}
\footnotesize
\setlength{\tabcolsep}{4pt}
\begin{tabular}{l  r r cc rrr}
\toprule
& & & \multicolumn{2}{c}{Exact / interval} & \multicolumn{3}{c}{Time at largest $n$ (s)} \\
\cmidrule(lr){4-5} \cmidrule(lr){6-8}
Family & $n$ & $|E|$ & Sharpened & BF & BF interval & Sharpened & Exact \\
\midrule
\resrows{runtime}
\bottomrule
\end{tabular}
\end{adjustbox}
\end{table}

Figure~\ref{fig:runtime_scaling} and Table~\ref{tab:runtime} report total edgewise execution times on \res{BenchNumGraphs} logarithmically spaced sizes from $n=\res{BenchMinN}$ to $n=\res{BenchMaxN}$ in each of four families, up to \res{BenchMaxEdges} edges in ER: ER, BA, RGG in the unit square, and WS with rewiring probability $0.2$. To separate size from density and keep the families comparable, all four have average degree $\bar d(n)\approx 2.1\,n^{0.4}$: ER uses $p=\bar d/(n-1)$ (so $p\propto n^{-0.6}$), BA attaches $\mathrm{round}(\bar d/2)$ edges per new node, RGG uses the radius $\sqrt{\bar d/(\pi(n-1))}$, and WS a ring lattice of degree $2\,\mathrm{round}(\bar d/2)$. ER has homogeneous degrees, so its per-edge costs are the most uniform; BA has hubs whose degree grows with $n$; RGG and WS are strongly clustered. All computations run in one process on \res{BenchThreadsText} after the just-in-time compilation of every kernel, with the edges dealt to the threads in cyclic blocks so that the expensive edges at hubs are spread over them. \res{BenchChunkSentence}

On ER graphs, from $n=\res{BenchRatioFromN}$ on, exact evaluation is \res{BenchRatioSharpRange} times slower than the sharpened interval and \res{BenchRatioBFRange} times slower than the BF interval; at $n=\res{BenchBothLastN}$, the largest size at which both were measured, it takes \res{BenchExactAtMax}\,s, against \res{BenchSharpAtMax}\,s for the sharpened interval. The gap is widest on BA graphs, whose hubs carry the largest transport problems: there exact evaluation is up to \res{BenchRatioBFMax:BA} times slower than the BF interval, and at $n=\res{BenchBothLastN:BA}$ it takes \res{BenchExactAtMax:BA}\,s, against \res{BenchBFAtMax:BA}\,s for the BF interval. Table~\ref{tab:runtime} gives the ratios for all families.

\section{Discussion}
\label{sec:conclusion}
We have shown that the BF value of an edge certifies a two-sided interval for its lazy OR curvature, and that the interval can be sharpened when the 2-hop neighborhood is available. Both ends come from explicit transport arguments built on the one-sided structure of the lazy problem, and along the way we sharpened the JL lower bound with the matching number of the cross-edge graph and identified BF as an unclipped form of it. 
Empirically, the sharpened interval is the tightest combinatorial certificate in our comparison: it settles the sign of most edges on most graphs, ranks edges more faithfully than BF, and lets certify-then-solve avoid most transport problems, although next to the fastest exact solver we know of this saves only a modest share of the time.

\subsection{Limitations}
The proposed bounds and moduli have several boundaries.

\begin{description}
\item[Graph:] The results are restricted to unweighted, undirected, simple graphs. Extending them to weighted or directed graphs requires defining a weighted (or directed) analogue of BF curvature and re-deriving the couplings behind the bounds.\footnote{TLW have already extended the JL bounds to weighted graphs.}

\item[Idleness:] Both moduli use Lemma~\ref{lem:one-sided}; for other idleness values, Propositions~\ref{prop:mixing-transfer} and \ref{prop:lazy-envelope-sharpened} still apply, but the relationship is limited to Proposition~\ref{prop:bf-unclipped}. Extending the moduli to other idleness values, or to non-uniform laws such as $\nu_u(v)\propto g(d_v)$, remains open.

\item[{Tradeoffs:}] The intervals can be wide, and the upper envelope is usually looser than TLW due to its definition (Table~\ref{tab:bound-comparison-lazy}). Intersecting the two has the same asymptotic cost but a larger constant, whence Algorithm~\ref{alg:certify-then-solve} evaluates TLW only on undecided edges.

\item[{Locality:}] The bounds are strictly edgewise and local. 
They do not control Ricci flows or govern global geometric invariants such as hyperbolicity and isoperimetric profiles.

\item[{Downstream use:}] We evaluate the intervals as certificates, not inside learning pipelines, and we do not measure the downstream effect of curvature-based rewiring, which has been reported to depend strongly on hyperparameter choices \citep{tori_effectiveness_2024}.
\end{description}

\subsection{Open Problems}
Several questions and potential improvements remain open for future work.

\begin{description}

\item[Exactness:] Characterizing the edges on which the lower modulus is exact remains open. Theorem~\ref{thm:bf-lower} gives some conditions, and Proposition~\ref{prop:bf-unclipped} indicates that the general answer involves the BF deficits and the gap between the true matching number and its 4-cycle proxy.

\item[{Asymptotics:}] While the limiting behavior of OR curvature on random geometric graphs is known \citep{vanderhoorn_ollivier_2021}, an analogous limit theorem for BF curvature does not yet exist.
Establishing such a theorem would determine whether $\varphi_{a,b}$ and $\psi_{a,b}$ converge and whether the certified interval collapses in the asymptotic limit.

\end{description}

\section*{Code Availability}
The code needed to reproduce the experiments of Section~\ref{sec:analysis} is available at \url{https://github.com/GiorgioMB/Curvature-Transfer-Code}. The Lean formalization (Appendix~\ref{app:lean}) is available at \url{https://github.com/GiorgioMB/BF-OR-Lean}.

{\acks{This work received no specific funding. During the preparation of this manuscript and the accompanying software repository, the authors used Claude Opus 5.5 to refine the grammar of the manuscript and to optimize the implementation of Algorithm~\ref{alg:certify-then-solve}. All core conceptual arguments, mathematical proofs, and conclusions remain entirely the original work of the authors, and the authors retain full responsibility for the accuracy, validity, and integrity of both the written text and the final performance metrics reported in this work. The optimized code and experimental workflows were manually verified and evaluated by the human authors to ensure technical correctness.}}

\newpage
\appendix
\section{Comparison with Bhattacharya and Mukherjee}
\label{app:bm-comparison}
As both bounds are symmetric in $i$ and $j$, assume $d_i\le d_j$, so that $|\mathcal U_i|=a-1-t_{ij}\le|\mathcal U_j|$, and recall that $M_{ij}\le|\mathcal U_i|$, that $L^{\square}_{ij}=t_{ij}/b-[A^\square_{ij}]_+-[B^\square_{ij}]_+$, and that $A^\square_{ij}-B^\square_{ij}=t_{ij}(1/a-1/b)\ge 0$.

A matching between $\mathcal Q_i=\{j\}\cup\mathcal U_i$ and $\mathcal Q_j=\{i\}\cup\mathcal U_j$ uses at most one edge at $j$ and at most one at $i$, and its other edges form a matching of $H_{ij}$; hence $M^{\rm BM}_{ij}\le\min\{|\mathcal U_i|+1,\,|\mathcal U_j|+1,\,M_{ij}+2\}$. Conversely, a maximum matching of $H_{ij}$ can be extended by $(j,y)$ and $(x,i)$ for unmatched $x\in\mathcal U_i$, $y\in\mathcal U_j$ when $M_{ij}<|\mathcal U_i|$, and by $(j,y)$ or $(j,i)$ when $M_{ij}=|\mathcal U_i|$. Hence $M^{\rm BM}_{ij}=M_{ij}+2$ if $M_{ij}<|\mathcal U_i|$, and $M^{\rm BM}_{ij}=|\mathcal U_i|+1=a-t_{ij}$ if $M_{ij}=|\mathcal U_i|$.

\begin{description}
\item[$M_{ij}<|\mathcal U_i|$] $M^{\rm BM}_{ij}=M_{ij}+2$ and, since $M_{ij}\le a-t_{ij}-2$,
\[
B^\square_{ij}\ \ge\ 1-\frac{1+t_{ij}}{a}+\frac{1+t_{ij}-a}{b}\ =\ |\mathcal U_i|\left(\frac1a-\frac1b\right)\ \ge\ 0,
\]
so $L^{\square}_{ij}-L^{\rm BM}_{ij}=(2+t_{ij})(1/a-1/b)\ge 0$, with equality if and only if $a=b$.

\item[$M_{ij}=|\mathcal U_i|$] $M^{\rm BM}_{ij}=a-t_{ij}$, so $L^{\rm BM}_{ij}=t_{ij}/b-2+2a/b$, $A^\square_{ij}=1-1/a-a/b$, and $B^\square_{ij}=A^\square_{ij}-t_{ij}(1/a-1/b)$. If $B^\square_{ij}\ge 0$, then $L^{\square}_{ij}-L^{\rm BM}_{ij}=2/a+t_{ij}(1/a-1/b)>0$. If $B^\square_{ij}<0\le A^\square_{ij}$, the difference is $1+1/a-a/b>0$. If $A^\square_{ij}<0$, then $L^{\square}_{ij}=t_{ij}/b$ and the difference is $2(1-a/b)\ge 0$, with equality if and only if $a=b$ (which falls in the last subcase as $A^\square_{ij}=-1/a$).
\end{description}
In both cases $L^{\square}_{ij}\ge L^{\rm BM}_{ij}$, with equality if and only if $a=b$. On a tree edge with $a\ge 2$ we have $t_{ij}=M_{ij}=0$ and $M^{\rm BM}_{ij}=2$, which gives the two values quoted in Remark~\ref{rem:bm-comparison}.
\hfill\BlackBox

\section{An Envelope for Arbitrary Idleness}
\label{app:envelope}
For idleness $\alpha_i,\alpha_j\in[0,1)$, let $m_u:=\alpha_u\delta_u+(1-\alpha_u)\nu_u$ for $u\in\{i,j\}$; Definition~\ref{def:lazy-rw} is the case $\alpha_u=1/(d_u+1)$, and $\alpha_i=\alpha_j=0$ gives $\kappa_0$. Write $w_u:=(1-\alpha_u)/d_u$ for the mass that $m_u$ puts on each neighbor of $u$ (for the lazy measure this is $w_u=1/(d_u+1)$, as in Definition~\ref{def:lazy-rw}) and $w_{\min}:=\min\{w_i,w_j\}$, and let $z_i:=\min\{\alpha_i,w_j\}$, $r_i:=[\alpha_i-w_j]_+$, and $\bar r_i:=[w_j-\alpha_i]_+$ be the mass that $m_i$ and $m_j$ share at $i$ and the residual supply and demand there, and symmetrically $z_j$, $r_j$, and $\bar r_j$.

\begin{proposition}[Envelope for Arbitrary Idleness]
\label{prop:lazy-envelope-sharpened}
$\forall (i,j)\in E$ and $\forall \alpha_i,\alpha_j\in[0,1)$,
\begin{equation}
\label{eq:OR-master-corrected-sharp}
1-W_1(m_i,m_j) \le U^{\rm env}_{ij}:=-1+2(z_i+z_j)+(r_i+\bar r_i+r_j+\bar r_j)+2t_{ij}\,w_{\min}+\widetilde m^{\rm UU}_{ij}+\widetilde m^{\rm CU}_{ij},
\end{equation}
where
\begin{align*}
\widetilde m^{\rm UU}_{ij}&:=\min\left\{|\mathcal U_i|\,w_i,\ |\mathcal U_j|\,w_j,\ \frac{\Xi_{ij}}{1/w_i+1/w_j}\right\},\\
\widetilde m^{\rm CU}_{ij}&:=\min\left\{t_{ij}\left|w_i-w_j\right|,\ |\mathcal U_i|\,w_i+|\mathcal U_j|\,w_j\right\}.
\end{align*}
For the lazy measure, $U^{\rm lazy}_{ij}\le U^{\rm env}_{ij}$, with $U^{\rm lazy}_{ij}$ as in Proposition~\ref{prop:lazy-envelope-tight}.
\end{proposition}

\begin{remark}[Scope of the Envelope]
\label{rem:envelope-role}
At zero idleness, $U^{\rm env}_{ij}$ improves on the JL upper bound $t_{ij}/b$ on most edges of the random graph models we consider (Table~\ref{tab:bound-comparison}), whereas in the lazy setting the dual bound of TLW is tighter on most edges (Table~\ref{tab:bound-comparison-lazy}). For the lazy measure, the endpoint term of $U^{\rm lazy}_{ij}$ is $\Delta$ instead of $2\Delta$, because the residual at the lower-degree endpoint lies at distance $2$ from all demand, and its cross-edge term uses the matching number.
\end{remark}

\begin{proof}[Proof of Proposition~\ref{prop:lazy-envelope-sharpened}]
Let $\mu^+:=(m_i-m_j)_+$ and $\mu^-:=(m_i-m_j)_-$, which have disjoint supports and common mass $\tau:=1-\sum_k\min\{m_i(k),m_j(k)\}$. Since $W_1(m_i,m_j)$ depends only on $m_i-m_j$ \citep[Corollary~1.16]{villani_topics_2003}, the coupling that keeps $\min\{m_i(k),m_j(k)\}$ in place at every $k$ and moves $\mu^+$ to $\mu^-$ by an optimal coupling $\pi'$ is optimal. Let $\eta$ be the mass that $\pi'$ moves along edges. Every unit moved by $\pi'$ travels a distance of at least $1$, and of at least $2$ unless it moves along an edge, so $W_1(m_i,m_j)\ge \eta+2(\tau-\eta)$, that is,
\[
1-W_1(m_i,m_j)\ \le\ -1+2(1-\tau)+\eta,\qquad 1-\tau=z_i+z_j+t_{ij}\,w_{\min},
\]
because $\min\{m_i(k),m_j(k)\}$ vanishes outside $\{i,j\}\cup\mathcal C$. Each pair moved along an edge by $\pi'$ goes from $\mathrm{supp}\,\mu^+$ to $\mathrm{supp}\,\mu^-$, and we bound the mass on four classes of such pairs separately.
\begin{itemize}
\item[(end)] Pairs containing an endpoint $u$ carry at most $\mu^+(u)+\mu^-(u)=|m_i(u)-m_j(u)|=r_u+\bar r_u$.
\item[(CC)] Every $k\in\mathcal C$ has $\mu^+(k)=[w_i-w_j]_+$ and $\mu^-(k)=[w_j-w_i]_+$, so $\mu^+$ and $\mu^-$ do not both charge $\mathcal C$, and no mass moves inside $\mathcal C$.
\item[(UU)] As $m_j$ vanishes on $\mathcal U_i$ and $m_i$ on $\mathcal U_j$, we have $\mu^+=w_i$ on $\mathcal U_i$ and $\mu^-=w_j$ on $\mathcal U_j$, and mass moved along cross edges leaves $\xi_i$ and enters $\xi_j$. With $s_i:=|\xi_i|\,w_i$ and $s_j:=|\xi_j|\,w_j$, this mass is at most
\[
\min\{s_i,s_j\}\ \le\ \frac{s_i/w_i+s_j/w_j}{1/w_i+1/w_j}\ =\ \frac{\Xi_{ij}}{1/w_i+1/w_j},
\]
and at most $\min\{|\mathcal U_i|\,w_i,|\mathcal U_j|\,w_j\}$, hence at most $\widetilde m^{\rm UU}_{ij}$.
\item[(CU)] Each pair between a unique and a common neighbor has one vertex in $\mathcal C$ and one in $\mathcal U_i\cup\mathcal U_j$, so the mass on such pairs is at most $\min\{(\mu^++\mu^-)(\mathcal C),(\mu^++\mu^-)(\mathcal U_i\cup\mathcal U_j)\}=\widetilde m^{\rm CU}_{ij}$.
\end{itemize}
Adding the four bounds to $-1+2(1-\tau)$ gives \eqref{eq:OR-master-corrected-sharp}. For the lazy measure, label the endpoints so that $d_i\le d_j$. Then $\alpha_u=w_u$, $z_i=z_j=w_{\min}=w_j$, and $r_i+\bar r_i=r_j+\bar r_j=\Delta$, so $U^{\rm env}_{ij}=-1+2(2+t_{ij})\,w_j+2\Delta+\widetilde m^{\rm UU}_{ij}+\widetilde m^{\rm CU}_{ij}$. Here $\widetilde m^{\rm UU}_{ij}\ge\min\{s_i,s_j\}\ge m^{\rm UU}_{ij}$, because $M_{ij}\le|\xi_i|$, and $\widetilde m^{\rm CU}_{ij}=t_{ij}\,\Delta$, because $t_{ij}\,\Delta\le\tau_{ij}$ by Lemma~\ref{lem:one-sided}. Comparing with \eqref{eq:lazy-envelope-tight} gives $U^{\rm lazy}_{ij}\le U^{\rm env}_{ij}$.
\end{proof}

\section{The Reverse Direction}
\label{app:reverse}
This appendix states what an OR value certifies about BF. Throughout, the endpoints are labeled so that $d_i\le d_j$. By Proposition~\ref{prop:lazy-envelope-tight},
\begin{equation}
\label{eq:coverage}
\kappa_{\rm OR}(i,j)\ \le\ \Theta_{ij}=\frac{3+t_{ij}}{b+1}-\frac{1}{a+1},
\end{equation}
which depends only on the degrees and the triangle count. It improves the zero-cost bound $\kappa_{\rm OR}(i,j)\le(2+t_{ij})\,w_j$ by $\Delta$, because the residual at the lower-degree endpoint must travel a distance of at least $2$.

\begin{corollary}[A Lower Bound on BF Curvature]
\label{cor:OR-to-BF-lower}
Fix $(i,j)\in E$ with $a\ge2$, suppose $\kappa_{\rm OR}(i,j)\ge\vartheta$, and let
\[
t_{\min}(\vartheta):=\max\left\{0,\ \left\lceil\big(\vartheta+\Delta\big)(b+1)\right\rceil-2\right\}.
\]
If $t_{\min}(\vartheta)>a-1$, no edge with these degrees satisfies $\kappa_{\rm OR}(i,j)\ge\vartheta$. Otherwise $t_{ij}\ge t_{\min}(\vartheta)$ and
\[
\kappa_{\rm BF}(i,j)\ \ge\ \varphi^{\rm OR\to BF}_{a,b}(\vartheta):=S+T\,t_{\min}(\vartheta).
\]
\end{corollary}

\begin{proof}
By \eqref{eq:coverage}, $(2+t_{ij})\,w_j-\Delta\ge\kappa_{\rm OR}(i,j)\ge\vartheta$, so $t_{ij}\ge(\vartheta+\Delta)(b+1)-2$, and $t_{ij}$ is an integer in $\{0,\dots,a-1\}$. Since $a\ge2$ and $\square_{ij}\ge0$, \eqref{eq:BF-lower-decomp} gives $\kappa_{\rm BF}(i,j)\ge S+Tt_{ij}\ge\varphi^{\rm OR\to BF}_{a,b}(\vartheta)$.
\end{proof}

\begin{corollary}[An Upper Bound on BF Curvature]
\label{cor:or-to-bf}
Fix $(i,j)\in E$ with $a\ge2$ and suppose $\kappa_{\rm OR}(i,j)\le\vartheta$. With $L^0_{ij}$ and $\lambda_1$ as in Theorem~\ref{thm:bf-lower},
\[
\kappa_{\rm BF}(i,j)\ \le\ \psi^{\rm OR\to BF}_{a,b}(\vartheta):=S+\min\left\{T(a-1),\ \frac{\vartheta-L^0_{ij}}{\lambda_1}\right\}.
\]
If $a=1$, then $\kappa_{\rm BF}(i,j)=0$.
\end{corollary}

\begin{proof}
Theorem~\ref{thm:bf-lower} gives $\zeta=\kappa_{\rm BF}(i,j)\ge S$, $L^0_{ij}+\lambda_1(\kappa_{\rm BF}(i,j)-S)\le\kappa_{\rm OR}(i,j)\le\vartheta$, and Proposition~\ref{prop:bf-range} bounds the range of BF; together they give the bound.
\end{proof}

\noindent
The bound depends only on the degrees and inverts the lower modulus. It is never weaker than the bound obtained by inverting the lift $\Lambda^{\rm tf}_{ij}$ of Corollary~\ref{cor:mixing-triangle-free} through Lemma~\ref{lem:topping}, for every $\vartheta\ge L^0_{ij}$, which the hypothesis forces by Theorem~\ref{thm:bf-lower}. If $\Lambda^{\rm tf}_{ij}(\zeta_0)=\vartheta$ with $\zeta_0\le S+T(a-1)$, then $\zeta_0\ge S$, because $\Lambda^{\rm tf}_{ij}$ is strictly increasing and $\Lambda^{\rm tf}_{ij}(S)\le\varphi_{a,b}(S)=L^0_{ij}\le\vartheta$; hence $L^0_{ij}+\lambda_1\big(\zeta_0-S\big)=\varphi_{a,b}(\zeta_0)\ge\vartheta$ by Proposition~\ref{prop:bf-range}, so $\psi^{\rm OR\to BF}_{a,b}(\vartheta)\le\zeta_0$; otherwise the first term of the minimum is already below $\zeta_0$. Unlike that bound, it uses neither Lemma~\ref{lem:idleness} nor Lemma~\ref{lem:topping}. Evaluated at the exact OR curvature on the graphs of Section~\ref{sec:comparison}, it is strictly tighter than the inverted lift on \res{CorFirstTighterMin}--\res{CorFirstTighterMax}\% of the edges.

\section{Deferred Proofs}
\label{appendix:proofs}

\subsection{Proof of Proposition \ref{prop:mixing-transfer}}
\label{subsec:proof-mixing-transfer}
Label the endpoints so that $d_i\le d_j$, and let $p:=w_j=1/(b+1)$ and $\theta:=\Delta/(1-p)\in[0,1]$. Then $m_j=m_j^p$, and $m_i=(1-\theta)\,m_i^p+\theta\,\delta_i$, because both sides put mass $w_i$ at $i$ and spread $1-w_i$ uniformly over $\mathcal N(i)$. Writing $m_j=(1-\theta)\,m_j^p+\theta\,m_j^p$ and mixing an optimal coupling of $m_i^p$ and $m_j^p$ with the product coupling of $\delta_i$ and $m_j^p$ gives a coupling of $m_i$ and $m_j$, so
\[
W_1(m_i,m_j)\ \le\ (1-\theta)\,W_1(m_i^p,m_j^p)+\theta\,W_1(\delta_i,m_j^p).
\]
By Lemma~\ref{lem:idleness} at idleness $p$, $W_1(m_i^p,m_j^p)=1-(1-p)\,\kappa_{\rm LLY}(i,j)$. Under $m_j^p$, mass $p$ sits at $j$, at distance $1$ from $i$, and mass $(1-p)/b$ at each neighbor of $j$, at distance $0$, $1$, or $2$ from $i$ according as it is $i$, lies in $\mathcal C$, or lies in $\mathcal U_j$; hence $W_1(\delta_i,m_j^p)=p+(1-p)\,(t_{ij}+2|\mathcal U_j|)/b=p+(1-p)(2-q_{ij})$. Substituting and using $(1-\theta)(1-p)=1-w_i=a/(a+1)$ and $\theta(1-p)=\Delta$ gives
\[
W_1(m_i,m_j)\ \le\ 1+\Delta(1-q_{ij})-\frac{a}{a+1}\,\kappa_{\rm LLY}(i,j),
\]
which is \eqref{eq:mixing-transfer}; the two consequences follow because $a/(a+1)>0$. For $\kappa_{\rm LLY}(i,j)\le q_{ij}$, $(1-p)/d_u\ge p$ for $u\in\{i,j\}$, so $m^p_i$ and $m^p_j$ share the mass $p$ at $i$, $p$ at $j$, and $(1-p)/b=p$ at every common neighbor; every other unit moves at cost at least $1$, so $\kappa_p(i,j)\le(2+t_{ij})\,p$, and Lemma~\ref{lem:idleness} gives $\kappa_{\rm LLY}(i,j)\le(2+t_{ij})\,p/(1-p)=q_{ij}$.
\hfill\BlackBox

\subsection{Proof of Theorem \ref{thm:JL-plus-squares}}
\label{subsec:proof-JL-squares}
Both sides of \eqref{eq:JL-aug} are symmetric in $i$ and $j$, so assume $d_i\le d_j$; then $1/a\ge 1/b$, $\sigma_{ij}=M_{ij}/b$, and $A^\square_{ij}\ge B^\square_{ij}$. It suffices to exhibit a coupling of $\nu_i$ and $\nu_j$ of cost at most $1-t_{ij}/b+[A^\square_{ij}]_++[B^\square_{ij}]_+$. Recall that $\{j\},\mathcal C,\mathcal U_i$ partition $\mathcal N(i)$, that $\{i\},\mathcal C,\mathcal U_j$ partition $\mathcal N(j)$, and that $|\mathcal U_u|=d_u-1-t_{ij}$.

\emph{Step 1 (cost $0$ and $1$).} Keep mass $1/b$ in place at every $k\in\mathcal C$, and move mass $1/b$ along every edge of a maximum matching of $H_{ij}$; both moves are feasible because $1/a\ge 1/b$ and the matching edges are disjoint. This leaves supply $1/a$ at $j$, $A^\square_{ij}-B^\square_{ij}=t_{ij}(1/a-1/b)$ on $\mathcal C$, and $R:=|\mathcal U_i|/a-\sigma_{ij}=B^\square_{ij}+1/b$ on $\mathcal U_i$, and demand $1/b$ at $i$ and $D:=|\mathcal U_j|/b-\sigma_{ij}=A^\square_{ij}+1/a$ on the unmatched vertices of $\mathcal U_j$; both $R$ and $D$ are nonnegative because $M_{ij}\le\min\{|\mathcal U_i|,|\mathcal U_j|\}$ and $1/b\le 1/a$.

\emph{Step 2 (filling $\mathcal U_j$).} Fill the demand on $\mathcal U_j$ from $j$, then from $\mathcal C$, then from $\mathcal U_i$, at costs at most $1$, $2$ (through $j$), and $3$ (through $i$ and $j$). Since $D-1/a=A^\square_{ij}$, the vertex $j$ covers all of this demand except $[A^\square_{ij}]_+$; then $\mathcal C$ covers $[A^\square_{ij}]_+-[B^\square_{ij}]_+\le A^\square_{ij}-B^\square_{ij}$, and $\mathcal U_i$ covers the remaining $[B^\square_{ij}]_+\le R$.

\emph{Step 3 (all other mass).} All remaining demand is at $i$ and all remaining supply lies in $\mathcal N(i)$, so it moves at cost $1$.

The mass not kept in place totals $1-t_{ij}/b$, and every unit of it costs $1$, except $[A^\square_{ij}]_+-[B^\square_{ij}]_+$ units of cost at most $2$ and $[B^\square_{ij}]_+$ units of cost at most $3$. The coupling therefore costs at most $1-t_{ij}/b+([A^\square_{ij}]_+-[B^\square_{ij}]_+)+2[B^\square_{ij}]_+$, which proves \eqref{eq:JL-aug}. Monotonicity in $\sigma_{ij}$ is immediate because $\sigma_{ij}$ enters only through $A^\square_{ij}$ and $B^\square_{ij}$, and $M_{ij}=0$ recovers the bound of JL.

If $(i,j)$ lies on no cycle of length at most $5$ (for instance, if $G$ is a tree), then $t_{ij}=M_{ij}=0$ and \eqref{eq:JL-aug} reduces to $-2[1-1/a-1/b]_+$. The supports of $\nu_i$ and $\nu_j$ are disjoint, so $W_1(\nu_i,\nu_j)\ge1$, which settles the case $1-1/a-1/b\le0$. If $1-1/a-1/b>0$, then $b\ge a\ge2$, so $\mathcal U_j\neq\emptyset$, and every vertex of $\mathcal U_i$ is at distance $3$ from $\mathcal U_j$, because a shorter path would close a cycle of length $4$ or $5$ through $(i,j)$. The 1-Lipschitz function $\mathrm{dist}_G(\cdot,\mathcal U_j)$, which equals $1$ at $j$, $3$ on $\mathcal U_i$, $2$ at $i$, and $0$ on $\mathcal U_j$, then gives $W_1(\nu_i,\nu_j)\ge1/a+3(a-1)/a-2/b=3-2/a-2/b$, by Kantorovich--Rubinstein duality (Section~\ref{sec:prelims}). Hence $\kappa_0(i,j)\le-2[1-1/a-1/b]_+$, and equality holds. If $a=b$ and $M_{ij}=|\mathcal U_i|$, both arguments of the positive parts equal $-1/a$, so $L^\square_{ij}=t_{ij}/b$, the JL upper bound; this covers $Q_n$, where $u\mapsto u\oplus(i\oplus j)$ is a perfect matching (giving the value $0$), and $K_n$.
\hfill\BlackBox

\subsection{Proof of Theorem \ref{thm:lazy-quadrangle}}
\label{subsec:proof-lazy-quadrangle}
Label the endpoints so that $d_i\le d_j$, whence $\Delta=w_i-w_j$, and write $\mu^+$ and $\mu^-$ for the residuals $(m_i-m_j)_+$ and $(m_i-m_j)_-$ of Lemma~\ref{lem:one-sided}. Keeping $\min\{m_i(k),m_j(k)\}$ in place at every $k$ and adding any coupling of $\mu^+$ and $\mu^-$ gives a coupling of $m_i$ and $m_j$, so it suffices to couple $\mu^+$ and $\mu^-$. We do so in four pieces.
\begin{itemize}
\item[(1)] Along every edge $(x,y)$ of a maximum matching of $H_{ij}$, move $w_j$ from $x\in\mathcal U_i$ to $y\in\mathcal U_j$. This is feasible because $\mu^+(x)=w_i\ge w_j=\mu^-(y)$, and it costs $M_{ij}\,w_j$.
\item[(2)] Move the residual $\Delta$ at $j$ into the unmatched part $Y'$ of $\mathcal U_j$, at cost $1$ per unit, since $j$ is adjacent to all of $\mathcal U_j$.
\item[(3)] Move the residuals at $i$ and on $\mathcal C$, of total mass $(1+t_{ij})\,\Delta$, into $Y'$, at cost at most $2$ per unit, through $j$.
\item[(4)] Move the remaining mass $R_{ij}=|\mathcal U_i|\,w_i-M_{ij}\,w_j\ge0$ on $\mathcal U_i$ into $Y'$, at cost at most $3$ per unit, through $i$ and $j$.
\end{itemize}
By Lemma~\ref{lem:one-sided}, the unmatched demand is $|Y'|\,w_j=\tau_{ij}-M_{ij}\,w_j=(2+t_{ij})\,\Delta+R_{ij}$, so splitting it among pieces (2)--(4) in proportion to their masses $\Delta$, $(1+t_{ij})\,\Delta$, and $R_{ij}$ gives a coupling of $\mu^+$ and $\mu^-$, of cost at most
\[
M_{ij}\,w_j+\Delta+2(1+t_{ij})\,\Delta+3R_{ij} .
\]
Substituting $|\mathcal U_i|=a-1-t_{ij}$, $w_i=1/(a+1)$, and $w_j=1/(b+1)$ shows that this equals $1-L^{\rm lazy}_{ij}$, which proves \eqref{eq:lazy-quadrangle}, and, since $(a+1)\,w_i=1$, also $1-\Theta_{ij}+t_{ij}\,\Delta+2R_{ij}$, which gives the equivalent form.

For exactness, suppose first that $(i,j)$ lies on no cycle of length at most $5$, and that $G\ne K_2$ (on $K_2$, $m_i=m_j$). Then $t_{ij}=M_{ij}=0$, $\mathcal U_j\neq\emptyset$, and every $x\in\mathcal U_i$ is at distance $3$ from $\mathcal U_j$, because a path of length $1$ or $2$ from $x$ to $\mathcal U_j$ would close a cycle of length $4$ or $5$ through $(i,j)$. The 1-Lipschitz function $f:=\mathrm{dist}_G(\cdot,\mathcal U_j)$ gives $W_1(m_i,m_j)\ge\sum_x f(x)\big(m_i(x)-m_j(x)\big)=\Delta+2\Delta+3|\mathcal U_i|\,w_i$ (Section~\ref{sec:prelims}), which is the cost of the coupling above. If instead $a=b$ and $M_{ij}=|\mathcal U_i|=|\mathcal U_j|$, then $\Delta=R_{ij}=0$, so $L^{\rm lazy}_{ij}=\Theta_{ij}$, which is an upper bound on $\kappa_{\rm OR}(i,j)$ (Proposition~\ref{prop:lazy-envelope-tight}). Complete graphs ($\mathcal U_i=\mathcal U_j=\emptyset$) and hypercubes ($u\mapsto u\oplus(i\oplus j)$) belong to the second class.
\hfill\BlackBox

\subsection{Proof of Proposition \ref{prop:bf-range}}
\label{subsec:prop-bf-range}
If $a=1$, then $\kappa_{\rm BF}(i,j)=0<S=2/b$, and for $\zeta\le0$, $\Lambda^{\rm tf}_{ij}(\zeta)\le-\Delta(1-2/b)<L^0_{ij}=\varphi_{a,b}(\zeta)$, because $L^0_{ij}+\Delta(1-2/b)=1/b$. Let $a\ge2$. By \eqref{eq:BF-lower-decomp}, Lemma~\ref{lem:c4-matching}, and $M_{ij}\le|\mathcal U_i|=a-1-t_{ij}$,
\[
\kappa_{\rm BF}(i,j)-S\ \le\ T\,t_{ij}+2\sigma_{ij}\ \le\ T\,t_{ij}+\frac{2(a-1-t_{ij})}{b}=\frac{t_{ij}}a+\frac{2(a-1)}{b}\ \le\ T(a-1),
\]
which is the range of BF. For the comparison, let $\zeta\in[S,\,S+T(a-1)]$, and put $t:=(\zeta-S)/T\in[0,a-1]$, so that $\varphi_{a,b}(\zeta)=L^0_{ij}+\lambda_1Tt$. The identity of Remark~\ref{rem:lazy-dominates} is an identity of rational functions of $(t_{ij},M_{ij})$; at $(t,0)$, where $L^{\rm lazy}_{ij}$ becomes $\varphi_{a,b}(\zeta)$ (Theorem~\ref{thm:bf-lower}) and $P_{ij}$ becomes $\zeta$, it reads
\[
\varphi_{a,b}(\zeta)=\frac{a}{a+1}\,\zeta-\Delta\Big(1-\frac{2+t}{b}\Big)+\frac{(b-a)\,t}{b(a+1)(b+1)}.
\]
Since $(2+t)/b-\zeta=\big(2b(a-1)-(a+b)t\big)/(ab)\ge(a-1)(b-a)/(ab)\ge0$, we have $(2+t)/b\ge\max\{\zeta,2/b\}$, so $\varphi_{a,b}(\zeta)\ge\Lambda^{\rm tf}_{ij}(\zeta)$. If $\zeta<S$, then $\varphi_{a,b}(\zeta)=\varphi_{a,b}(S)\ge\Lambda^{\rm tf}_{ij}(S)\ge\Lambda^{\rm tf}_{ij}(\zeta)$, because $\Lambda^{\rm tf}_{ij}$ is nondecreasing. 
\hfill\BlackBox

\subsection{Proof of Theorem \ref{thm:BF-to-OR-lazy}}
By Proposition~\ref{prop:lazy-envelope-tight}, $\kappa_{\rm OR}(i,j)\le-1+2(2+t_{ij})\,w_j+(1+t_{ij})\,\Delta+m^{\rm UU}_{ij}$, with $m^{\rm UU}_{ij}\le|\mathcal U_i|\,w_i=(a-1-t_{ij})\,w_i$ and $m^{\rm UU}_{ij}\le\Xi_{ij}/(a+b+2)$. Since $a\,w_i=1-w_i$, the first gives
\[
-1+2(2+t_{ij})\,w_j+(1+t_{ij})\,\Delta+(a-1-t_{ij})\,w_i\ =\ (3+t_{ij})\,w_j-w_i\ =\ (2+t_{ij})\,w_j-\Delta,
\]
and Lemma~\ref{lem:budget} turns the second into the second piece of $\Psi(t_{ij})$. Hence $\kappa_{\rm OR}(i,j)\le\Psi(t_{ij})$, and since $t_{ij}\in\{0,1,\dots,t_{\max}(\zeta)\}$ by Lemma~\ref{lem:budget}, $\kappa_{\rm OR}(i,j)\le\psi_{a,b}(\zeta)$. The same computation with $t_{ij}$ replaced by $t$ gives the equivalent form of $\Psi$, because $\min\{x,x-y\}=x-[y]_+$. The two pieces are affine in $t$, and the first has slope $w_j>0$, so their minimum, restricted to $\{0,1,\dots,t_{\max}(\zeta)\}$, attains its maximum at an end of this range or at an integer adjacent to the point where the pieces cross. Finally, $\beta(\zeta)$ is nondecreasing in $\zeta$, so as $\zeta$ grows the set of admissible $t$ grows and the second piece of $\Psi$ increases pointwise, while the first does not depend on $\zeta$; hence $\psi_{a,b}$ is nondecreasing.
\hfill\BlackBox

\section{Correspondence with the Lean Formalization}
\label{app:lean}
The bounds of Sections~\ref{sec:prelims} to~\ref{sec:bf-to-or} and Appendices~\ref{app:bm-comparison} to~\ref{app:reverse}, and the correctness of Algorithm~\ref{alg:certify-then-solve}, are formalized in Lean~4 with Mathlib, except for the statements listed at the end of this appendix. The formalization builds with the Lean toolchain \texttt{v4.36.0-rc1} and the Mathlib revision pinned in its repository, with no axioms beyond the three standard ones of Lean. Its theorems hold for every distance with the properties of the graph metric that the proofs use (\texttt{GraphMetric}): it vanishes on the diagonal, equals $1$ on edges, is at least $1$ between distinct vertices and at least $2$ between distinct non-adjacent ones, is symmetric, and satisfies the triangle inequality.

Table~\ref{tab:lean} matches each result of this paper with the theorems that state it and the modules of the code that contain them. The modules \texttt{Paper}, \texttt{Unclipped}, \texttt{Exactness}, \texttt{Modulus}, \texttt{ArbitraryIdleness}, and \texttt{Screening} state results in the form of the paper, with the endpoints labeled so that $d_i\le d_j$ wherever a statement depends on it; the other results are stated in the modules that prove them. Every theorem in the table is proved from the definitions alone, except \texttt{mixing\_transfer\_paper}, \texttt{cLLY\_le\_q\_paper}, and \texttt{mixing\_tf\_paper} (Proposition~\ref{prop:mixing-transfer} and Corollary~\ref{cor:mixing-triangle-free}), which take the result of \citet{bourne_idleness_2018} (Lemma~\ref{lem:idleness}) as the hypothesis \texttt{BourneFact}.

\begin{table}[h]
\centering
\caption{Results of the paper, their Lean equivalent, and the modules that contain them.}
\label{tab:lean}
\footnotesize
\renewcommand{\arraystretch}{1.15}
\begin{tabularx}{\textwidth}{@{}>{\RaggedRight\arraybackslash}p{0.24\textwidth}>{\RaggedRight\arraybackslash}X>{\RaggedRight\arraybackslash}p{0.19\textwidth}@{}}
\toprule
Result & Lean & Module \\
\midrule
Lemma~\ref{lem:c4-matching} & \texttt{C4\_le\_two\_matchStat} & \texttt{Cover} \\
Proposition~\ref{prop:mixing-transfer} & \texttt{mixing\_transfer\_paper}, \texttt{cLLY\_le\_q\_paper} & \texttt{Paper} \\
Corollary~\ref{cor:mixing-triangle-free} & \texttt{mixing\_tf\_paper}, \texttt{lamTF\_closed} & \texttt{Paper} \\
Theorem~\ref{thm:JL-plus-squares} & \texttt{jl\_aug}, \texttt{jl\_aug\_exact\_of\_perfect}, \texttt{cOR0\_le\_tri\_div} & \texttt{Bounds}, \texttt{Exactness} \\
Remark~\ref{rem:bm-comparison}, Appendix~\ref{app:bm-comparison} & \texttt{LBM\_le\_Lsq}, \texttt{LBM\_eq\_Lsq\_iff} & \texttt{Appendix} \\
Remark~\ref{rem:coarse-C4} & \texttt{pos\_part\_sub\_pos\_part} & \texttt{Algebra} \\
Lemma~\ref{lem:one-sided} & \texttt{lazy\_min}, \texttt{lazy\_resid\_i}, \texttt{lazy\_resid\_j} & \texttt{Curvature} \\
Theorem~\ref{thm:lazy-quadrangle} & \texttt{lazy\_quadrangle}, \texttt{Llazy\_deficit}, \texttt{lazy\_quadrangle\_exact\_of\_perfect} & \texttt{Bounds}, \texttt{Paper}, \texttt{Exactness} \\
Proposition~\ref{prop:lazy-envelope-tight} & \texttt{lazy\_envelope\_cover}, \texttt{lazy\_envelope\_paper}, \texttt{mUU\_le} & \texttt{Cover}, \texttt{Paper} \\
Proposition~\ref{prop:bf-unclipped} & \texttt{bf\_is\_third\_piece}, \texttt{LBF\_eq}, \texttt{LBF\_le\_Lsq}, \texttt{bf\_excess}, \texttt{bfB\_le\_bfA}, \texttt{bf\_le\_cOR0\_of\_bfB\_nonneg} & \texttt{Bounds}, \texttt{Unclipped} \\
Remark~\ref{rem:lazy-dominates} & \texttt{Llazy\_sub\_lamMix}, \texttt{lamMix\_piece3\_le} & \texttt{Algebra} \\
Theorem~\ref{thm:bf-lower} & \texttt{bf\_lower\_paper}, \texttt{phiAff\_mono}, \texttt{phiAff\_virtual}, \texttt{phiAff\_continuous}, \texttt{phiAff\_convexOn}, \texttt{phiAff\_of\_le}, \texttt{phiAff\_of\_ge} & \texttt{Paper}, \texttt{Modulus} \\
Proposition~\ref{prop:bf-range} & \texttt{cBF\_le\_range}, \texttt{phiAff\_ge\_lamTF} & \texttt{Paper} \\
Lemma~\ref{lem:budget} & \texttt{budget\_region} & \texttt{Paper} \\
Theorem~\ref{thm:BF-to-OR-lazy} & \texttt{bf\_upper\_paper}, \texttt{psiInt\_mono}, \texttt{psiTwo\_deficit} & \texttt{Paper} \\
Corollary~\ref{cor:bf-interval} & \texttt{bf\_interval\_paper}, \texttt{phiInt\_ge\_phiAff} & \texttt{Paper} \\
Proposition~\ref{prop:sharpened-interval} & \texttt{two\_hop\_interval\_paper} & \texttt{Paper} \\
Algorithm~\ref{alg:certify-then-solve} & \texttt{certifyThenSolve\_eq}, \texttt{certifyThenSolve\_paper} & \texttt{Screening} \\
Proposition~\ref{prop:lazy-envelope-sharpened} & \texttt{envelope\_arbitrary\_idleness}, \texttt{Ulazy\_le\_Uenv} & \texttt{ArbitraryIdleness} \\
Corollary~\ref{cor:OR-to-BF-lower} & \texttt{or\_to\_bf\_lower\_paper} & \texttt{Paper} \\
Corollary~\ref{cor:or-to-bf} & \texttt{or\_to\_bf\_upper\_paper} & \texttt{Paper} \\
\bottomrule
\end{tabularx}
\end{table}

\noindent
Two formal proofs take a different route from the paper: Proposition~\ref{prop:lazy-envelope-sharpened} follows from weak duality with two pairs of explicit potentials, and the correctness of Algorithm~\ref{alg:certify-then-solve} is proved for an arbitrary upper bound in its second stage. The formalization also proves the upper bound of JL, so no theorem in the table takes it as a hypothesis.

\noindent
Not formalized are the exactness of Theorems~\ref{thm:JL-plus-squares} and~\ref{thm:lazy-quadrangle} on edges that lie on no cycle of length at most 5; the exactness statements of Theorem~\ref{thm:bf-lower} and Corollary~\ref{cor:bf-interval}; the location of the maximizers of $\Psi$ and the $\mathcal O(1)$ evaluation in Corollary~\ref{cor:bf-interval}; the closed form of the matching number of \citet{bhattacharya_exact_2020} in Remark~\ref{rem:bm-comparison}; the comparison that follows Corollary~\ref{cor:or-to-bf}; the examples in Section~\ref{sec:introduction} and in Remarks~\ref{rem:bf-not-lower} and~\ref{rem:bf-interpretation}; and the validity of the dual bound of TLW that Algorithm~\ref{alg:certify-then-solve} uses in its second stage.

{\small
\bibliography{tex/export-data}

@article{knight_sinkhorn_2008,
author = {Knight, Philip A.},
title = {The {Sinkhorn}--{Knopp} Algorithm: Convergence and Applications},
journal = {SIAM Journal on Matrix Analysis and Applications},
volume = {30},
number = {1},
pages = {261-275},
year = {2008},
doi = {10.1137/060659624},
URL = {https://doi.org/10.1137/060659624},
eprint = {https://doi.org/10.1137/060659624}
}

@article{tian_curvature_2026,
  author  = {Yu Tian and Zachary Lubberts and Melanie Weber},
  title   = {Curvature-based Clustering on Graphs},
  journal = {Journal of Machine Learning Research},
  year    = {2025},
  volume  = {26},
  number  = {52},
  pages   = {1--67},
  url     = {http://jmlr.org/papers/v26/24-0781.html}
}

@inproceedings{liu_curvdrop_2023,
author = {Liu, Yang and Zhou, Chuan and Pan, Shirui and Wu, Jia and Li, Zhao and Chen, Hongyang and Zhang, Peng},
title = {CurvDrop: A {Ricci} Curvature Based Approach to Prevent Graph Neural Networks from Over-Smoothing and Over-Squashing},
year = {2023},
isbn = {9781450394161},
publisher = {Association for Computing Machinery},
address = {New York, NY, USA},
url = {https://doi.org/10.1145/3543507.3583269},
doi = {10.1145/3543507.3583269},
booktitle = {Proceedings of the ACM Web Conference 2023},
pages = {221–230},
numpages = {10},
location = {Austin, TX, USA},
series = {WWW '23}
}

@misc{zhou_cycle_2026,
      title={{Ollivier--Ricci} curvature in cycle overlap mode}, 
      author={Zexian Zhou and Bo Jiao},
      year={2026},
      eprint={2606.03317},
      archivePrefix={arXiv},
      primaryClass={cs.SI},
      url={https://arxiv.org/abs/2606.03317}, 
}

@article{bhattacharya_exact_2020,
title = {Exact and asymptotic results on coarse {Ricci} curvature of graphs},
journal = {Discrete Mathematics},
volume = {338},
number = {1},
pages = {23-42},
year = {2015},
issn = {0012-365X},
doi = {https://doi.org/10.1016/j.disc.2014.08.012},
url = {https://www.sciencedirect.com/science/article/pii/S0012365X14003252},
author = {Bhaswar B. Bhattacharya and Sumit Mukherjee}
}

@inproceedings{topping_squashing_2022,
title={Understanding over-squashing and bottlenecks on graphs via curvature},
author={Jake Topping and Francesco Di Giovanni and Benjamin Paul Chamberlain and Xiaowen Dong and Michael M. Bronstein},
booktitle={International Conference on Learning Representations},
year={2022},
url={https://openreview.net/forum?id=7UmjRGzp-A}
}

@inproceedings{nguyen_over_2023,
author = {Nguyen, Khang and Nong, Hieu and Nguyen, Vinh and Ho, Nhat and Osher, Stanley and Nguyen, Tan},
title = {Revisiting over-smoothing and over-squashing using {Ollivier--Ricci} curvature},
year = {2023},
publisher = {JMLR.org},
booktitle = {Proceedings of the 40th International Conference on Machine Learning},
articleno = {1080},
numpages = {24},
location = {Honolulu, Hawaii, USA},
series = {ICML'23}
}

@article{wu_ricci_2021,
AUTHOR = {Wu, Wei and Hu, Guangmin and Yu, Fucai},
TITLE = {Ricci Curvature-Based Semi-Supervised Learning on an Attributed Network},
JOURNAL = {Entropy},
VOLUME = {23},
YEAR = {2021},
NUMBER = {3},
ARTICLE-NUMBER = {292},
URL = {https://www.mdpi.com/1099-4300/23/3/292},
PubMedID = {33673440},
ISSN = {1099-4300},
DOI = {10.3390/e23030292}
}

@article{ni_community_2019,
author={Ni, Chien-Chun
and Lin, Yu-Yao
and Luo, Feng
and Gao, Jie},
title={Community Detection on Networks with {Ricci} Flow},
journal={Scientific Reports},
year={2019},
month={Jul},
day={10},
volume={9},
number={1},
pages={9984},
issn={2045-2322},
doi={10.1038/s41598-019-46380-9},
url={https://doi.org/10.1038/s41598-019-46380-9}
}

@misc{gu_residual_2026,
title = {A residual-shell-based lower bound for {Ollivier--Ricci} curvature},
journal = {Applied Mathematics Letters},
volume = {183},
pages = {110095},
year = {2026},
issn = {0893-9659},
doi = {https://doi.org/10.1016/j.aml.2026.110095},
url = {https://www.sciencedirect.com/science/article/pii/S0893965926002260},
author = {Xiang Gu and Huichun Zhang and Jian Sun}
}

@inproceedings{tori_effectiveness_2024,
title={The Effectiveness of Curvature-Based Rewiring and the Role of Hyperparameters in {GNN}s Revisited},
author={Floriano Tori and Vincent Holst and Vincent Ginis},
booktitle={The Third Learning on Graphs Conference},
year={2024},
url={https://openreview.net/forum?id=pplvaG4Gbh}
}

@misc{attali_rewiring_2026,
  title     = {Graph Rewiring in GNNs to Mitigate Over-Squashing and Over-Smoothing: A Survey},
  author    = {Attali, Hugo and Pernelle, Nathalie and Buscaldi, Davide and Malliaros, Fragkiskos D.},
  booktitle = {Proceedings of the Thirty-Fifth International Joint Conference on
               Artificial Intelligence, {IJCAI-26}},
  publisher = {International Joint Conferences on Artificial Intelligence Organization},
  editor    = {Diego Calvanese},
  pages     = {7749--7757},
  year      = {2026},
  month     = {8},
  note      = {Survey Track},
  doi       = {10.24963/ijcai.2026/861},
  url       = {https://doi.org/10.24963/ijcai.2026/861},
}

@inproceedings{sun_ricci_2024,
  author={Sun, Li and Huang, Zhenhao and Wu, Hua and Ye, Junda and Peng, Hao and Yu, Zhengtao and Yu, Philip S.},
  booktitle={2023 IEEE International Conference on Data Mining (ICDM)}, 
  title={{DeepRicci}: Self-supervised Graph Structure-Feature Co-Refinement for Alleviating Over-squashing}, 
  year={2023},
  volume={},
  number={},
  pages={558-567},
  doi={10.1109/ICDM58522.2023.00065}}

@article{weber_forman_2017,
    author = {Weber, Melanie and Saucan, Emil and Jost, Jürgen},
    title = {Characterizing complex networks with {Forman-Ricci} curvature and associated geometric flows},
    journal = {Journal of Complex Networks},
    volume = {5},
    number = {4},
    pages = {527-550},
    year = {2017},
    month = {08},
    issn = {2051-1310},
    doi = {10.1093/comnet/cnw030},
    url = {https://doi.org/10.1093/comnet/cnw030},
    eprint = {https://academic.oup.com/comnet/article-pdf/5/4/527/19402603/cnw030.pdf},
}

@article{steger_generating_1999, 
title={Generating Random Regular Graphs Quickly},
volume={8}, DOI={10.1017/S0963548399003867}, 
number={4}, 
journal={Combinatorics, Probability and Computing}, 
author={Steger, A. and Wormald, N. C.}, 
year={1999}, 
pages={377–396}
}

@article{holland_stochastic_1983,
title = {Stochastic blockmodels: First steps},
journal = {Social Networks},
volume = {5},
number = {2},
pages = {109-137},
year = {1983},
issn = {0378-8733},
doi = {https://doi.org/10.1016/0378-8733(83)90021-7},
url = {https://www.sciencedirect.com/science/article/pii/0378873383900217},
author = {Paul W. Holland and Kathryn Blackmond Laskey and Samuel Leinhardt}
}

@article{gehrke_arxiv_2003, 
author = {Gehrke, Johannes and Ginsparg, Paul and Kleinberg, Jon},
title = {Overview of the 2003 {KDD} Cup},
year = {2003},
issue_date = {December 2003},
publisher = {Association for Computing Machinery},
address = {New York, NY, USA},
volume = {5},
number = {2},
issn = {1931-0145},
url = {https://doi.org/10.1145/980972.980992},
doi = {10.1145/980972.980992},
journal = {SIGKDD Explor. Newsl.},
month = dec,
pages = {149–151},
numpages = {3}
}

@article{milo_motifs_2002,
author = {R. Milo  and S. Shen-Orr  and S. Itzkovitz  and N. Kashtan  and D. Chklovskii  and U. Alon },
title = {Network Motifs: Simple Building Blocks of Complex Networks},
journal = {Science},
volume = {298},
number = {5594},
pages = {824-827},
year = {2002},
doi = {10.1126/science.298.5594.824},
URL = {https://www.science.org/doi/abs/10.1126/science.298.5594.824},
eprint = {https://www.science.org/doi/pdf/10.1126/science.298.5594.824}}

@misc{kelly_exact_2020,
      title={Exact Expressions and Reduced Linear Programmes for the {Ollivier} Curvature in Graphs}, 
      author={Christy Kelly},
      year={2020},
      eprint={1909.12156},
      archivePrefix={arXiv},
      primaryClass={math.OC},
      url={https://arxiv.org/abs/1909.12156}, 
}

@article{barabasi_emergence_1999,
author = {Albert-László Barabási  and Réka Albert },
title = {Emergence of Scaling in Random Networks},
journal = {Science},
volume = {286},
number = {5439},
pages = {509-512},
year = {1999},
doi = {10.1126/science.286.5439.509},
URL = {https://www.science.org/doi/abs/10.1126/science.286.5439.509},
eprint = {https://www.science.org/doi/pdf/10.1126/science.286.5439.509}
}

@article{gleiser_community_2003,
author = {Gleiser, Pablo M. and Danon, Leon},
title = {COMMUNITY STRUCTURE IN JAZZ},
journal = {Advances in Complex Systems},
volume = {06},
number = {04},
pages = {565-573},
year = {2003},
doi = {10.1142/S0219525903001067}
}

@article{zachary_karate_1977,
 ISSN = {00917710},
 URL = {http://www.jstor.org/stable/3629752},
 author = {Wayne W. Zachary},
 journal = {Journal of Anthropological Research},
 number = {4},
 pages = {452--473},
 publisher = {[University of New Mexico, University of Chicago Press]},
 title = {An Information Flow Model for Conflict and Fission in Small Groups},
 urldate = {2026-11-04},
 volume = {33},
 year = {1977}
}

@article{krioukov_hyperbolic_2010,
  title = {Hyperbolic geometry of complex networks},
  author = {Krioukov, Dmitri and Papadopoulos, Fragkiskos and Kitsak, Maksim and Vahdat, Amin and Bogu{\~n}{\'a}, Mari{\'a}n},
  journal = {Phys. Rev. E},
  volume = {82},
  issue = {3},
  pages = {036106},
  numpages = {18},
  year = {2010},
  month = {Sep},
  publisher = {American Physical Society},
  doi = {10.1103/PhysRevE.82.036106},
  url = {https://link.aps.org/doi/10.1103/PhysRevE.82.036106}
}

@book{penrose_random_2003,
    author = {Penrose, Mathew},
    title = {Random Geometric Graphs},
    publisher = {Oxford University Press},
    year = {2003},
    month = {05},
    isbn = {9780198506263},
    doi = {10.1093/acprof:oso/9780198506263.001.0001},
    url = {https://doi.org/10.1093/acprof:oso/9780198506263.001.0001},
}

@article{watts_collective_1998,
author={Watts, Duncan J.
and Strogatz, Steven H.},
title={Collective dynamics of `small-world' networks},
journal={Nature},
year={1998},
month={Jun},
day={01},
volume={393},
number={6684},
pages={440-442},
issn={1476-4687},
doi={10.1038/30918},
url={https://doi.org/10.1038/30918}
}

@article{erdos_random_1959,
  title        = {On random graphs {I}},
  author       = {Erd{\H{o}}s, Paul and R{\'e}nyi, Alfr{\'e}d},
  journal      = {Publicationes Mathematicae Debrecen},
  volume       = {6},
  pages        = {290--297},
  year         = {1959},
  doi          = {10.5486/pmd.1959.6.3-4.12}
}

@article{vanderhoorn_ollivier_2021,
  title = {{Ollivier-Ricci} curvature convergence in random geometric graphs},
  author = {van der Hoorn, Pim and Cunningham, William J. and Lippner, Gabor and Trugenberger, Carlo and Krioukov, Dmitri},
  journal = {Phys. Rev. Res.},
  volume = {3},
  issue = {1},
  pages = {013211},
  numpages = {14},
  year = {2021},
  month = {Mar},
  publisher = {American Physical Society},
  doi = {10.1103/PhysRevResearch.3.013211},
  url = {https://link.aps.org/doi/10.1103/PhysRevResearch.3.013211}
}

@article{tee_enhanced_2021,
   title={Enhanced {Forman} curvature and its relation to {Ollivier} curvature},
   volume={133},
   ISSN={1286-4854},
   url={http://dx.doi.org/10.1209/0295-5075/133/60006},
   DOI={10.1209/0295-5075/133/60006},
   number={6},
   journal={Europhysics Letters},
   publisher={IOP Publishing},
   author={Tee, Philip and Trugenberger, C. A.},
   year={2021},
   month=mar, pages={60006} }

@article{eidi_edge_2020,
author={Eidi, Marzieh
and Farzam, Amirhossein
and Leal, Wilmer
and Samal, Areejit
and Jost, J{\"u}rgen},
title={Edge-based analysis of networks: curvatures of graphs and hypergraphs},
journal={Theory in Biosciences},
year={2020},
month={Dec},
day={01},
volume={139},
number={4},
pages={337-348},
issn={1611-7530},
doi={10.1007/s12064-020-00328-0},
url={https://doi.org/10.1007/s12064-020-00328-0}
}

@article{saucan_discrete_2018,
  title={Discrete Curvatures and Network Analysis},
  author={Saucan, Emil and Samal, Areejit and Weber, Melanie and Jost, J{\"u}rgen},
  journal={MATCH},
  volume={80},
  number={3},
  pages={605--622},
  year={2018}
}

@book{villani_topics_2003,
  author    = {Villani, C{\'e}dric},
  title     = {Topics in Optimal Transportation},
  series    = {Graduate Studies in Mathematics},
  volume    = {58},
  publisher = {American Mathematical Society},
  address   = {Providence, RI},
  year      = {2003},
  month     = {03},
  pages     = {xvi+370},
  isbn      = {978-0-8218-3312-4},
  doi       = {10.1090/gsm/058}
}

@inproceedings{chen_graph_2025,
title={Graph Neural {Ricci} Flow: Evolving Feature from a Curvature Perspective},
author={Jialong Chen and Bowen Deng and Zhen WANG and Chuan Chen and Zibin Zheng},
booktitle={The Thirteenth International Conference on Learning Representations},
year={2025},
url={https://openreview.net/forum?id=7b2JrzdLhA}
}

@article{jost_clustering_2014,
author={Jost, J{\"u}rgen
and Liu, Shiping},
title={{Ollivier's Ricci} Curvature, Local Clustering and Curvature-Dimension Inequalities on Graphs},
journal={Discrete {\&} Computational Geometry},
year={2014},
month={Mar},
day={01},
volume={51},
number={2},
pages={300-322},
issn={1432-0444},
doi={10.1007/s00454-013-9558-1},
url={https://doi.org/10.1007/s00454-013-9558-1}
}

@inproceedings{fu_discrete_2024,
author = {Fu, Xingcheng and Wang, Jian and Gao, Yisen and Sun, Qingyun and Yuan, Haonan and Li, Jianxin and Li, Xianxian},
title = {Discrete curvature graph information bottleneck},
year = {2025},
isbn = {978-1-57735-897-8},
publisher = {AAAI Press},
url = {https://doi.org/10.1609/aaai.v39i16.33831},
doi = {10.1609/aaai.v39i16.33831},
booktitle = {Proceedings of the Thirty-Ninth AAAI Conference on Artificial Intelligence and Thirty-Seventh Conference on Innovative Applications of Artificial Intelligence and Fifteenth Symposium on Educational Advances in Artificial Intelligence},
articleno = {1858},
numpages = {8},
series = {AAAI'25/IAAI'25/EAAI'25}
}

@misc{kang_accelerated_2026,
      title={Transport-Free Lower Bounds for Ollivier--Ricci Curvature in Higher-Order Networks: Scope, Tightness, and Computational Applications}, 
      author={Wonwoo Kang and Heehyun Park},
      year={2026},
      eprint={2405.13302},
      archivePrefix={arXiv},
      primaryClass={stat.ML},
      url={https://arxiv.org/abs/2405.13302}, 
}

@misc{jost_forman_2021,
      title={Characterizations of {Forman} curvature}, 
      author={Jürgen Jost and Florentin Münch},
      year={2021},
      eprint={2110.04554},
      archivePrefix={arXiv},
      primaryClass={math.DG},
      url={https://arxiv.org/abs/2110.04554}, 
}

@article{sreejith_systematic_2017,
title = {Systematic evaluation of a new combinatorial curvature for complex networks},
journal = {Chaos, Solitons {\&} Fractals},
volume = {101},
pages = {50-67},
year = {2017},
issn = {0960-0779},
doi = {https://doi.org/10.1016/j.chaos.2017.05.021},
url = {https://www.sciencedirect.com/science/article/pii/S0960077917302102},
author = {R.P. Sreejith and Jürgen Jost and Emil Saucan and Areejit Samal}
}

@article{sreejith_forman_2016,
   title={{Forman} curvature for complex networks},
   volume={2016},
   ISSN={1742-5468},
   url={http://dx.doi.org/10.1088/1742-5468/2016/06/063206},
   DOI={10.1088/1742-5468/2016/06/063206},
   number={6},
   journal={Journal of Statistical Mechanics: Theory and Experiment},
   publisher={IOP Publishing},
   author={Sreejith, R P and Mohanraj, Karthikeyan and Jost, Jürgen and Saucan, Emil and Samal, Areejit},
   year={2016},
   month=jun, pages={063206} 
}

@article{samal_comparative_2018,
author={Samal, Areejit
and Sreejith, R. P.
and Gu, Jiao
and Liu, Shiping
and Saucan, Emil
and Jost, J{\"u}rgen},
title={Comparative analysis of two discretizations of {Ricci} curvature for complex networks},
journal={Scientific Reports},
year={2018},
month={Jun},
day={05},
volume={8},
number={1},
pages={8650},
issn={2045-2322},
doi={10.1038/s41598-018-27001-3},
url={https://doi.org/10.1038/s41598-018-27001-3}
}

@Article{sia_ollivier_2019,
author={Sia, Jayson
and Jonckheere, Edmond
and Bogdan, Paul},
title={Ollivier-Ricci Curvature-Based Method to Community Detection in Complex Networks},
journal={Scientific Reports},
year={2019},
month={Jul},
day={05},
volume={9},
number={1},
pages={9800},
issn={2045-2322},
doi={10.1038/s41598-019-46079-x},
url={https://doi.org/10.1038/s41598-019-46079-x}
}

@InProceedings{fesser_mitigating_2024,
  title = 	 {Mitigating Over-Smoothing and Over-Squashing Using Augmentations of {Forman}-{Ricci} Curvature},
  author =       {Fesser, Lukas and Weber, Melanie},
  booktitle = 	 {Proceedings of the Second Learning on Graphs Conference},
  pages = 	 {19:1--19:28},
  year = 	 {2024},
  editor = 	 {Villar, Soledad and Chamberlain, Benjamin},
  volume = 	 {231},
  series = 	 {Proceedings of Machine Learning Research},
  month = 	 {27--30 Nov},
  publisher =    {PMLR},
  url = 	 {https://proceedings.mlr.press/v231/fesser24a.html}
}

@article{feng_graph_2024,
title={Graph Pooling via {Ricci} Flow},
author={Amy Feng and Melanie Weber},
journal={Transactions on Machine Learning Research},
issn={2835-8856},
year={2024},
url={https://openreview.net/forum?id=xpBHp9WFvk},
note={}
}

@inproceedings{digiovanni_oversquashing_2023,
author = {Di Giovanni, Francesco and Giusti, Lorenzo and Barbero, Federico and Luise, Giulia and Li{\`o}, Pietro and Bronstein, Michael},
title = {On over-squashing in message passing neural networks: the impact of width, depth, and topology},
year = {2023},
publisher = {JMLR.org},
booktitle = {Proceedings of the 40th International Conference on Machine Learning},
articleno = {312},
numpages = {21},
location = {Honolulu, Hawaii, USA},
series = {ICML'23}
}

@article{flamary_pot_2021,
  author  = {R\'{e}mi Flamary and Nicolas Courty and Alexandre Gramfort and Mokhtar Z. Alaya and AurÃ©lie Boisbunon and Stanislas Chambon and Laetitia Chapel and Adrien Corenflos and Kilian Fatras and Nemo Fournier and LÃ©o Gautheron and Nathalie T.H. Gayraud and Hicham Janati and Alain Rakotomamonjy and Ievgen Redko and Antoine Rolet and Antony Schutz and Vivien Seguy and Danica J. Sutherland and Romain Tavenard and Alexander Tong and Titouan Vayer},
  title   = {{POT}: {Python} Optimal Transport},
  journal = {Journal of Machine Learning Research},
  year    = {2021},
  volume  = {22},
  number  = {78},
  pages   = {1--8},
  url     = {http://jmlr.org/papers/v22/20-451.html}
}

@article{park_lower_2025,
title={Lower {Ricci} Curvature for Efficient Community Detection},
author={Yun Jin Park and Didong Li},
journal={Transactions on Machine Learning Research},
issn={2835-8856},
year={2025},
url={https://openreview.net/forum?id=EoiuRII7MQ},
note={}
}

@InProceedings{sanders_pooling_2025,
author="Sanders, Cedric
and Roth, Andreas
and Liebig, Thomas",
editor="Meo, Rosa
and Silvestri, Fabrizio",
title="Curvature-Based Pooling Within Graph Neural Networks",
booktitle="Machine Learning and Principles and Practice of Knowledge Discovery in Databases",
year="2025",
publisher="Springer Nature Switzerland",
address="Cham",
pages="471--485",
isbn="978-3-031-74633-8"
}

@article{munch_curvature_2017,
title = {{Ollivier Ricci} curvature for general graph {Laplacians}: {Heat} equation, {Laplacian} comparison, non-explosion and diameter bounds},
journal = {Advances in Mathematics},
volume = {356},
pages = {106759},
year = {2019},
issn = {0001-8708},
doi = {https://doi.org/10.1016/j.aim.2019.106759},
url = {https://www.sciencedirect.com/science/article/pii/S000187081930369X},
author = {Florentin Münch and Radosław K. Wojciechowski}
}

@article{lin_curvature_2011,
author = {Yong Lin and Linyuan Lu and Shing-Tung Yau},
title = {{Ricci curvature of graphs}},
volume = {63},
journal = {Tohoku Mathematical Journal},
number = {4},
publisher = {Tohoku University, Mathematical Institute},
pages = {605 -- 627},
year = {2011},
doi = {10.2748/tmj/1325886283},
URL = {https://doi.org/10.2748/tmj/1325886283}
}

@article{bourne_idleness_2018,
author = {Bourne, D. P. and Cushing, D. and Liu, S. and M{\"u}nch, F. and Peyerimhoff, N.},
title = {{Ollivier--Ricci} Idleness Functions of Graphs},
journal = {SIAM Journal on Discrete Mathematics},
volume = {32},
number = {2},
pages = {1408-1424},
year = {2018},
doi = {10.1137/17M1134469},
URL = {https://doi.org/10.1137/17M1134469},
eprint = {https://doi.org/10.1137/17M1134469}
}

@article{smith_ricci_2014,
  author  = {Smith, Jonathan D. H.},
  title   = {Ricci Curvature, Circulants, and a Matching Condition},
  journal = {Discrete Mathematics},
  volume  = {329},
  pages   = {88--98},
  year    = {2014}
}

@article{hopcroft_n52_1973,
author = {Hopcroft, John E. and Karp, Richard M.},
title = {An $n^{5/2}$ Algorithm for Maximum Matchings in Bipartite Graphs},
journal = {SIAM Journal on Computing},
volume = {2},
number = {4},
pages = {225-231},
year = {1973},
doi = {10.1137/0202019},
URL = {https://doi.org/10.1137/0202019},
eprint = {https://doi.org/10.1137/0202019}
}

@book{ahuja_network_1993,
  author    = {Ahuja, Ravindra K. and Magnanti, Thomas L. and Orlin, James B.},
  title     = {Network Flows: Theory, Algorithms, and Applications},
  publisher = {Prentice-Hall, Inc.},
  address   = {Englewood Cliffs, NJ},
  year      = {1993},
  isbn = {013617549X},
  url       = {https://dl.acm.org/doi/10.5555/3153716}
}

@article{kuhn_hungarian_1955,
author = {Kuhn, H. W.},
title = {The {Hungarian} method for the assignment problem},
journal = {Naval Research Logistics Quarterly},
volume = {2},
number = {1-2},
pages = {83-97},
doi = {https://doi.org/10.1002/nav.3800020109},
url = {https://onlinelibrary.wiley.com/doi/abs/10.1002/nav.3800020109},
eprint = {https://onlinelibrary.wiley.com/doi/pdf/10.1002/nav.3800020109},
year = {1955}
}

@article{dasgupta_computing_2023,
author = {DasGupta, Bhaskar and Grigorescu, Elena and Mukherjee, Tamalika},
title = {On computing discretized {Ricci} curvatures of graphs: Local algorithms and (localized) fine-grained reductions},
year = {2023},
issue_date = {Oct 2023},
publisher = {Elsevier Science Publishers Ltd.},
address = {GBR},
volume = {975},
number = {C},
issn = {0304-3975},
url = {https://doi.org/10.1016/j.tcs.2023.114127},
doi = {10.1016/j.tcs.2023.114127},
journal = {Theor. Comput. Sci.},
month = oct,
numpages = {18}
}

@inproceedings{cuturi_sinkhorn_2013,
 author = {Cuturi, Marco},
 booktitle = {Advances in Neural Information Processing Systems},
 editor = {C.J. Burges and L. Bottou and M. Welling and Z. Ghahramani and K. Weinberger},
 pages = {},
 publisher = {Curran Associates, Inc.},
 title = {Sinkhorn Distances: Lightspeed Computation of Optimal Transport},
 url = {https://proceedings.neurips.cc/paper_files/paper/2013/file/af21d0c97db2e27e13572cbf59eb343d-Paper.pdf},
 volume = {26},
 year = {2013}
}

@article{ollivier_ricci_2009,
title = {{Ricci} curvature of {Markov} chains on metric spaces},
journal = {Journal of Functional Analysis},
volume = {256},
number = {3},
pages = {810-864},
year = {2009},
issn = {0022-1236},
doi = {https://doi.org/10.1016/j.jfa.2008.11.001},
url = {https://www.sciencedirect.com/science/article/pii/S002212360800493X},
author = {Yann Ollivier}
}

@article{forman_method_2003,
author={{Forman}},
title={Bochner's Method for Cell Complexes and Combinatorial {Ricci} Curvature},
journal={Discrete {\&} Computational Geometry},
year={2003},
month={Feb},
day={01},
volume={29},
number={3},
pages={323-374},
issn={1432-0444},
doi={10.1007/s00454-002-0743-x},
url={https://doi.org/10.1007/s00454-002-0743-x}
}

@inproceedings{rozemberczki_characteristic_2020,
author = {Rozemberczki, Benedek and Sarkar, Rik},
title = {Characteristic Functions on Graphs: Birds of a Feather, from Statistical Descriptors to Parametric Models},
year = {2020},
isbn = {9781450368599},
publisher = {Association for Computing Machinery},
address = {New York, NY, USA},
url = {https://doi.org/10.1145/3340531.3411866},
doi = {10.1145/3340531.3411866},
booktitle = {Proceedings of the 29th ACM International Conference on Information \& Knowledge Management},
pages = {1325–1334},
numpages = {10},
location = {Virtual Event, Ireland},
series = {CIKM '20}
}

@article{rozemberczki_multiscale_2021,
  author={Rozemberczki, Benedek and Allen, Carl and Sarkar, Rik and Thilo Gross, xx},
  journal={Journal of Complex Networks}, 
  title={Multi-Scale attributed node embedding}, 
  year={2021},
  volume={9},
  number={1},
  pages={1-22},
  doi={10.1093/comnet/cnab014},
  ISSN={2051-1329},
  month={April},
  }

@book{lovasz_matching_1986,
  title     = {Matching Theory},
  author    = {Lov{\'a}sz, L{\'a}szl{\'o} and Plummer, Michael D.},
  series    = {Annals of Discrete Mathematics},
  volume    = {29},
  publisher = {North-Holland},
  year      = {1986},
  address   = {Amsterdam},
  isbn      = {978-0-444-87916-5}
}

@article{leskovec_community_2009,
author = {Jure Leskovec and Kevin J. Lang and Anirban Dasgupta and Michael W. Mahoney},
title = {{Community Structure in Large Networks: Natural Cluster Sizes and the Absence of Large Well-Defined Clusters}},
volume = {6},
journal = {Internet Mathematics},
number = {1},
publisher = {A K Peters, Ltd.},
pages = {29 -- 123},
year = {2009},
}

@article{leskovec_graph_2007,
author = {Leskovec, Jure and Kleinberg, Jon and Faloutsos, Christos},
title = {Graph evolution: Densification and shrinking diameters},
year = {2007},
issue_date = {March 2007},
publisher = {Association for Computing Machinery},
address = {New York, NY, USA},
volume = {1},
number = {1},
issn = {1556-4681},
url = {https://doi.org/10.1145/1217299.1217301},
doi = {10.1145/1217299.1217301},
journal = {ACM Trans. Knowl. Discov. Data},
month = mar,
pages = {2–es},
numpages = {41}
}

@inproceedings{cho_friendship_2011,
author = {Cho, Eunjoon and Myers, Seth A. and Leskovec, Jure},
title = {Friendship and mobility: user movement in location-based social networks},
year = {2011},
isbn = {9781450308137},
publisher = {Association for Computing Machinery},
address = {New York, NY, USA},
url = {https://doi.org/10.1145/2020408.2020579},
doi = {10.1145/2020408.2020579},
booktitle = {Proceedings of the 17th ACM SIGKDD International Conference on Knowledge Discovery and Data Mining},
pages = {1082–1090},
numpages = {9},
location = {San Diego, California, USA},
series = {KDD '11}
}

@inproceedings{richardson_trust_2003,
author = {Richardson, Matthew and Agrawal, Rakesh and Domingos, Pedro},
title = {Trust management for the semantic web},
year = {2003},
isbn = {9783540203629},
publisher = {Springer-Verlag},
address = {Berlin, Heidelberg},
url = {https://doi.org/10.1007/978-3-540-39718-2_23},
doi = {10.1007/978-3-540-39718-2_23},
booktitle = {Proceedings of the Second International Conference on Semantic Web Conference},
pages = {351–368},
numpages = {18},
location = {Sanibel Island, FL},
series = {LNCS-ISWC'03}
}

@inproceedings{yang_defining_2012,
author = {Yang, Jaewon and Leskovec, Jure},
title = {Defining and evaluating network communities based on ground-truth},
year = {2012},
isbn = {9781450315463},
publisher = {Association for Computing Machinery},
address = {New York, NY, USA},
url = {https://doi.org/10.1145/2350190.2350193},
doi = {10.1145/2350190.2350193},
booktitle = {Proceedings of the ACM SIGKDD Workshop on Mining Data Semantics},
articleno = {3},
numpages = {8},
location = {Beijing, China},
series = {MDS '12}
}

@inproceedings{shchur_pitfalls_2018,
      title={Pitfalls of Graph Neural Network Evaluation}, 
      author={Oleksandr Shchur and Maximilian Mumme and Aleksandar Bojchevski and Stephan Günnemann},
      year={2019},
      booktitle={NeurIPS Workshop on Relational Representation Learning},
      eprint={1811.05868},
      archivePrefix={arXiv},
      primaryClass={cs.LG},
      url={https://arxiv.org/abs/1811.05868}, 
}

@misc{leskovec_snap_2014,
author = {Sosic, Rok and Leskovec, Jure},
title = {Large Scale Network Analytics with {SNAP}: Tutorial at the World Wide Web 2015 Conference},
year = {2015},
isbn = {9781450334730},
publisher = {Association for Computing Machinery},
address = {New York, NY, USA},
url = {https://doi.org/10.1145/2740908.2744708},
doi = {10.1145/2740908.2744708},
booktitle = {Proceedings of the 24th International Conference on World Wide Web},
pages = {1537–1538},
numpages = {2},
location = {Florence, Italy},
series = {WWW '15 Companion}
}

@inproceedings{kipf_semi_2017,
  title     = {Semi-Supervised Classification with Graph Convolutional Networks},
  author    = {Kipf, Thomas N. and Welling, Max},
  booktitle = {International Conference on Learning Representations},
  year      = {2017},
  url       = {https://openreview.net/forum?id=SJU4ayYgl}
}

@unpublished{micaletto_exact_2026,
  author = {Micaletto, Giorgio},
  title  = {Exact {Ollivier--Ricci} Curvature with One Minimum Cut per Edge},
  note   = {Software: \url{https://github.com/GiorgioMB/orcut}},
  year   = {2026}
}

@article{azarhooshang_review_2020,
AUTHOR = {Azarhooshang, Nazanin and Sengupta, Prithviraj and DasGupta, Bhaskar},
TITLE = {A Review of and Some Results for {Ollivier}--{Ricci} Network Curvature},
JOURNAL = {Mathematics},
VOLUME = {8},
YEAR = {2020},
NUMBER = {9},
ARTICLE-NUMBER = {1416},
URL = {https://www.mdpi.com/2227-7390/8/9/1416},
ISSN = {2227-7390},
DOI = {10.3390/math8091416}
}
}

\end{document}